\documentclass{article}

\parskip \medskipamount
\usepackage[T1]{fontenc}
\usepackage{geometry}
\usepackage{titling}
\usepackage{setspace}
\usepackage{xcolor}
\usepackage{authblk}
\usepackage{amsmath,amssymb,amsthm,amsfonts,bm,bbm}
\DeclareMathOperator*{\esssup}{ess\,sup}
\DeclareMathOperator*{\essinf}{ess\,inf}
\DeclareMathOperator*{\argmax}{arg\,max}
\DeclareMathOperator*{\argmin}{arg\,min}
\DeclareMathOperator*{\Var}{Var}
\DeclareMathOperator*{\Cov}{Cov}

\usepackage{mathrsfs,stmaryrd}
\usepackage{upgreek}
\usepackage{stmaryrd}
\SetSymbolFont{stmry}{bold}{U}{stmry}{m}{n}
\usepackage{indentfirst}
\usepackage{subdepth}
\usepackage{graphicx}
\usepackage{diagbox}
\usepackage{pdfpages}
\usepackage[latin1]{inputenc}
\usepackage{url}
\usepackage{float}
\usepackage{appendix}
\usepackage{etoolbox}
\usepackage{microtype}
\usepackage[centerlast]{caption}
\usepackage{subcaption}
\usepackage[plainpages=false,hyperfootnotes=false]{hyperref}
\hypersetup{colorlinks=true,urlcolor=blue,linkcolor=red,citecolor=blue}
\usepackage{cleveref}
\crefname{equation}{\hspace{-0.4em}}{\hspace{-0.4em}}
\crefrangeformat{equation}{(#3#1#4)--(#5#2#6)}
\usepackage{bookmark}
\usepackage{framed}
\usepackage{enumitem}
\usepackage{apacite}
\usepackage{booktabs}
\usepackage{array}
\usepackage[normalem]{ulem}

\newtheorem{theorem}{Theorem}[section]
\newtheorem{lemma}[theorem]{Lemma}
\newtheorem{proposition}[theorem]{Proposition}
\newtheorem{remark}[theorem]{Remark}
\newtheorem{definition}[theorem]{Definition}
\newtheorem{assumption}[theorem]{Assumption}

\newtheorem{example}[theorem]{Example}
\renewenvironment{proof}{\noindent {\bf Proof.}}{\hfill $\Box$}
\allowdisplaybreaks[4]

\title{Strictly monotone mean-variance preferences with applications to portfolio selection}
\author[1]{Yike Wang}
\author[2]{Yusha Chen}
\author[3]{Jingzhen Liu}
\author[4]{Zhenyu Cui}
\affil[1]{School of Finance, Chongqing Technology and Business University, Chongqing 400067, China.}
\affil[2]{School of Finance, Southwestern University of Finance and Economics, Chengdu 611130, China.}
\affil[3]{China Institute of Actuarial Science, Central University of Finance and Economics, Beijing 100081, China.}
\affil[4]{School of Business, Stevens Institute of Technology, Hoboken, NJ 07030, USA}

\begin{document}
\maketitle

\begin{abstract}
The monotone mean-variance (MMV) preference proposed by Maccheroni, et al. (Math. Finance 19(3): 487-521, 2009) fails to differentiate strictly dominant payoffs, which may cause inconsistency in portfolio decision-making.
This paper introduces a broader class of strictly monotone mean-variance (SMMV) preferences and demonstrates its applications to portfolio selection problems.
For the single-period portfolio problem under the SMMV preference, we derive the gradient condition for the optimal strategy, and investigate its association with the optimal mean-variance (MV) static strategy. 
We reduce the problem to solving a set of linear equations by analyzing the saddle point of some minimax problem. 
And results show that the optimal SMMV, MMV and MV strategies differ significantly in the single-period problem.
Furthermore, we conduct numerical experiments and compare our results with those of Maccheroni, et al. (Math. Finance 19(3): 487-521, 2009). 
The findings indicate that our SMMV preferences provide a more rational basis for assessing given prospects.
For the continuous-time portfolio problem under the SMMV preference, we consider continuous price processes with random coefficients, 
and establish a novel approach based on a general convex duality analysis to derive the optimal strategy.
Interestingly, we find that the optimal strategies for SMMV, MMV and MV preferences coincide under a certain condition, and provide a classical microeconomic interpretation for this condition.
We also characterize the optimal SMMV portfolio strategies relying on stochastic control techniques to facilitate potential extensions and refinements in future research.
\end{abstract}

\noindent {\bf Keywords:} monotone mean-variance preference, portfolio selection, Fenchel conjugate, stochastic Hamilton-Jacobi-Bellman-Isaacs equation, martingale convex duality method
\vspace{3mm}

\noindent {\bf AMS2010 classification:} Primary: 91G10, 49N10; Secondary: 91B05, 49N90
\vspace{3mm}


\section{Introduction}


Modern mean-variance (MV) analysis pioneered by \cite{Markowitz-1952} has been widely applied across various areas of mathematical finance theory and practice for decades due to their analytical tractability and intuitive appeal.
However, MV preferences suffer from a critical flaw---the lack of monotonicity. 
An investor may prefer a dominated payoff simply because the increase in variance outweighs the increase in mean, violating basic rationality principles. 
To illustrate this drawback, \cite{Maccheroni-Marinacci-Rustichini-Taboga-2009} gives a simple example, see \Cref{tab: example given by Maccheroni et al}. 
That is, although $g$ is evidently the preferable choice for any rational agent, as it clearly offers a statewise dominant payoff compared with $f$; an agent with MV preferences will choose the prospect $f$ instead of $g$. 
To address the issue caused by the lack of monotonicity, \cite{Maccheroni-Marinacci-Rustichini-Taboga-2009} proposes monotone mean-variance (MMV) preferences, which modify MV preferences via variational representation. 
As a result, using the MMV preference to evaluate a random variable $f$ is equivalent to using the corresponding MV preference to evaluate some truncated random variable $f \wedge {\lambda }_{f}$. 
And in recent years, MMV preferences have attracted considerable attention and have been used as the objective functional in dynamic portfolio selection, e.g., 
\cite{Trybula-Zawisza-2019,Cerny-2020,Strub-Li-2020,Shen-Zou-2022,Li-Guo-Tian-2023,Li-Guo-2021,Li-Liang-Pang-2022,Li-Liang-Pang-2023,Hu-Shi-Xu-2023,Cerny-Ruf-Schweizer-2025}.

\begin{table}[htbp!]
    \centering
    \begin{tabular}{lcccc}
    \hline
    State of nature & ${s}_{1}$ & ${s}_{2}$ & ${s}_{3}$ & ${s}_{4}$ \\
    \hline
    Probabilities & $0.25$ & $0.25$ & $0.25$ & $0.25$ \\
    Payoff of $f$ & $1$ & $2$ & $3$ & $4$ \\
    Payoff of $g$ & $1$ & $2$ & $3$ & $5$ \\
    \hline
    \end{tabular}
    \caption{Example given by \protect\cite{Maccheroni-Marinacci-Rustichini-Taboga-2009} with $\mathbb{E} [f] - \Var [f] > \mathbb{E} [g] - \Var [g]$.}
    \label{tab: example given by Maccheroni et al}
\end{table}

While an MMV preference restores monotonicity in a minimal sense, it still fails to distinguish between strictly dominant payoffs. 
Consequently, the MMV preference remains non-strictly monotone: two distinct prospects can be regarded as equally desirable even when one strictly dominates the other in every state. This subtle deficiency, though often overlooked, 
undermines the consistency of decision-making in portfolio selection problem.
Let us return to the aforementioned example given by \cite{Maccheroni-Marinacci-Rustichini-Taboga-2009}; see also \Cref{tab: example given by Maccheroni et al}.
Applying \cite[Theorem B.1]{Maccheroni-Marinacci-Rustichini-Taboga-2009} that gives the truncation level $( {\lambda }_{f}, {\lambda }_{g} )$ corresponding to $( f, g )$, 
one can find that the selection problem is reduced to choosing between $f \wedge 2.5$ and $g \wedge 2.5$ with a MV preference.
Obviously, $f \wedge 2.5 = g \wedge 2.5$ for every state of nature, and agents with the given MMV preferences could freely choose either of the two.   
However, if there are plenty of rational agents faced with this problem, or a rational agent is supposed to address this problem many times, the consistent choice will always be $g$! 
In other words, MMV preferences do not stand up repeated tests.
Building upon our example in \Cref{exp: non-monotonicity of MV}, we demonstrate in \Cref{exp: MMV is not strictly monotonic} that outcomes exhibiting statewise dominance may still be indistinguishable to agents with MMV preferences. 

To deal with this concern related to MMV preferences, we propose an analytically manageable modification that extends them to a class of strictly monotone mean-variance (SMMV) preferences, drawing inspiration from the celebrated envelope theorem. 
This modification includes the MMV preferences as a special case. 
It preserves tractability while reconciling the MV framework with the axioms of rational choice.
As shown in \Cref{thm: strict monotonicity of SMMV}, our SMMV preferences yield significantly different evaluations for any pair of non-identical prospects where one statewise dominates the other. 
In the example of \cite{Maccheroni-Marinacci-Rustichini-Taboga-2009} (see also \Cref{tab: example given by Maccheroni et al}), our SMMV preferences indicate that $g$ is clearly favored over $f$ (see \Cref{exp: SMMV}).
Moreover, in the continuous-time portfolio problem with our SMMV preferences, the optimal dynamic portfolio strategies for SMMV, MMV and MV are not necessarily the same, 
even when the underlying dynamic model is conventional and simple (i.e., without discontinuities).
The consistency among the optimal SMMV, MMV and MV portfolio strategies depends on a random variable, denoted by $\zeta$, which parameterizes the degree of strict monotonicity in the preference.

The main contributions of this paper are as follows.
\emph{Firstly}, we propose a class of preferences, formulated through a variational representation, and name them ``SMMV preferences'' 
(see \cref{eq: definition of SMMV :eq} for the definition and \Cref{thm: strict monotonicity of SMMV} for the strict monotonicity). 
And then we analyze their properties to facilitate comparison with \cite{Maccheroni-Marinacci-Rustichini-Taboga-2009} and establish several equivalent expressions.
In particular, a SMMV preference can be represented in terms of a truncated quadratic utility functional (see \Cref{thm: Gateaux differentiability and equivalent expressions of SMMV}), 
or alternatively, as the minimum/maximum of some class of MV utility functionals (see \Cref{thm: minimality and maximality of SMMV}).
Notably, this maximality implies that our SMMV preferences also extend the alternative construction of MMV preference reported in \cite{Cerny-2020}, where the monotone hull of MV preference is employed.
We provide a commutative diagram analogous to \cite[eqn.(7)]{Cerny-2020} to graphically show this extension (see \Cref{rem: Supremal convolution characterization of SMMV}).

\emph{Secondly}, we demonstrate that investment decisions under SMMV preferences can differ meaningfully from those under MMV and MV preferences in the single-period static portfolio problem. 
The discrepancy between them can be succinctly characterized by the Lagrange multiplier.
Moreover, we reduce the SMMV portfolio problem to solving a set of linear equations, 
offering a significantly simpler algorithm compared to the approach that considers the first-order derivative conditions used in \cite{Maccheroni-Marinacci-Rustichini-Taboga-2009}.
Leveraging this simplification, we carry out some numerical experiments to illustrate the single-period static portfolio selection under SMMV preferences.

\emph{Thirdly}, in continuous times, we apply a general convex duality framework (see \Cref{app: Convex duality analysis to the MV problem}) to derive the optimal strategy,
and find that the optimal strategies for SMMV and MV preferences coincide under \Cref{ass: SMMV-MV comparison}.
This reproduces and generalizes the identity of the optimal MMV and MV preferences as shown in, e.g., \cite{Trybula-Zawisza-2019,Strub-Li-2020}.
To facilitate comparison with existing studies, we also use classical methods to characterize the optimal SMMV portfolio strategies, 
including the dynamic programming for differential games, the martingale convex duality analysis for a minimax relationship, and the embedding method (proposed by \cite{Li-Ng-2000}) for multi-stage optimization.

The rest of this paper is organized as follows.
In \Cref{sec: Strictly monotone mean-variance preference}, we introduce the definition and key properties of SMMV preferences.  
To showcase their application, we study the single-period portfolio selection problem under SMMV preferences in \Cref{sec: Static portfolio selection}, 
followed by the continuous-time portfolio selection problem in \Cref{sec: Continuous-time dynamic portfolio management}.
Finally, \Cref{sec: Concluding remark} presents a brief concluding remark. 
A convex duality analysis to the MV problem is given in \Cref{app: Convex duality analysis to the MV problem}. 
And the proofs of lemmas, theorems and propositions for this work are collected in \Cref{app: Proof}.

\section{Strictly monotone mean-variance preferences}
\label{sec: Strictly monotone mean-variance preference}

Let $( \Omega, \mathcal{F}, \mathbb{P} )$ be a complete probability space, where $\mathcal{F}$ is a non-trivial $\sigma $-algebra (i.e. $\{ \varnothing, \Omega \} \subsetneq \mathcal{F}$).
$\mathbb{E} [ \cdot ]$ and $\Var [ \cdot ]$ represent the expectation and variance operator under $\mathbb{P}$, respectively.
And $\mathbb{L}^{2} ( \Omega )$ is the Banach space of all $\mathcal{F}$-measurable square-integrable random variables on $( \Omega, \mathcal{F}, \mathbb{P} )$,
equipped with the norm $\| \cdot \|_{ \mathbb{L}^{2} ( \Omega ) } := ( \mathbb{E} [ | \cdot |^{2} ] )^{ \frac{1}{2} }$.
In addition, we define the subsets $\mathbb{L}^{2}_{+} ( \Omega ) = \{ f \in \mathbb{L}^{2} ( \Omega ): f \ge 0, ~ \mathbb{P}-a.s. \}$, 
$\mathbb{L}^{2}_{ \zeta + } ( \Omega ) := \{ f \in \mathbb{L}^{2} ( \Omega ): f \ge \zeta, ~ \mathbb{P}-a.s. \}$ for $\zeta \in \mathbb{L}^{2}_{+} ( \Omega )$, and
\begin{equation*}
\mathbb{L}^{2}_{0} ( \Omega ) = \{ f \in \mathbb{L}^{2} ( \Omega ): 0 \le f \le 1, ~ \mathbb{P}-a.s.,~ \mathbb{E} [f] < 1 \}.
\end{equation*}

\subsection{MMV preference revisited: lack of strict monotonicity}
\label{subsec: MMV preference revisited}

For the conventional MV objective functional ${U}_{\theta }: \mathbb{L}^{2} ( \Omega ) \to \mathbb{R}$ given by ${U}_{\theta } (f) := \mathbb{E} [f] - \frac{\theta }{2} \Var [f]$,
where the preassigned constant $\theta > 0$ stands for the risk aversion to variance, we have the following estimate:
\begin{equation*}
\big| {U}_{\theta } ( f + g ) - {U}_{\theta } (f) - \mathbb{E} \big[ \big( 1 - \theta ( f - \mathbb{E} [f] ) \big) g \big] \big| = \frac{\theta }{2} \Var [g] \le o ( \| g \|_{ \mathbb{L}^{2} ( \Omega ) } ),
\quad \forall f,g \in \mathbb{L}^{2} ( \Omega ),
\end{equation*} 
which implies that ${U}_{\theta }$ is Fr\'echet differentiable and hence G\^ateaux differentiable.
Hereafter, with a slight abuse of notation, we denote by $d {U}_{\theta } (f) := 1 - \theta ( f - \mathbb{E} [f] )$ the G\^ateaux derivative of $d {U}_{\theta }$ at $f$, given that
\begin{equation*}
d {U}_{\theta } (f) (g) := \lim_{ \varepsilon \downarrow 0 } \frac{ {U}_{\theta } ( f + \varepsilon g ) - {U}_{\theta } (f) }{\varepsilon } = \mathbb{E} \big[ \big( 1 - \theta ( f - \mathbb{E} [f] ) \big) g \big]
\end{equation*}
provides a continuous linear functional $d {U}_{\theta } (f) ( \cdot )$ on $\mathbb{L}^{2} ( \Omega )$.
Moreover, ${U}_{\theta }$ is G\^ateaux differentiable at $f$ if and only if the superdifferential 
$\partial {U}_{\theta } (f) := \{ Y \in \mathbb{L}^{2} ( \Omega ): {U}_{\theta } (g) \le {U}_{\theta } (f) + \mathbb{E} [ Y ( g-f ) ], ~ \forall g \in \mathbb{L}^{2} ( \Omega ) \}$ is a singleton,
according to \cite[Proposition 1.8, p. 5]{Phelps-1993}.
Then, $\partial {U}_{\theta } (f) = \{ d {U}_{\theta } (f) \}$ follows from the concavity of ${U}_{\theta }$.

In \cite{Maccheroni-Marinacci-Rustichini-Taboga-2009}, the convex closed subset of $\mathbb{L}^{2} ( \Omega )$ denoted by
\begin{align*}
\mathcal{G}_{\theta } & := \{ f \in \mathbb{L}^{2} ( \Omega ): \partial {U}_{\theta } (f) \cap \mathbb{L}^{2}_{+} ( \Omega ) \ne \varnothing \} \\
                      &  = \{ f \in \mathbb{L}^{2} ( \Omega ): d {U}_{\theta } (f) \in \mathbb{L}^{2}_{+} ( \Omega ) \} \\
                      &  = \Big\{ f \in \mathbb{L}^{2} ( \Omega ): f - \mathbb{E} [f] \le \frac{1}{\theta }, \ \mathbb{P}-a.s. \Big\}
\end{align*}
is named ``the domain of monotonicity'' of ${U}_{\theta }$.
The word ``monotonicity'' arises from the phenomenon that ${U}_{\theta } (f) \le {U}_{\theta } (g)$ for any $f,g \in \mathcal{G}_{\theta }$ with $f \le g$, $\mathbb{P}$-a.s.
In comparison, for any $f \notin \mathcal{G}_{\theta }$, there exists $g \in \mathbb{L}^{2} ( \Omega )$ that is $\varepsilon $-close to $f$ such that $g > f$ but ${U}_{\theta } (f) > {U}_{\theta } (g)$;
see \cite[Lemma 2.1]{Maccheroni-Marinacci-Rustichini-Taboga-2009}.
This exactly shows the drawback of conventional MV objective functionals. In particular, the conventional MV preferences contradict the statewise dominance in such occasions.
In addition, the following outcomes of a single fair coin toss with 50/50 chance of heads or tails, can intuitively illustrate the limitation of MV preferences with respect to monotonicity.

\begin{example}\label{exp: non-monotonicity of MV}
We firstly consider $h$ with $\mathbb{E} [h] = 0$ on the boundary of $\mathcal{G}_{\theta }$.
A positive (resp. negative) perturbation added to $h$ produces $g \notin \mathcal{G}_{\theta }$ (resp. $f \in \mathcal{G}_{\theta }$).
Obviously, for any $\varepsilon > 0$, one obtains ${U}_{\theta } (g) < {U}_{\theta } (f)$. However, $g \notin \mathcal{G}_{\theta }$ statewise dominates $f \in \mathcal{G}_{\theta }$, indicating that the MV criterion is irrational in this case.

\begin{table}[htbp!]
    \centering
    \renewcommand{\arraystretch}{1.3}
    \begin{tabular}{ccccc}
    \hline
             & \multicolumn{2}{c}{Payoffs} & Positive deviation & MV functional values \\
               \cmidrule(r){2-3} 
    Prospect & Heads & Tails & $\esssup \{ \cdot \} - \mathbb{E} [ \cdot ]$ 
    & ${U}_{\theta } ( \cdot ) = \mathbb{E} [\cdot] - \frac{\theta }{2} \Var [\cdot]$ \\
    \hline
    $h$ & $\frac{1}{\theta }$ & $- \frac{1}{\theta }$ 
    & $\frac{1}{\theta } $ 
    & $- \frac{1}{ 2 \theta }$ \\
    $g$ & $\frac{1}{\theta } + \varepsilon $ & $ - \frac{1}{\theta }$ 
    & $\frac{1}{\theta } + \frac{\varepsilon }{2}$ 
    & $- \frac{1}{ 2 \theta } - \frac{\theta }{8} {\varepsilon }^{2}$ \\
    $f$ & $\frac{1}{\theta } - \frac{\theta }{9} {\varepsilon }^{2}$ & $- \frac{1}{\theta } - \frac{\theta }{9} {\varepsilon }^{2}$ 
    & $\frac{1}{\theta }$ 
    & $- \frac{1}{ 2 \theta } - \frac{\theta }{9} {\varepsilon }^{2}$ \\
    \hline
    \end{tabular}
    \caption{The outcomes of a single fair coin toss (50/50 heads or tails).}
    \label{tab: strictly statewise dominance violates general MV preference}
\end{table} 
\end{example}

To tackle the monotonicity problem of MV preferences, the Fenchel conjugate stands out as a key tool.
Let us proceed with the Fenchel conjugate of ${U}_{\theta }$, defined as follows:
\begin{equation}\label{eq: Fenchel conjugate of MV :eq}
{U}_{\theta }^{*} (Y) := \inf_{ f \in \mathbb{L}^{2} ( \Omega ) } \{ \mathbb{E} [ Y f ] - {U}_{\theta } (f) \}, \quad Y \in \mathbb{L}^{2} ( \Omega ).
\end{equation}
If there exists $c \in \mathbb{R}$ such that $f = c$, $\mathbb{P}$-a.s., then $\mathbb{E} [ Y f ] - {U}_{\theta } (f) = c ( \mathbb{E} [Y] - 1 )$, and hence
\begin{equation*}
{U}_{\theta }^{*} (Y) \le \inf_{ c \in \mathbb{R} } c ( \mathbb{E} [Y] - 1 )
= \left\{ \begin{aligned} & 0, && if \ \mathbb{E} [Y] = 1; \\
                          & - \infty, && otherwise. \end{aligned} \right.
\end{equation*}
This implies that ${U}_{\theta }^{*} (Y) = - \infty $ if $\mathbb{E} [Y] \ne 1$.
And if $\mathbb{E} [Y] = 1$, then the minimizer $\hat{f}$ for the right-hand side of \cref{eq: Fenchel conjugate of MV :eq} fulfills the G\^ateaux derivative optimality condition 
$Y = d {U}_{\theta } ( \hat{f} ) \equiv 1 - \theta ( \hat{f} - \mathbb{E} [ \hat{f} ] )$, $\mathbb{P}$-a.s.,
implying 
\begin{equation}\label{eq: Fenchel conjugate of MV on effective domain :eq}
  {U}_{\theta }^{*} (Y) 
= \mathbb{E} [ Y \hat{f} ] - \mathbb{E} [ \hat{f} ] + \frac{\theta }{2} \mathbb{E} \big[ ( \hat{f} - \mathbb{E} [ \hat{f} ] )^{2} \big]
= - \frac{1}{2 \theta } ( \mathbb{E} [ {Y}^{2} ] - 1 ), ~ \forall Y \in \mathbb{L}^{2} ( \Omega ), ~ \mathbb{E} [Y] = 1.
\end{equation} 
Furthermore, due to Fenchel-Moreau theorem (for Hilbert spaces, see \cite[Theorem 13.37]{Bauschke-Combettes-2017}), 
the concave functional ${U}_{\theta }$ must be equal to the Fenchel conjugate of ${U}_{\theta }^{*}$,
which is also known as the variational representation of MV preference (see \cite{Maccheroni-Marinacci-Rustichini-2006}). 
That is,
\begin{equation}\label{eq: biconjugate form of MV :eq}
{U}_{\theta } (f) = \inf_{ Y \in \mathbb{L}^{2} ( \Omega ) } \{ \mathbb{E} [ Y f ] - {U}_{\theta }^{*} (Y) \}
             \equiv \inf_{ Y \in \mathbb{L}^{2} ( \Omega ), \mathbb{E} [Y] = 1 } \Big\{ \mathbb{E} [ Y f ] + \frac{1}{2 \theta } ( \mathbb{E} [ {Y}^{2} ] - 1 ) \Big\};
\end{equation}
while the envelope theorem (see, e.g., \cite{Milgrom-Segal-2002}) implies that $d {U}_{\theta } (f)$ realizes the minimum.

In \cite{Maccheroni-Marinacci-Rustichini-Taboga-2009}, a minor modification to the constraint in \cref{eq: biconjugate form of MV :eq} leads to the MMV preference:
\begin{equation}\label{eq: biconjugate form of MMV :eq}
{V}_{\theta } (f) := \inf_{ Y \in \mathbb{L}^{2}_{+} ( \Omega ) } \{ \mathbb{E} [ Y f ] - {U}_{\theta }^{*} (Y) \}.
\end{equation}
As a point-wise infimum of some affine functionals of $f$, ${V}_{\theta }$ is concave.
Moreover, ${V}_{\theta }$ is G\^ateaux differentiable, and $d {V}_{\theta } (f)$ is the minimizer for the right-hand side of \cref{eq: biconjugate form of MMV :eq}.
Consequently,
\begin{equation*}
{V}_{\theta } (g) \le \mathbb{E} [ g d {V}_{\theta } (f) ] - {U}_{\theta }^{*} ( d {V}_{\theta } (f) ) = {V}_{\theta } (f) + \mathbb{E} [ ( g - f ) d {V}_{\theta } (f) ] \le {V}_{\theta } (f)
\end{equation*}
for any $f,g \in \mathbb{L}^{2} ( \Omega )$ with $g \le f$, $\mathbb{P}$-a.s., which demonstrates the monotonicity of ${V}_{\theta }$ over the entire domain $\mathbb{L}^{2} ( \Omega )$.
However, if $d {V}_{\theta } (f)$ vanishes on some $A \in \mathcal{F}$,
then ${V}_{\theta } ( f + \varepsilon {1}_{A} ) \le {V}_{\theta } (f) + \varepsilon \mathbb{E} [ {1}_{A} d {V}_{\theta } (f) ] = {V}_{\theta } (f)$ for any $\varepsilon > 0$, implying that ${V}_{\theta } ( f + \varepsilon {1}_{A} ) = {V}_{\theta } (f)$. 
Thus, ${V}_{\theta }$ is not strictly monotonic.

\begin{example}\label{exp: MMV is not strictly monotonic}
Let us use the MMV preference ${V}_{\theta }$ to evaluate the prospects $h \in \mathcal{G}_{\theta }$ and $g \notin \mathcal{G}_{\theta }$ in \Cref{exp: non-monotonicity of MV}.
According to \cite[Theorem 2.2]{Maccheroni-Marinacci-Rustichini-Taboga-2009}, one obtains ${V}_{\theta } (g) = {U}_{\theta } ( g \wedge \frac{1}{\theta } ) = {V}_{\theta } (h)$.
Thus, $( h,g )$ are indistinguishable under ${V}_{\theta }$.
However, $g$ statewise dominates $h$; that is, a rational choice should be $g$ rather than $h$.
\end{example}

\subsection{SMMV preferences}

To address the aforementioned limitation of MMV preferences---the lack of strict monotonicity---we propose a generalized class of preferences and refer to them as ``SMMV preferences''.

\begin{definition}[SMMV preferences]\label{def: definition of SMMV :eq}
Given a mapping $\zeta: \Omega \to \mathbb{R}$ with $\zeta \in \mathbb{L}^{2}_{0} ( \Omega )$ and $\mathbb{P} ( \zeta > 0 ) = 1$,
a SMMV preference corresponding to the MV and MMV preferences $( {U}_{\theta }, {V}_{\theta } )$ with $\theta \in \mathbb{R}_{+}$ is formulated by
\begin{equation}\label{eq: definition of SMMV :eq}
{V}_{ \theta, \zeta } (f) := \inf_{ Y \in \mathbb{L}^{2}_{ \zeta + } ( \Omega ) } \{ \mathbb{E} [ Y f ] - {U}_{\theta }^{*} (Y) \}, \quad \forall f \in \mathbb{L}^{2} ( \Omega )
\end{equation}
\end{definition}

To avoid misunderstanding, we emphasize again that it is a mapping $\zeta: \Omega \to \mathbb{R}$ that a decision maker is supposed to actively choose in advance, rather than passively taking a random outcome of $\zeta $ after the fact.
Just as the primitively chosen parameter $\theta$ in the MV preference ${U}_{\theta }$ (or $\gamma$ in the power utility function $U(x) = \frac{1}{1-\gamma} x^{1-\gamma}$) captures risk aversion, 
$\zeta$ plays a key role in evaluating the monotonicity requirement.
Intuitively, if an atom $\omega \in \Omega $ is regarded as a direction,
then $\zeta ( \omega ) \times \mathbb{P} ( \{ \omega \} )$ formulates a lower bound on the partial derivative $\partial {V}_{ \theta, \zeta } (f) / \partial f ( \omega )$, 
which acts as the marginal utility on the ``good/endowment'' $f ( \omega )$. 
See also the upcoming \Cref{exp: illustrate MV MMV SMMV}.

Furthermore, the upcoming \Cref{thm: strict monotonicity of SMMV} shows that ${V}_{ \theta, \zeta }$ defined in \Cref{def: definition of SMMV :eq} is strictly monotonic,
since ${V}_{ \theta, \zeta } (g) < {V}_{ \theta, \zeta } (f)$ for any $f,g \in \mathbb{L}^{2} ( \Omega )$ with $\| f - g \|_{ \mathbb{L}^{2} ( \Omega ) } > 0$ and $g \le f$.
It is $\zeta > 0$ $\mathbb{P}$-a.s. that contributes to distinguishing the performance of $f$ and $g$.  
In addition, setting $\zeta \equiv 0$ produces ${V}_{ \theta, 0 } = {V}_{\theta }$.
Therefore, unless otherwise stated, hereafter we take $\zeta \in \mathbb{L}^{2}_{0} ( \Omega )$ to include both the classical MMV preference and our SMMV preferences,
and we still refer to ${V}_{ \theta, \zeta }$ as the SMMV preference, in order to distinguish it from the classical MMV preference. 

\begin{theorem}\label{thm: strict monotonicity of SMMV}
${V}_{ \theta, \zeta } (g) \le {V}_{ \theta, \zeta } (f) - \mathbb{E} [ ( f - g ) \zeta ]$ for $\zeta \in \mathbb{L}^{2}_{0} ( \Omega )$ and any $f,g \in \mathbb{L}^{2} ( \Omega )$ with $g \le f$.
\end{theorem}

\begin{proof}
Given the quadratic functional \cref{eq: Fenchel conjugate of MV on effective domain :eq}, 
the minimum on the right-hand side of \cref{eq: definition of SMMV :eq} can be attained, and the minimizer (denoted by $\hat{Y}_{f}$) is unique.
The precise results can be found in \Cref{pf-thm: Gateaux differentiability and equivalent expressions of SMMV}.
Since $\hat{Y}_{f} \in \mathbb{L}_{ \zeta + }^{2} ( \Omega )$ and $\zeta \in \mathbb{L}^{2}_{0} ( \Omega )$, one obtains
\begin{equation*}
{V}_{ \theta, \zeta } (g) \le \mathbb{E} [ g \hat{Y}_{f} ] - {U}_{\theta }^{*} ( \hat{Y}_{f} ) = {V}_{ \theta, \zeta } (f) - \mathbb{E} [ ( f - g ) \hat{Y}_{f} ] \le {V}_{ \theta, \zeta } (f) - \mathbb{E} [ ( f - g ) \zeta ]
\end{equation*}
for any $f,g \in \mathbb{L}^{2} ( \Omega )$ with $g \le f$. Thus, the proof is completed.
\end{proof}

\begin{remark}[reasons for $\zeta \in \mathbb{L}^{2}_{0} ( \Omega )$]
If $\mathbb{E} [ \zeta ] = 1$, then $Y \in \mathbb{L}^{2}_{ \zeta + } ( \Omega )$ and ${U}_{\theta }^{*} (Y) \ne - \infty $ lead to $Y = \zeta $, $\mathbb{P}$-a.s.
This implies that ${V}_{ \theta, \zeta } (f) = \mathbb{E} [ f \zeta ] - {U}_{\theta }^{*} ( \zeta )$, which is an affine functional of $f$.
If $\mathbb{E} [ \zeta ] > 1$, then ${U}_{\theta }^{*} (Y) = - \infty $ for any $Y \in \mathbb{L}^{2}_{ \zeta + } ( \Omega )$, which leads to an improper ${V}_{ \theta, \zeta }$. 
These considerations motivate the assumption that $\mathbb{E} [ \zeta ] < 1$ in \cref{eq: definition of SMMV :eq}.
In addition, a straightforward calculation shows that 
for all constant $c \in \mathbb{R}$, we have ${V}_{ \theta, \zeta } (c) = c$ if and only if $1 \in \mathbb{L}^{2}_{ \zeta + } ( \Omega )$.
This is why we impose the condition $\zeta \le 1$, $\mathbb{P}$-a.s.
\end{remark}

\begin{example}\label{exp: SMMV}
For the example given by \Cref{tab: example given by Maccheroni et al}, once $\zeta > 0$ under the state ${s}_{4}$, 
\Cref{thm: strict monotonicity of SMMV} gives ${V}_{ \theta, \zeta } (f) \le {V}_{ \theta, \zeta } (g) - 0.25 \zeta ( {s}_{4} ) < {V}_{ \theta, \zeta } (g)$.
That is, under the SMMV preferences, $g$ is significantly favored over $f$.
In terms of \Cref{exp: non-monotonicity of MV}, we let $\zeta $ be a fixed constant in $( 0,1 )$, and obtain
\begin{equation*}
{V}_{ \theta, \zeta } (f) \le {V}_{ \theta, \zeta } (h) - \frac{\theta }{9} {\varepsilon }^{2} \zeta < {V}_{ \theta, \zeta } (h) \le {V}_{ \theta, \zeta } (g) - \frac{1}{2} \varepsilon \zeta < {V}_{ \theta, \zeta } (g).
\end{equation*}
In contrast to the indistinguishability of $( h,g )$ under the MMV preference (see \Cref{exp: MMV is not strictly monotonic}), $g$ is significantly favored over $h$ under the SMMV preferences.
\end{example}

\begin{example}\label{exp: illustrate MV MMV SMMV} 
To compare MV, MMV and SMMV preferences, we still consider the outcomes of a single fair coin toss with equal probabilities of heads and tails.
Suppose that $\theta = 2$, $\zeta \equiv \frac{1}{2}$, ${X}_{x} ( head ) = x$ and ${X}_{x} ( tail ) = 0$, implying that ${X}_{x}$ statewise dominates ${X}_{y}$ for any $x \ge y$. Then, we have
\begin{align*}
\text{MV: } 
& {U}_{2} ( {X}_{x} ) = \frac{1}{2} x - \frac{1}{4} {x}^{2}, \\
\text{MMV: } 
& {V}_{2} ( {X}_{x} ) = \begin{cases} x + \frac{1}{4} > {U}_{2} ( {X}_{x} ), & if ~ x \in ( - \infty, - 1 ); \\ 
                                    \frac{1}{2} x - \frac{1}{4} {x}^{2} = {U}_{2} ( {X}_{x} ), & if ~ x \in [ -1,1 ]; \\ 
                                    \frac{1}{4} = {U}_{2} ( {X}_{1} ) > {U}_{2} ( {X}_{x} ), & otherwise; \end{cases} \\
\text{and SMMV: }
& {V}_{ 2, \frac{1}{2} } ( {X}_{x} ) = \begin{cases} \frac{3}{4} x + \frac{1}{16} > {V}_{2} ( {X}_{x} ), & if ~ x \in ( - \infty, - \frac{1}{2} ); \\ 
                                                   \frac{1}{2} x - \frac{1}{4} {x}^{2} = {V}_{2} ( {X}_{x} ) = {U}_{2} ( {X}_{x} ), & if ~ x \in [ - \frac{1}{2}, \frac{1}{2} ]; \\ 
                                                   \frac{1}{4} x + \frac{1}{16} > {V}_{2} ( {X}_{x} ), & otherwise. \end{cases}
\end{align*}

As ${U}_{2} ( {X}_{x} )$ is strictly decreasing in $x$ on $( 1, + \infty )$, the lack of monotonicity in the MV preference (${U}_{2}$) becomes apparent.
Although the MMV preference (${V}_{2}$) fixes this monotonicity issue, it fails to differentiate among the ${X}_{x}$ for $x \in ( 1, + \infty )$, rendering them indistinguishable under ${V}_{2}$.
Our SMMV preference ${V}_{ 2, \frac{1}{2} }$ is not only strictly increasing in $x$ on $\mathbb{R}$, but also satisfies the inequality
${V}_{ 2, \frac{1}{2} } ( {X}_{x} ) - {V}_{ 2, \frac{1}{2} } ( {X}_{y} ) \ge \frac{1}{4} ( x - y ) = \mathbb{E} [ ( {X}_{x} - {X}_{y} ) \zeta ]$, as shown in \Cref{thm: strict monotonicity of SMMV}.
Thus, every unit increase in $x$ yields a marginal utility greater than $\zeta ( head ) \times \mathbb{P} ( head ) = \frac{1}{4}$. 
Following the setup of this example, \Cref{fig: illustrate MV MMV SMMV} gives an intuitive illustration of the MV, MMV and our proposed SMMV preferences.

\begin{figure}[H]
  \centering
  \includegraphics[width=8cm]{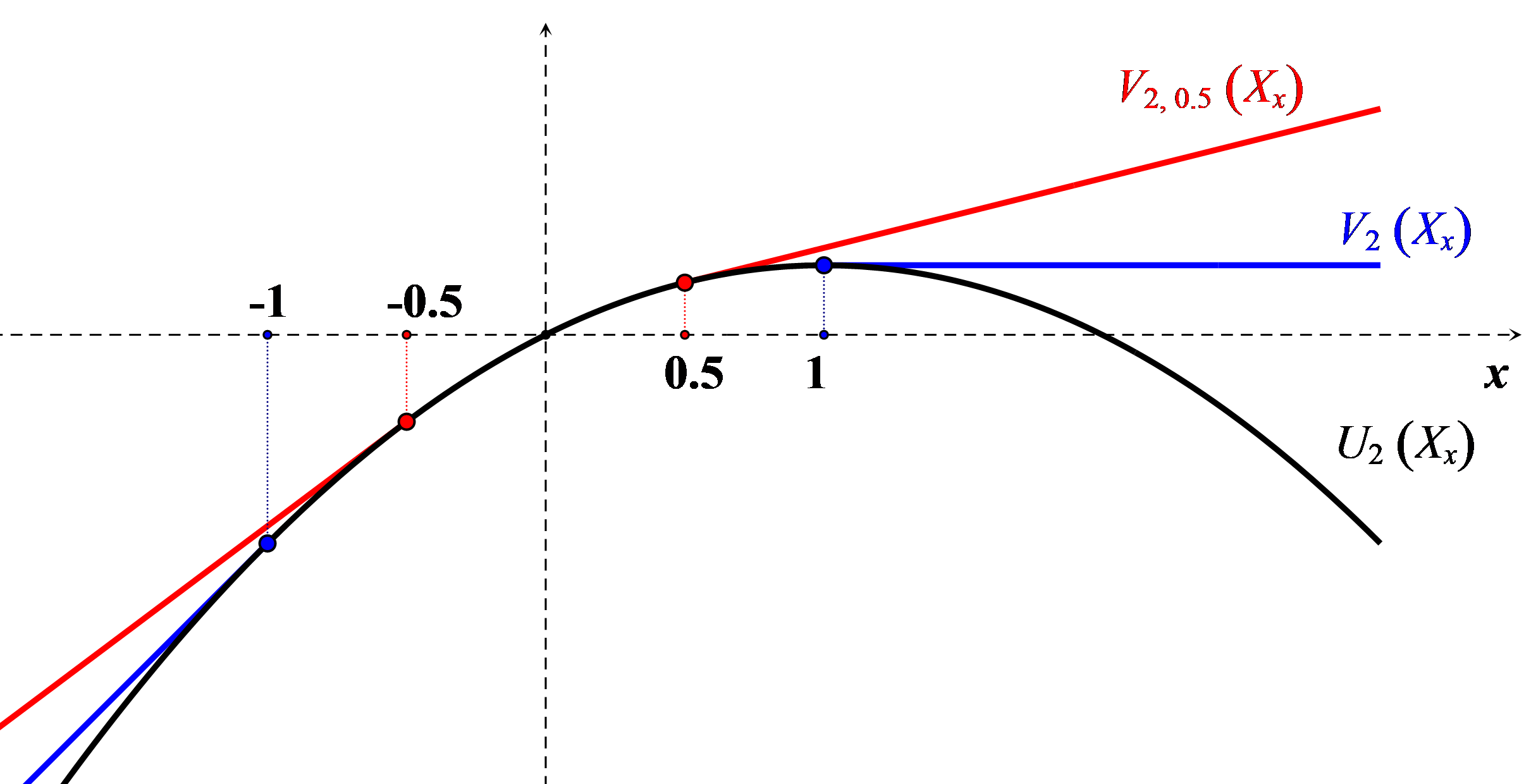}
  \caption{\centering Comparison of MV, MMV, and SMMV preferences: \\ ${U}_{2} ( {X}_{x} )$, ${V}_{2} ( {X}_{x} )$ and ${V}_{ 2, \frac{1}{2} } ( {X}_{x} )$ based on the setting described in \Cref{exp: illustrate MV MMV SMMV}.}
  \label{fig: illustrate MV MMV SMMV}
\end{figure}
\end{example}

\subsection{Properties and equivalent expressions of SMMV preference}
\label{subsec: Properties of SMMV preference}

First of all, according to \Cref{def: definition of SMMV :eq}, one can easily establish the following basic properties of the SMMV preferences:
\begin{itemize}
\item ${V}_{ \theta, \zeta }$ is bounded, due to ${U}_{\theta } (f) \le {V}_{\theta } (f) \le {V}_{ \theta, \zeta } (f) \le \frac{ \mathbb{E} [ \zeta f ] }{ \mathbb{E} [ \zeta ] } - {U}_{\theta }^{*} ( \frac{\zeta }{ \mathbb{E} [ \zeta ] } )$.
\item ${V}_{ \theta, \zeta }$ is concave, since it is a point-wise infimum of some affine functionals of $f$.
\item ${V}_{ \theta, \zeta } (f)$ is decreasing in $\theta $, since \cref{eq: Fenchel conjugate of MV on effective domain :eq} implies that ${U}_{\theta }^{*}$ is increasing in $\theta $.
\item ${V}_{ \theta, \zeta } (f)$ is increasing in $\zeta $; that is, ${V}_{ \theta, \hat{\zeta } } (f) \ge {V}_{ \theta, \zeta } (f)$ follows from $\mathbb{L}^{2}_{ \hat{\zeta } + } ( \Omega ) \subseteq \mathbb{L}^{2}_{ \zeta + } ( \Omega )$, 
      provided that $\hat{\zeta } \ge \zeta $, $\mathbb{P}$-a.s. 
\end{itemize} 
Then, to facilitate comparison among the SMMV, MMV and MV preferences and deriving the explicit expressions for ${V}_{ \theta, \zeta } (f)$,
we list some related truncation properties in the following \Cref{lem: identity and truncation}. 

\begin{lemma}\label{lem: identity and truncation}
Fix $( \theta, \zeta ) \in \mathbb{R}_{+} \times \mathbb{L}^{2}_{0} ( \Omega )$. The following properties hold.
\begin{itemize}
\item ${V}_{ \theta, \zeta } (f) = {U}_{\theta } (f)$, if and only if
      \begin{align*}
      f \in
      \mathcal{G}_{ \theta, \zeta } & := \{ f \in \mathbb{L}^{2} ( \Omega ): \partial {U}_{\theta } (f) \cap \mathbb{L}^{2}_{ \zeta + } ( \Omega ) \ne \varnothing \} \\
                                    &  = \Big\{ f \in \mathbb{L}^{2} ( \Omega ): f - \mathbb{E} [f] \le \frac{ 1 - \zeta }{\theta }, ~ \mathbb{P}-a.s. \Big\}
                               \subseteq \mathcal{G}_{\theta }.
      \end{align*}
\item $f \wedge ( \lambda - \frac{\zeta }{\theta } ) \in \mathcal{G}_{ \theta, \zeta }$ for any $f \in \mathbb{L}^{2} ( \Omega )$ and $\lambda \le {\lambda }_{ f, \theta, \zeta }$, 
      where ${\lambda }_{ f, \theta, \zeta } \in ( \essinf \{ f + \frac{\zeta }{\theta } \}, \mathbb{E} [f] + \frac{1}{\theta } ]$ uniquely fulfills
      \begin{equation}\label{eq: lambda-equation :eq}
        \frac{ 1 - \mathbb{E} [ \zeta ] }{\theta } 
      = \int_{ - \infty }^{ {\lambda }_{ f, \theta, \zeta } } \mathbb{P} \Big( f + \frac{\zeta }{\theta } \le s \Big) ds
      \equiv {\lambda }_{ f, \theta, \zeta } - \mathbb{E} \Big[ \Big( f + \frac{\zeta }{\theta } \Big) \wedge {\lambda }_{ f, \theta, \zeta } \Big].
      \end{equation}
\item $f \in \mathcal{G}_{ \theta, \zeta }$, if and only if $f \in \mathbb{L}^{2} ( \Omega )$ and $f + \frac{\zeta }{\theta } \le {\lambda }_{ f, \theta, \zeta }$, $\mathbb{P}$-a.s.
      In other words, 
      \begin{equation*}
      \mathcal{G}_{ \theta, \zeta } = \bigg\{ f \in \mathbb{L}^{2} ( \Omega ): \esssup \Big\{ f + \frac{\zeta }{\theta } \Big\} \le {\lambda }_{ f, \theta, \zeta } = \mathbb{E} [f] + \frac{1}{\theta } \bigg\}.
      \end{equation*}
\item $f \in \mathcal{G}_{ \theta, \zeta }$, 
      if (resp. only if) $f \in \mathbb{L}^{2} ( \Omega )$ and $f \wedge ( \lambda - \frac{\zeta }{\theta } ) \in \mathcal{G}_{ \theta, \zeta }$ for some (resp. every) $\lambda > {\lambda }_{ f, \theta, \zeta }$.
      This implies that 
      \begin{equation*}
      \sup \Big\{ \lambda: f \wedge \Big( \lambda - \frac{\zeta }{\theta } \Big) \in \mathcal{G}_{ \theta, \zeta } \Big\}
      = \left\{ \begin{aligned}
        & {\lambda }_{ f, \theta, \zeta }, && if ~ f \in \mathbb{L}^{2} ( \Omega ) \setminus \mathcal{G}_{ \theta, \zeta }; \\
        & + \infty, && otherwise.
        \end{aligned} \right.
      \end{equation*}
\end{itemize}
\end{lemma}

\begin{proof}
See \Cref{pf-lem: identity and truncation}.
\end{proof}

Regarding \Cref{lem: identity and truncation}, two points are worth mentioning. 
On the one hand, the first assertion in \Cref{lem: identity and truncation} with its proof provides the following comparison results for $( {V}_{ \theta, \zeta }, {V}_{\theta }, {U}_{\theta } )$ 
with $\zeta \in \mathbb{L}^{2}_{0} ( \Omega )$ and $\mathbb{P} ( \zeta > 0 ) = 1$:
\begin{itemize}
\item ${V}_{ \theta, \zeta } = {V}_{\theta } = {U}_{\theta }$ on $\mathcal{G}_{ \theta, \zeta }$;
\item ${V}_{ \theta, \zeta } > {V}_{\theta } = {U}_{\theta }$ on $\mathcal{G}_{\theta } \setminus \mathcal{G}_{ \theta, \zeta }$;
\item ${V}_{ \theta, \zeta } > {V}_{\theta } > {U}_{\theta }$ on $\mathbb{L}^{2} ( \Omega ) \setminus \mathcal{G}_{\theta }$.
\end{itemize}
In \Cref{exp: illustrate MV MMV SMMV}, ${X}_{x} \in \mathcal{G}_{2}$ (resp. ${X}_{x} \in \mathcal{G}_{ 2, \frac{1}{2} }$) if and only if $x \in [ -1,1 ]$ (resp. $x \in [ - \frac{1}{2}, \frac{1}{2} ]$),
and the aforementioned comparison results are obvious in \Cref{fig: illustrate MV MMV SMMV}.
Based on these comparison results, $\mathcal{G}_{ \theta, \zeta }$ can be named as ``the domain of $\zeta $-monotonicity'', in analogy to the domain of monotonicity $\mathcal{G}_{\theta }$.
In fact, only in $\mathcal{G}_{\theta }$ (resp. $\mathcal{G}_{ \theta, \zeta }$), the MV preference ${U}_{\theta }$ obeys the desired monotonicity condition
${U}_{\theta } (f) \ge {U}_{\theta } (g)$ (resp. ${U}_{\theta } (f) \ge {U}_{\theta } (g) + \mathbb{E} [ ( f - g ) \zeta ]$) for arbitrarily chosen $f \ge g$.
${V}_{ \theta, \zeta }$ fixes the failure of desired strict monotonicity in both ${U}_{\theta }$ and ${V}_{\theta }$ on $\mathbb{L}^{2} ( \Omega ) \setminus \mathcal{G}_{ \theta, \zeta }$,
just as ${V}_{\theta }$ addresses the lack of monotonicity in ${U}_{\theta }$ on $\mathbb{L}^{2} ( \Omega ) \setminus \mathcal{G}_{\theta }$.

On the other hand, the last three assertions in \Cref{lem: identity and truncation} show how to truncate $f \in \mathbb{L}^{2} ( \Omega )$ 
so that the truncation result falls into the domain $\mathcal{G}_{ \theta, \zeta }$, as listed below.
\begin{itemize}
\item $f \wedge ( \lambda - \frac{\zeta }{\theta } ) \in \mathcal{G}_{ \theta, \zeta }$ for $\lambda \le {\lambda }_{ f, \theta, \zeta }$ is necessarily true.
\item $f \wedge ( \lambda - \frac{\zeta }{\theta } ) \in \mathcal{G}_{ \theta, \zeta }$ for $\lambda > {\lambda }_{ f, \theta, \zeta }$ is equivalent to $f \in \mathcal{G}_{ \theta, \zeta }$, 
      namely, $f + \frac{\zeta }{\theta } \le {\lambda }_{ f, \theta, \zeta }$, $\mathbb{P}$-a.s.
\end{itemize}

Next, we show the G\^ateaux differentiability and some explicit expressions of ${V}_{ \theta, \zeta }$.
For the sake of brevity, hereafter we omit the statement $( \theta, \zeta ) \in \mathbb{R}_{+} \times \mathbb{L}^{2}_{0} ( \Omega )$, unless otherwise noted.

\begin{theorem}\label{thm: Gateaux differentiability and equivalent expressions of SMMV}
For any $f \in \mathbb{L}^{2} ( \Omega )$, $d {V}_{ \theta, \zeta } (f) = \zeta + \theta ( {\lambda }_{ f, \theta, \zeta } - f - \frac{\zeta }{\theta } )_{+}$, which realizes the minimum on the right-hand side of \cref{eq: definition of SMMV :eq}, and
\begin{align}\label{eq: direct calculation form of SMMV :eq}
  {V}_{ \theta, \zeta } (f) 
& = \theta \int_{ - \infty }^{ {\lambda }_{ f, \theta, \zeta } } s \mathbb{P} \Big( f + \frac{\zeta }{\theta } \le s \Big) ds
  + \mathbb{E} [ f \zeta ] + \frac{1}{ 2 \theta } \mathbb{E} [ {\zeta }^{2} ] - \frac{1}{ 2 \theta } \\
\notag
& = {U}_{\theta } \bigg( \Big( f + \frac{\zeta }{\theta } \Big) \wedge {\lambda }_{ f, \theta, \zeta } \bigg)
  + \mathbb{E} [ ( f - {\lambda }_{ f, \theta, \zeta } ) \zeta ] + \frac{1}{2 \theta } \Var [ \zeta ] \\
\label{eq: modificed MV form of SMMV :eq}
& = {U}_{\theta } \bigg( f \wedge \Big( {\lambda }_{ f, \theta, \zeta } - \frac{\zeta }{\theta } \Big) \bigg)
  + \mathbb{E} \bigg[ \bigg( f - f \wedge \Big( {\lambda }_{ f, \theta, \zeta } - \frac{\zeta }{\theta } \Big) \bigg) \zeta \bigg] \\
\notag
& = {U}_{\theta } \bigg( f \wedge \Big( {\lambda }_{ f, \theta, \zeta } - \frac{\zeta }{\theta } \Big) \bigg)
  + \mathbb{E} \Big[ \Big( f + \frac{\zeta }{\theta } - {\lambda }_{ f, \theta, \zeta } \Big)_{+} \zeta \Big].
\end{align}
\end{theorem}

\begin{proof}
See \Cref{pf-thm: Gateaux differentiability and equivalent expressions of SMMV}.
\end{proof}

Combining \Cref{lem: identity and truncation} and the expression \cref{eq: modificed MV form of SMMV :eq},
one can characterize the SMMV preference ${V}_{ \theta, \zeta }$ by MV preference ${U}_{ \theta, \zeta }$ with a truncation approach.
In \Cref{lem: identity and truncation} it has been shown that ${V}_{ \theta, \zeta } (f) = {U}_{\theta } (f)$ for $f \in \mathcal{G}_{ \theta, \zeta }$.
In terms of $f \notin \mathcal{G}_{ \theta, \zeta }$, ${V}_{ \theta, \zeta } (f)$ equals the sum of 
\begin{itemize}
\item the MV performance for the truncated random variable $f \wedge ( {\lambda }_{ f, \theta, \zeta } - \frac{\zeta }{\theta } ) \in \mathcal{G}_{ \theta, \zeta }$
\item and the linear modification term for the gap $f - f \wedge ( {\lambda }_{ f, \theta, \zeta } - \frac{\zeta }{\theta } )$, of which the ``growth speed'' is exactly $\zeta $.
\end{itemize}
In fact, fixing the ``basis'' ${U}_{\theta } ( f \wedge ( {\lambda }_{ f, \theta, \zeta } - \frac{\zeta }{\theta } ) )$ 
and arbitrarily choosing $Y$ for the ``growth'' $\mathbb{E} [ ( f - f \wedge ( {\lambda }_{ f, \theta, \zeta } - \frac{\zeta }{\theta } ) ) Y ]$ deliver a minimality of ${V}_{ \theta, \zeta }$.
Conversely, if we arbitrarily choose $g$ for the basis ${U}_{\theta } (g)$ but fix the growth $\mathbb{E} [ ( f - g ) \zeta ]$, then this leads to a maximality of ${V}_{ \theta, \zeta } (f)$.
This maximality property is similar to the definition of MMV preference in \cite{Cerny-2020,Cerny-Ruf-Schweizer-2025}, which exploits the monotone hull of the MV preference ${U}_{1}$.
We formulate these minimality/maximality properties in the following \Cref{thm: minimality and maximality of SMMV}.

\begin{proposition}\label{thm: minimality and maximality of SMMV}
The minimality of SMMV preferences is formulated by
\begin{equation*}
{V}_{ \theta, \zeta } (f) = \min \{ V (f): V |_{ \mathcal{G}_{ \theta, \zeta } } = {U}_{\theta } |_{ \mathcal{G}_{ \theta, \zeta } }, 
                                           \partial V (g) \cap \mathbb{L}^{2}_{\zeta + } ( \Omega ) \ne \varnothing, \forall g \in \mathbb{L}^{2} ( \Omega ) \},
\end{equation*}
while the maximality of SMMV preferences is formulated by
\begin{equation*}
{V}_{ \theta, \zeta } (f) = \max_{ g \in \mathcal{G}_{ \theta, \zeta }, g \le f } \{ {U}_{\theta } (g) + \mathbb{E} [ ( f - g ) \zeta ] \}
                          = \max_{ g \in \mathbb{L}^{2} ( \Omega ), g \le f } \{ {U}_{\theta } (g) + \mathbb{E} [ ( f - g ) \zeta ] \}.
\end{equation*}
\end{proposition}

\begin{proof}
See \Cref{pf-thm: minimality and maximality of SMMV}.
\end{proof}

\begin{remark}[Supremal convolution characterization of SMMV]\label{rem: Supremal convolution characterization of SMMV}
We now extend the results for generating the MMV preference ${V}_{1}$ illustrated by \cite[eqn. (7)]{Cerny-2020}, which exploits the monotone and the cash-invariant hull (i.e. $\Box {L}_{0}$ and $\Box C$ below),
to our SMMV preferences ${V}_{ \theta, \zeta }$.
More precisely, in addition to the cash indicator function
\begin{equation*}
C (f) := \left\{ \begin{aligned}
         & c,        && for ~ f = c, ~ \mathbb{P}-a.s., ~ c \in \mathbb{R}; \\
         & - \infty, && otherwise,
         \end{aligned} \right.
\end{equation*}
and the supermal convolution $( U \Box V ) (f) := \sup \{ U (g) + V (h): g,h \in \mathbb{L}^{2} ( \Omega ), g + h = f \}$ considered in \cite{Cerny-2020},
we define the quadratic utility ${Q}_{\theta } (f) := \mathbb{E} [ f - \frac{\theta }{2} {f}^{2} ]$ 
and the concave function (which is linear on its effective domain $\mathbb{L}_{+}^{2} ( \Omega )$):
\begin{equation*}
{L}_{\zeta } (f) := \left\{ \begin{aligned}
                    & \mathbb{E} [ f \zeta ], && for ~ f \in \mathbb{L}_{+}^{2} ( \Omega ); \\
                    & - \infty,               && otherwise.
                    \end{aligned} \right.
\end{equation*}
Let ${Q}_{ \theta, \zeta } (f) := {Q}_{\theta } ( f \wedge \frac{ 1 - \zeta }{\theta } ) + {L}_{\zeta } ( f - f \wedge \frac{ 1 - \zeta }{\theta } )$.
Given that the supermal convolution is commutative and associative, 
one obtains $( {Q}_{\theta } \Box C ) \Box {L}_{\zeta } = {U}_{\theta } \Box {L}_{\zeta } = {V}_{ \theta, \zeta } = {Q}_{ \theta, \zeta } \Box C = ( {Q}_{\theta } \Box {L}_{\zeta } ) \Box C$,
which shows how to generate the SMMV preference ${V}_{ \theta, \zeta }$ from the quadratic utility ${Q}_{\theta }$.
Graphically, we establish the following commutative diagram:
\begin{equation*}
\begin{matrix}
{Q}_{\theta } & \xrightarrow{ \Box {L}_{\zeta } } & {Q}_{ \theta, \zeta } \\
\quad \Big\downarrow {\scriptstyle \Box C} &  & \Big\downarrow {\scriptstyle \Box C} \\
{U}_{\theta } & \xrightarrow{ \Box {L}_{\zeta } } & {V}_{ \theta, \zeta }.
\end{matrix}
\end{equation*}
By setting $( \theta, \zeta ) = ( 1,0 )$, \cite[eqn. (7)]{Cerny-2020} can be recovered.
\end{remark}

Furthermore, the SMMV preferences have the following properties for risk ranking.

\begin{proposition}\label{thm: MUA and SSD for SMMV}
On the one hand, $\theta \ge \hat{\theta }$ if and only if ${V}_{ \theta, \zeta }$ is more uncertainty averse than ${V}_{ \hat{\theta }, \zeta }$, namely,
\begin{equation*}
\Big( {V}_{ \theta, \zeta } (f) \ge {V}_{ \theta, \zeta } (c) ~ \Rightarrow ~ {V}_{ \hat{\theta }, \zeta } (f) \ge {V}_{ \hat{\theta }, \zeta } (c) \Big),
\quad \forall ( f,c ) \in \mathbb{L}^{2} ( \Omega ) \times \mathbb{R}.
\end{equation*}
On the other hand, ${\lambda }_{ g, \theta, \zeta } \le {\lambda }_{ f, \theta, \zeta }$ and ${V}_{ \theta, \zeta } (f) \ge {V}_{ \theta, \zeta } (g) + \mathbb{E} [ ( f - g ) \zeta ]$,
if $f + \frac{\zeta }{\theta }$ second-order stochastically dominates $g + \frac{\zeta }{\theta }$.
\end{proposition}

\begin{proof}
See \Cref{pf-thm: MUA and SSD for SMMV}.
Notably, our approach is considerably more straightforward and accessible, compared to \cite{Maccheroni-Marinacci-Rustichini-Taboga-2009} involving multiple technical steps and mathematical tools such as convex order and Hardy-Littlewood inequality.
\end{proof}

The first property stated in \Cref{thm: MUA and SSD for SMMV} is an analogue to \cite[Proposition 2.1]{Maccheroni-Marinacci-Rustichini-Taboga-2009}, which establishes a monotonicity in risk aversion: 
a certain result that is not preferred under higher risk aversion will also not be preferred under lower risk aversion. 
The second property extends the result of \cite[Theorem 2.3]{Maccheroni-Marinacci-Rustichini-Taboga-2009}, which states that ${V}_{\theta } (f) \ge {V}_{\theta } (g)$ if $f$ second-order stochastically dominates $g$.

To end this section, we establish an equivalent expression for ${V}_{ \theta, \zeta } (f)$ through a change of measure, 
analogous to the formulation of MMV preference in \cite{Maccheroni-Marinacci-Rustichini-Taboga-2009,Trybula-Zawisza-2019,Strub-Li-2020,Li-Guo-2021}. 
For $Y \in \mathbb{L}^{2}_{+} ( \Omega )$ and $\mathbb{E} [Y] = 1$, one can define the probability measure $\mathbb{Q} \ll \mathbb{P}$ on $( \Omega, \mathcal{F} )$ by $\mathbb{Q} (A) = \int_{A} Y d \mathbb{P}$.
Thus, $Y = \frac{ d \mathbb{Q} }{ d \mathbb{P} }$ is the Radon-Nikod\'ym derivative, and 
\begin{equation*}
  \mathbb{E} [ Y f ] + \frac{1}{2 \theta } ( \mathbb{E} [ {Y}^{2} ] - 1 ) 
= \mathbb{E}^{\mathbb{Q}} [f] + \frac{1}{ 2 \theta } \bigg( \mathbb{E} \Big[ \Big( \frac{ d \mathbb{Q} }{ d \mathbb{P} } \Big)^{2} \Big] - 1 \bigg),
\end{equation*}
where $\mathbb{E}^{\mathbb{Q}}$ denotes the expectation operator under $\mathbb{Q}$.
Conversely, for $\mathbb{Q} \ll \mathbb{P}$, according to Radon-Nikod\'ym theorem, there exists a Radon-Nikod\'ym derivative $Y = \frac{ d \mathbb{Q} }{ d \mathbb{P} }$.
Therefore, the MMV preference \cref{eq: biconjugate form of MMV :eq} can be re-expressed as
\begin{equation*}
  {V}_{\theta } (f) 
= \inf_{ \mathbb{Q} \ll \mathbb{P}, \frac{ d \mathbb{Q} }{ d \mathbb{P} } \in \mathbb{L}^{2} ( \Omega ) } 
  \bigg\{ \mathbb{E}^{\mathbb{Q}} [f] + \frac{1}{ 2 \theta } \bigg( \mathbb{E} \Big[ \Big( \frac{ d \mathbb{Q} }{ d \mathbb{P} } \Big)^{2} \Big] - 1 \bigg) \bigg\}
= \inf_{ \frac{ d \mathbb{Q} }{ d \mathbb{P} } \in \mathbb{L}^{2} ( \Omega ) } 
  \bigg\{ \mathbb{E}^{\mathbb{Q}} [f] + \frac{1}{ 2 \theta } C ( \mathbb{Q} \| \mathbb{P} ) \bigg\},
\end{equation*}
where $C ( \mathbb{Q} \| \mathbb{P} )$ is the so-called ``relative Gini concentration index'' given by
\begin{equation*}
  C ( \mathbb{Q} \| \mathbb{P} ) 
= \left\{ \begin{aligned}
          & \mathbb{E} \Big[ \Big( \frac{ d \mathbb{Q} }{ d \mathbb{P} } \Big)^{2} \Big] - 1, && if ~ \mathbb{Q} \ll \mathbb{P}; \\
          & + \infty, && otherwise.
          \end{aligned} \right. 
\end{equation*}
In the same manner, in conjunction with the substitution $Z = \frac{ Y - \zeta }{ 1 - \mathbb{E} [ \zeta ] }$ and the shorthand notation $\kappa := 1 - \mathbb{E} [ \zeta ]$ for brevity, 
we have the following equivalent expression for the SMMV preference:
\begin{align}\label{eq: measure optimization form of SMMV}
    {V}_{ \theta, \zeta } (f) 
& = \mathbb{E} [ f \zeta ]
  + \inf_{ Z \in \mathbb{L}^{2}_{+} ( \Omega ), ~ \mathbb{E} [Z] = 1 }
    \bigg\{ \kappa \mathbb{E} \Big[ \Big( f + \frac{\zeta }{\theta } \Big) Z \Big] + \frac{ {\kappa }^{2} }{ 2 \theta } \mathbb{E} [ {Z}^{2} ] \bigg\}
  + \frac{1}{2 \theta } \mathbb{E} [ {\zeta }^{2} ] - \frac{1}{ 2 \theta } \\ 
\notag
& = \mathbb{E} \Big[ \Big( f + \frac{\zeta }{\theta } \Big) \zeta \Big]
  + \inf_{ \frac{ d \mathbb{Q} }{ d \mathbb{P} } \in \mathbb{L}^{2} ( \Omega ) }
    \bigg\{ \kappa \mathbb{E}^{\mathbb{Q}} \Big[ f + \frac{\zeta }{\theta } \Big] + \frac{ {\kappa }^{2} }{ 2 \theta } C ( \mathbb{Q} \| \mathbb{P} ) \bigg\}
  - \frac{1}{\theta } \mathbb{E} [ \zeta ] - \frac{1}{ 2 \theta } \Var [ \zeta ]. 
\end{align}

\section{Single-period static portfolio selection}
\label{sec: Static portfolio selection}

\subsection{Model and problem formulation}

Let $r$ be the risk-free yield rate, $\vec{R}$ the vector of the yield rates of $n$ risky assets with the variance-covariance matrix $\Var [ \vec{R} ]$ under $\mathbb{P}$, 
and $\vec{\alpha }$ the ratio of wealth invested on the $n$ risky assets without any constraint.
Assume that the market has no arbitrage, and $\Var [ \vec{R} ]$ is invertible so that any asset cannot be replicated by others.
For a unit initial wealth, the terminal wealth corresponding to portfolio strategy $\vec{\alpha }$ is 
\begin{equation*}
{X}_{ \vec{\alpha } } := r ( 1 - \langle \vec{\alpha }, \vec{1} \rangle ) + \langle \vec{\alpha }, \vec{R} \rangle = r + \langle \vec{\alpha }, \vec{R} - r \vec{1} \rangle,
\end{equation*}
where $\langle \cdot, \cdot \rangle $ denotes the inner product on $\mathbb{R}^{n} \times \mathbb{R}^{n}$ and $\vec{1}$ is the vector whose components are all $1$.  

For initial wealth $x > 0$, the agent with MV preferences ${U}_{\theta }$ aims to maximize
\begin{equation*}
  \mathbb{E} [ x {X}_{ \vec{\alpha } } ] - \frac{\theta }{2} \Var [ x {X}_{ \vec{\alpha } } ]
= x r + x \Big( \langle \vec{\alpha }, \mathbb{E} [ \vec{R} - r \vec{1} ] \rangle - \frac{ x \theta }{2} \langle \vec{\alpha }, \Var [ \vec{R} ] \vec{\alpha } \rangle \Big)
\end{equation*}
on $\vec{\alpha } \in \mathbb{R}^{n}$.
Denote by $\vec{\alpha }^{*}_{mv} (x)$ the maximizer for this classical linear-quadratic optimization problem. 
It is easy to arrive at the gradient optimality condition
\begin{equation}\label{eq: MV portfolio for SP :eq}
\vec{\alpha }^{*}_{mv} (x) = \frac{1}{ x \theta } ( \Var [ \vec{R} ] )^{-1} \mathbb{E} [ \vec{R} - r \vec{1} ],
\end{equation}
which implies that the optimal investment amount $x \vec{\alpha }^{*}_{mv} (x)$ is independent of the initial wealth $x$.
Moreover, the problem with initial wealth $x > 0$ and risk aversion parameter $\theta $ is equivalent to that with unit initial wealth and risk aversion parameter $x \theta $.
This equivalence property also holds for the portfolio problems with the SMMV preferences ${V}_{ \theta, \zeta }$, since
\begin{align*}
    \sup_{ \vec{\alpha } \in \mathbb{R}^{n} } {V}_{ \theta, \zeta } ( x {X}_{ \vec{\alpha } } )
& = x \sup_{ \vec{\alpha } \in \mathbb{R}^{n} } \inf_{ Z \in \mathbb{L}^{2}_{+} ( \Omega ), ~ \mathbb{E} [Z] = 1 }
    \Big\{ \mathbb{E} [ {X}_{ \vec{\alpha } } ( \zeta + \kappa Z ) ]
         + \frac{1}{ 2 x \theta } \mathbb{E} [ ( \zeta + \kappa Z )^{2} ] - \frac{1}{ 2 x \theta } \Big\} \\
& = x \sup_{ \vec{\alpha } \in \mathbb{R}^{n} } {V}_{ x \theta, \zeta } ( {X}_{ \vec{\alpha } } ), \qquad \forall x > 0,
\end{align*}
where $\kappa = 1 - \mathbb{E} [ \zeta ]$ and $Z = \frac{ Y - \zeta }{\kappa }$.
Hence, hereafter we only consider the SMMV portfolio problems with unit initial wealth, i.e.,
\begin{equation}\label{eq: SP :eq}
  \sup_{ \vec{\alpha } \in \mathbb{R}^{n} } {V}_{ \theta, \zeta } ( {X}_{ \vec{\alpha } } )
= r + \sup_{ \vec{\alpha } \in \mathbb{R}^{n} } \inf_{ Z \in \mathbb{L}^{2}_{+} ( \Omega ), \mathbb{E} [Z] = 1 }
      \bigg\{ \langle \vec{\alpha }, \mathbb{E} [ ( \vec{R} - r \vec{1} ) ( \kappa Z + \zeta ) ] \rangle
            + \frac{1}{ 2 \theta } \mathbb{E} [ ( \kappa Z + \zeta )^{2} ] \bigg\}
    - \frac{1}{ 2 \theta },
\end{equation}
for which the maximizer is denoted by $\vec{\alpha }^{*}$.

\subsection{Properties of solution}

In this subsection, we shall characterize the solution for \cref{eq: SP :eq} at first, including the gradient condition and the max-min duality, and then discuss the existence and uniqueness of solution.
In addition, exploiting the Lagrange multiplier method, we do not only compare the SMMV and MV optimal strategies, but also conduct a linear equation algorithm for deriving the SMMV optimal strategy.
Notably, this algorithm (even for the MMV preference, i.e., $\zeta \equiv 0$) is not reported in \cite{Maccheroni-Marinacci-Rustichini-Taboga-2009}, 
and is not straightforward given the gradient condition, i.e., a system of several nonlinear equations.

Define ${\lambda }_{ \vec{\alpha } } := {\lambda }_{ {X}_{ \vec{\alpha } }, \theta, \zeta }$ for brevity.
First, we derive the gradient condition (namely, the first-order derivative condition) that is necessary and sufficient for maximality.
This extends existing results reported in \cite[Theorem 4.1]{Maccheroni-Marinacci-Rustichini-Taboga-2009}.

\begin{theorem}\label{thm: optimality condition for SP}
If the maximizer $\vec{\alpha }^{*}$ for \cref{eq: SP :eq} exists, then it fulfills the following equations:
\begin{equation}\label{eq: optimality condition for SP :eq}
\left\{ \begin{aligned}
&   \mathbb{E} [ ( \vec{R} - r \vec{1} ) \zeta ]
  + \kappa \mathbb{E} \Big[ \vec{R} - r \vec{1} \Big| {X}_{ \vec{\alpha }^{*} } + \frac{\zeta }{\theta } \le {\lambda }_{ \vec{\alpha }^{*} } \Big] \\
& - \mathbb{P} \Big( {X}_{ \vec{\alpha }^{*} } + \frac{\zeta }{\theta } \le {\lambda }_{ \vec{\alpha }^{*} } \Big)
    \Cov \Big[ \vec{R}, \zeta \Big| {X}_{ \vec{\alpha }^{*} } + \frac{\zeta }{\theta } \le {\lambda }_{ \vec{\alpha }^{*} } \Big] \\
& = \mathbb{P} \Big( {X}_{ \vec{\alpha }^{*} } + \frac{\zeta }{\theta } \le {\lambda }_{ \vec{\alpha }^{*} } \Big)
    \Var \Big[ \vec{R} \Big| {X}_{ \vec{\alpha }^{*} } + \frac{\zeta }{\theta } \le {\lambda }_{ \vec{\alpha }^{*} } \Big] \theta \vec{\alpha }^{*}, \\
& \frac{\kappa }{\theta } = \mathbb{E} \Big[ \Big( {\lambda }_{ \vec{\alpha }^{*} } - {X}_{ \vec{\alpha }^{*} } - \frac{\zeta }{\theta } \Big)_{+} \Big],
\end{aligned} \right.
\end{equation}
where $\Cov [ \cdot, \cdot | \cdot ]$ denotes the conditional covariance vector under $\mathbb{P}$.
Conversely, if \cref{eq: optimality condition for SP :eq} admits a solution $( \vec{\alpha }^{*}, {\lambda }_{ \vec{\alpha }^{*} } )$ as a system of $n+1$ equations for $n+1$ unknowns,
then $\vec{\alpha }^{*}$ realizes the maximum in \cref{eq: SP :eq}.
\end{theorem}

\begin{proof}
See \Cref{pf-thm: optimality condition for SP}.
\end{proof}

Combine \cref{eq: SP :eq} with \Cref{thm: Gateaux differentiability and equivalent expressions of SMMV},
we conclude that the pair $( \vec{\alpha }^{*}, \frac{\theta }{\kappa } ( {\lambda }_{ \vec{\alpha }^{*} } - {X}_{ \vec{\alpha }^{*} } - \frac{\zeta }{\theta } )_{+} )$, if  exists,
solves the linear-quadratic max-min problem formulated by
\begin{equation}\label{eq: max-min problem for SP :eq}
\mathcal{M} := \max_{ \vec{\alpha } \in \mathbb{R}^{n} } \min_{ Z \in \mathbb{L}^{2}_{+} ( \Omega ), \mathbb{E} [Z] = 1 } 
               \Big\{ \langle \theta \vec{\alpha }, \mathbb{E} [ ( \vec{R} - r \vec{1} ) ( \kappa Z + \zeta ) ] \rangle + \frac{1}{2} \mathbb{E} [ ( \kappa Z + \zeta )^{2} ] - \frac{1}{2} \Big\}.
\end{equation}
In particular, taking $\vec{\alpha } = 0$ provides the lower bound of \cref{eq: max-min problem for SP :eq}; that is,
\begin{equation*}
\mathcal{M} \ge \min_{ Z \in \mathbb{L}^{2}_{+} ( \Omega ), \mathbb{E} [Z] = 1 } \bigg\{ \frac{1}{2} \mathbb{E} [ ( \kappa Z + \zeta )^{2} ] - \frac{1}{2} \bigg\}
              = \frac{1}{2} \min_{ Z \in \mathbb{L}^{2}_{+} ( \Omega ), \mathbb{E} [Z] = 1 } \Var [ \kappa Z + \zeta ] 
              = 0.
\end{equation*}
Consequently, plugging $\vec{\alpha } = \vec{\alpha }^{*}$ and $Z = \frac{ 1 - \zeta }{\kappa }$ into the right-hand side of \cref{eq: max-min problem for SP :eq} yields $\langle \vec{\alpha }^{*}, \mathbb{E} [ \vec{R} ] - r \vec{1} \rangle \ge 0$. 
In general, we have the following theorem regarding the existence of a solution.

\begin{proposition}\label{lem: saddle point in SP}
$( \vec{\alpha }^{*}, {Z}_{*} )$ solves the max-min problem formulated by \cref{eq: max-min problem for SP :eq}, if and only if $( \vec{\alpha }^{*}, {Z}_{*} )$ is the saddle point for \cref{eq: max-min problem for SP :eq}.
\end{proposition}

\begin{proof}
See \Cref{pf-lem: saddle point in SP}.
\end{proof}

For the embedded minimization problem in \cref{eq: max-min problem for SP :eq} corresponding to the maximizer $\vec{\alpha }^{*}$, the minimizer
\begin{equation*}
{Z}_{*} := \frac{\theta }{\kappa } \Big( {\lambda }_{ \vec{\alpha }^{*} } - {X}_{ \vec{\alpha }^{*} } - \frac{\zeta }{\theta } \Big)_{+}
\in \argmin_{ Z \in \mathbb{L}^{2}_{+} ( \Omega ), \mathbb{E} [Z] = 1 }
    \Big\{ \mathbb{E} [ \langle \theta \vec{\alpha }^{*}, \vec{R} \rangle Z ] + \frac{\kappa }{2} \mathbb{E} \Big[ \Big( Z + \frac{\zeta }{\kappa } \Big)^{2} \Big] \Big\}.
\end{equation*}
For this constrained problem, one can define the Lagrangian $\mathcal{L}_{1}: \mathbb{L}^{2} ( \Omega ) \times \mathbb{L}^{2} ( \Omega ) \times \mathbb{R} \to \mathbb{R}$ by
\begin{equation*}
\mathcal{L}_{1} ( Z, \beta, \mu ) 
:= \frac{\kappa }{2} \mathbb{E} \Big[ \Big( Z + \frac{\zeta }{\kappa } \Big)^{2} \Big] + \mathbb{E} [ \langle \theta \vec{\alpha }^{*}, \vec{R} \rangle Z ]
 - \mathbb{E} [ \beta Z ] - \mu ( \mathbb{E} [Z] - 1 ),
\end{equation*}
and then obtains the Karush-Kuhn-Tucker (KKT) condition
\begin{equation}\label{eq: KKT condition :eq}
\left\{ \begin{aligned}
& 0 = \kappa {Z}_{*} + \zeta + \langle \theta \vec{\alpha }^{*}, \vec{R} \rangle - \beta - \mu, \quad \mathbb{P}-a.s.; \\
& \mathbb{E} [ {Z}_{*} ] = 1; \\
& \beta \ge 0, ~ {Z}_{*} \ge 0, ~ \beta {Z}_{*} = 0, ~ \mathbb{P}-a.s.
\end{aligned} \right.
\end{equation}
Then, the following \Cref{lem: express SMMV portfolio} provides the expression for $\vec{\alpha }^{*}$ in terms of the Lagrange multiplier $\beta $ and the MV optimal portfolio $\vec{\alpha }^{*}_{mv} (1)$ given by \cref{eq: MV portfolio for SP :eq}.

\begin{proposition}\label{lem: express SMMV portfolio}
Assume that $\vec{\alpha }^{*}$ realizes the maximum in \cref{eq: SP :eq}, 
and $( {Z}_{*}, \beta, \mu ) \in \mathbb{L}^{2} ( \Omega ) \times \mathbb{L}^{2} ( \Omega ) \times \mathbb{R}$ solves the KKT condition \cref{eq: KKT condition :eq}.
Then,
\begin{align}\label{eq: express SMMV portfolio by MV portfolio :eq}
& \vec{\alpha }^{*} = \vec{\alpha }^{*}_{mv} (1) + \frac{1}{\theta } ( \Var [ \vec{R} ] )^{-1} \Cov [ \vec{R}, \beta ], \\
\label{eq: express SMMV portfilio by Lagrange multiplier :eq}
& \langle \theta \vec{\alpha }^{*}, \Cov [ \vec{R}, \beta ] \rangle = \Var [ \beta ] + \mathbb{E} [ \beta ( 1 - \zeta ) ].
\end{align}
Furthermore, ${X}_{ \vec{\alpha }^{*} } \in \mathcal{G}_{ \theta, \zeta }$ is equivalent to $\Var [ \beta ] = 0$, which implies $\vec{\alpha }^{*} = \vec{\alpha }^{*}_{mv} (1)$ due to \cref{eq: express SMMV portfolio by MV portfolio :eq}.
\end{proposition}

\begin{proof}
See \Cref{pf-lem: express SMMV portfolio}.
\end{proof}

However, the Lagrange multiplier $\beta $ is implicit, even though given the KKT condition \cref{eq: KKT condition :eq}.
We aim to compare $\vec{\alpha }^{*}$ with $\vec{\alpha }^{*}_{mv} (1)$ based on the sign of $\vec{\alpha }^{*}$, similar to \cite[Proposition 4.1]{Maccheroni-Marinacci-Rustichini-Taboga-2009}.
The results are collected in the following \Cref{thm: optimal portfolio comparison}. 
Notably, unlike \cite[Proposition 4.1]{Maccheroni-Marinacci-Rustichini-Taboga-2009}, we do not assume the finiteness of $\Omega $.

\begin{proposition}\label{thm: optimal portfolio comparison}
Assume that there is only one risky asset, so that the vector triplet $( \vec{R}, \vec{\alpha }^{*}, \vec{\alpha }^{*}_{mv} )$ reduces to $( R, {\alpha }^{*}, {\alpha }^{*}_{mv} ) \in \mathbb{L}^{2} ( \Omega ) \times \mathbb{R} \times \mathbb{R}$.
Then,
\begin{equation*}
{\alpha }^{*} = {\alpha }^{*}_{mv} (1) + \frac{ \Cov [ R, \beta ] }{ \theta \Var [R] }, \quad
{\alpha }^{*} \Cov [ R, \beta ] = \frac{1}{\theta } \big( \Var [ \beta ] + \mathbb{E} [ \beta ( 1 - \zeta ) ] \big) \ge 0.
\end{equation*}
Consequently,
\begin{equation*}
\left\{ \begin{aligned}
& {\alpha }^{*} > 0 \quad \Rightarrow \quad \mathbb{E} [R] \ge r, ~ \Cov [ R, \beta ] \ge 0 \quad \Rightarrow \quad {\alpha }^{*} \ge {\alpha }^{*}_{mv} (1) \ge 0; \\
& {\alpha }^{*} < 0 \quad \Rightarrow \quad \mathbb{E} [R] \le r, ~ \Cov [ R, \beta ] \le 0 \quad \Rightarrow \quad {\alpha }^{*} \le {\alpha }^{*}_{mv} (1) \le 0; \\
& {\alpha }^{*} = 0 \quad \Rightarrow \quad {X}_{ {\alpha }^{*} } \in \mathcal{G}_{ \theta, \zeta } \quad \Rightarrow \quad \Var[ \beta ] = 0 
                    \quad \Rightarrow \quad \Cov [ R, \beta ] = 0 \quad \Rightarrow \quad {\alpha }^{*} = {\alpha }^{*}_{mv} (1).
\end{aligned} \right.
\end{equation*}
In particular, if $\mathbb{P} ( {X}_{ {\alpha }^{*} } + \frac{\zeta }{\theta } > {\lambda }_{ {\alpha }^{*} } ) > 0$, i.e. ${X}_{ {\alpha }^{*} } \notin \mathcal{G}_{ \theta, \zeta }$, then
\begin{equation*}
\left\{ \begin{aligned}
& {\alpha }^{*} > 0 \quad \Rightarrow \quad \Cov [ R, \beta ] > 0 \quad \Rightarrow \quad {\alpha }^{*} > {\alpha }^{*}_{mv} (1); \\
& {\alpha }^{*} < 0 \quad \Rightarrow \quad \Cov [ R, \beta ] < 0 \quad \Rightarrow \quad {\alpha }^{*} < {\alpha }^{*}_{mv} (1).
\end{aligned} \right.
\end{equation*}
\end{proposition}

\begin{proof}
Given \Cref{lem: express SMMV portfolio} with \cref{eq: MV portfolio for SP :eq} and ${\alpha }^{*} ( \mathbb{E} [R] - r ) \ge 0$, the proof is straightforward, so we omit the details for brevity. 
\end{proof}

Now we provide an algorithm for deriving the solution $( \vec{\alpha }^{*}, {Z}_{*} )$ for \cref{eq: max-min problem for SP :eq}, in order to facilitate practical applications of the SMMV preferences.

\begin{theorem}\label{thm: Lagrange multiplier as solution}
Suppose that ${Y}_{*}$ is the unique solution of the following minimization problem:
\begin{equation}\label{eq: minimization problem in SP :eq}
\text{minimizing} \quad \mathbb{E} [ |Y|^{2} ] \quad
\text{subject to} \quad Y \in \mathbb{L}^{2}_{ \zeta + } ( \Omega ), ~ \mathbb{E} [Y] = 1, ~ \mathbb{E} [ ( \vec{R} - r \vec{1} ) Y ] = \vec{0}.
\end{equation}
Then, arising from the Lagrangian $\mathcal{L}_{2}: \mathbb{L}^{2} ( \Omega ) \times \mathbb{L}^{2} ( \Omega ) \times \mathbb{R} \times \mathbb{R}^{n} \to \mathbb{R}$ formulated by
\begin{equation*}
\mathcal{L}_{2} ( Y, \beta, \mu, \vec{\alpha } ) 
:= \frac{1}{2} \mathbb{E} [ {Y}^{2} ] - \mathbb{E} [ \beta ( Y - \zeta ) ] + \theta ( r - \mu ) ( \mathbb{E} [Y] - 1 ) + \theta \langle \vec{\alpha }, \mathbb{E} [ ( \vec{R} - r \vec{1} ) Y ] \rangle
\end{equation*}
with the following KKT condition:
\begin{equation}\label{eq: modified KKT condition in SP :eq}
\left\{ \begin{aligned}
& 0 = {Y}_{*} - \beta - \theta ( \mu - r ) + \theta \langle \vec{\alpha }, \vec{R} - r \vec{1} \rangle, \quad \mathbb{P}-a.s.; \\
& \mathbb{E} [ {Y}_{*} ] = 1, \quad \mathbb{E} [ ( \vec{R} - r \vec{1} ) {Y}_{*} ] = \vec{0}; \\
& \beta \ge 0, ~ {Y}_{*} \ge \zeta, ~ \beta ( {Y}_{*} - \zeta ) = 0, \quad \mathbb{P}-a.s.,
\end{aligned} \right.
\end{equation}
$\vec{\alpha }$ satisfies ${Y}_{*} = \zeta + \theta ( {\lambda }_{ \vec{\alpha } } - {X}_{ \vec{\alpha } } - \frac{\zeta }{\theta } )_{+}$ and $\mu = {\lambda }_{ \vec{\alpha } }$, and realizes the maximum for \cref{eq: SP :eq}.
Furthermore, if $\Var [ \vec{R} | {Y}_{*} > \zeta ]$ is invertible, then 
\begin{equation*}
\vec{\alpha } = \frac{1}{\theta } ( \Var [ \vec{R} | {Y}_{*} > \zeta ] )^{-1} 
                \bigg( \frac{1}{ \mathbb{P} ( {Y}_{*} > \zeta ) } \Big( \mathbb{E} [ ( \vec{R} - r \vec{1} ) \zeta ] + \kappa \mathbb{E} [ \vec{R} - r \vec{1} | {Y}_{*} > \zeta ] \Big) 
                     - \Cov [ \vec{R}, \zeta | {Y}_{*} > \zeta ] \bigg).
\end{equation*}
\end{theorem}

\begin{proof}
See \Cref{pf-thm: Lagrange multiplier as solution}.
\end{proof}

Notably, the two Lagrange multipliers $\beta $ in the KKT conditions \cref{eq: modified KKT condition in SP :eq} and \cref{eq: KKT condition :eq} coincide,
since these two KKT conditions coincide under some appropriate variable substitution for $\mu $.
This coincidence not only implies that \Cref{lem: express SMMV portfolio} and \Cref{thm: optimal portfolio comparison} also hold for $\beta $ as defined by \cref{eq: modified KKT condition in SP :eq},
but also highlights the necessity of such a minimizer ${Y}_{*}$ for the existence of $\vec{\alpha }^{*}$. 
In other words, $\vec{\alpha }^{*}$ realizes the maximum in \cref{eq: SP :eq}, if and only if
${Y}_{*} = \zeta + \theta ( {\lambda }_{ \vec{\alpha }^{*} } - {X}_{ \vec{\alpha }^{*} } - \frac{\zeta }{\theta } )_{+}$ realizes the minimum in \cref{eq: minimization problem in SP :eq}.
In particular, if
\begin{equation*}
\zeta \le 1 - \langle \vec{R} - \mathbb{E} [ \vec{R} ], ( \Var [ \vec{R} ] )^{-1} \mathbb{E} [ \vec{R} - r \vec{1} ] \rangle
\in \argmin_{ Y \in \mathbb{L}^{2} ( \Omega ), ~ \mathbb{E} [Y] = 1, ~ \mathbb{E} [ ( \vec{R} - r \vec{1} ) Y ] = \vec{0} } \mathbb{E} [ {Y}^{2} ],
\end{equation*}
then ${Y}_{*} = 1 - \langle \vec{R} - \mathbb{E} [ \vec{R} ], ( \Var [ \vec{R} ] )^{-1} \mathbb{E} [ \vec{R} - r \vec{1} ] \rangle $ is the minimizer for \cref{eq: minimization problem in SP :eq},
and the constraint $Y \in \mathbb{L}^{2}_{ \zeta + } ( \Omega )$ therein is not tight, resulting in $\beta = 0$.
By \Cref{lem: express SMMV portfolio}, it then follows immediately that $\vec{\alpha }^{*} = \vec{\alpha }^{*}_{mv} (1)$.
In general, to derive $\vec{\alpha }^{*}$, it suffices to solve the following linear system for $( \vec{\alpha}, \mu, \beta )$ arising from \cref{eq: modified KKT condition in SP :eq}:
\begin{equation}\label{eq: linear system given by KKT :eq}
\left\{ \begin{aligned}
& 1 + \theta r = \mathbb{E} [ \beta ] + \theta \mu - \theta ( \mathbb{E} [ \vec{R} ] - r \vec{1} )^{\top } \vec{\alpha }, \\ 
& \theta r ( \mathbb{E} [ \vec{R} ] - r \vec{1} ) = \mathbb{E} [ ( \vec{R} - r \vec{1} ) \beta ] + \theta ( \mathbb{E} [ \vec{R} ] - r \vec{1} ) \mu - \theta \mathbb{E} [ ( \vec{R} - r \vec{1} ) ( \vec{R} - r \vec{1} )^{\top } ] \vec{\alpha }, \\
& \zeta + \theta r = \beta + \theta \mu - \theta ( \vec{R} - r \vec{1} )^{\top } \vec{\alpha }, \quad on ~ \{ \beta > 0 \}.
\end{aligned} \right.
\end{equation}

\begin{remark}[existence of solution]\label{rem: no static solution}
Let us treat $Y$ as a Radon-Nikod\'ym derivative of the risk-neutral measure $\mathbb{Q}$ with respect to $\mathbb{P}$, as $\mathbb{E} [ \vec{R} Y ] = r \vec{1}$.
According to the first fundamental theorem of asset pricing (namely, Dalang-Morton-Willinger theorem, see \cite[Theorem 6.5.1]{Delbaen-Schachermayer-2006}), the no-arbitrage condition ensures that 
$\mathbb{Y}_{0} := \{ Y \in \mathbb{L}^{2}_{+} ( \Omega ): ~ \mathbb{E} [Y] = 1, ~ \mathbb{E} [ ( \vec{R} - r \vec{1} ) Y ] = \vec{0} \}$ is not empty.
Notably, the no-arbitrage condition can be replaced by the no-free-lunch-with-vanishing-risk condition, according to Kreps-Yan Theorem (see \cite[Theorem 5.2.2, p. 77]{Delbaen-Schachermayer-2006}).
However, if $\mathbb{L}^{2}_{ \zeta + } ( \Omega ) \cap \mathbb{Y}_{0} = \varnothing $, then $\mathbb{Y} = \varnothing $, and hence the maximization problem given by \cref{eq: SP :eq} has no solution.
Nevertheless, as $\zeta $ can be artificially chosen and $\mathbb{Y}_{0}$ completely relies on the market, 
one can take $\zeta \le Y \in \mathbb{Y}_{0}$ in advance to avoid these situations where there is no solution.

Suppose that $\mathbb{Y} \ne \varnothing $.
If $\mathcal{F}$ is generated by a finite partition of $\Omega $, including the case with finite $\Omega $, then the existence of minimizer is a straightforward result of Weierstrass Theorem.
In general, observing that the minimization problem given by \cref{eq: minimization problem in SP :eq} is defined on a (weakly) closed convex subset of reflexive Hilbert space $\mathbb{L}^{2} ( \Omega )$,
and thus $\mathbb{Y} \cap \{ Y \in \mathbb{L}^{2} ( \Omega ): \mathbb{E} [ {Y}^{2} ] \le t \}$ for some $t$ is weakly compact (according to Kakutani's Theorem), 
we refer to the infinite-dimensional version of the Weierstrass Theorems (see, e.g., \cite[Theorems 2.3.4 and 2.3.5, p. 56]{Bobylev-Emelyanov-Korovin-1999}) to conclude the existence of solution.
\end{remark}

\subsection{Some examples}

Let us consider the three examples provided by \cite[Section 6]{Maccheroni-Marinacci-Rustichini-Taboga-2009}, where 
$\theta = 10$, $r = 1$ and other data are listed in the following \Cref{tab: three examples given by Maccheroni et al}.

\begin{table}[H]
    \centering
    \begin{tabular}{ccccc}
    \hline
              &                          & Example 1 ($n=1$) & Example 2 ($n=1$) & Example 3 ($n=2$)                 \\
                                           \cmidrule(r){3-3}   \cmidrule(r){4-4}   \cmidrule(r){5-5}
              & $\mathbb{P} ( {s}_{i} )$ & $R ( {s}_{i} )$   & $R ( {s}_{i} )$   & $( \vec{R} ( {s}_{i} ) )^{\top }$    \\
    \hline
    ${s}_{1}$ & $0.1$                    & $0.97$            & $0.97$            & $( 0.97, 1.05 )$ \\ 
    ${s}_{2}$ & $0.2$                    & $0.99$            & $0.99$            & $( 0.99, 1.00 )$ \\
    ${s}_{3}$ & $0.4$                    & $1.01$            & $1.01$            & $( 1.01, 0.99 )$ \\
    ${s}_{4}$ & $0.2$                    & $1.03$            & $1.03$            & $( 1.03, 0.99 )$ \\
    ${s}_{5}$ & $0.1$                    & $1.05$            & $1.10$            & $( 1.10, 0.99 )$ \\
    \multicolumn{2}{c}{$\vec{\alpha }^{*}_{mv} (1)$}
                                         & $2.08333$         & $1.35747$        & $( 1.66140, 1.04952 )$ \\
    \multicolumn{2}{c}{${y}^{\S}$}       & $1/6$             & $- 0.153846$     & $- 0.328228$ \\
    \hline
    \end{tabular}
    \caption{\centering Examples given by \protect\cite[Section 6]{Maccheroni-Marinacci-Rustichini-Taboga-2009}, \\
             where ${y}^{\S} := \inf_{i} \{ 1 - \langle \vec{R} ( {s}_{i} ) - \mathbb{E} [ \vec{R} ], ( \Var [ \vec{R} ] )^{-1} \mathbb{E} [ \vec{R} - r \vec{1} ] \rangle \}$.}
    \label{tab: three examples given by Maccheroni et al}
\end{table}

In Example 1, ${y}^{\S} = 1/6$ means that the optimal strategies for SMMV and MV preferences coincide if one chooses the constant $\zeta ( \cdot ) \equiv \zeta \le 1/6$.
This coincidence disappears in Examples 2 and 3, since ${y}^{\S} < 0$ therein.
Indeed, in Example 1, ${\alpha }^{*} = {\alpha }^{*}_{mv} (1)$ if $\zeta ( {s}_{5} ) \le 1/6$ and $\zeta ( {s}_{4} ) \le 7/12$, as illustrated in \Cref{tab: Example 1 with random zeta -a}.
However, as we merely change $\zeta ( {s}_{5} )$ to $1/5$, the result becomes significantly different, e.g., the optimal investment ratio is greater than that for the constant $\zeta \equiv 0.5$;
see \Cref{tab: Example 1 with random zeta -b,tab: Example 1 with deterministic zeta -d}.
Moreover, these change makes the set $\{ {X}_{ {\alpha }^{*} } \le {\lambda }_{ {\alpha }^{*} } - \frac{\zeta }{\theta } \}$, 
where the optimal wealth ${X}_{ {\alpha }^{*} }$ falls below the truncation level ${\lambda }_{ {\alpha }^{*} } - \frac{\zeta }{\theta }$ (as mentioned in \Cref{subsec: Properties of SMMV preference}),
shrink from $\Omega $ to $\{ {s}_{1}, {s}_{2} \}$.

\begin{table}[H]
    \centering
    \begin{subtable}[t]{0.495\linewidth}\centering
    \begin{tabular}{ccccr}
    \hline 
              & $\zeta $ & ${Y}_{*}$ & ${X}_{ {\alpha }^{*} }$ & ${\lambda }_{ {\alpha }^{*} } - \frac{\zeta }{\theta }$ \\
    \hline
    ${s}_{1}$ & $1$      & $11/6$    & $0.93750$               & $1.02083$ \\ 
    ${s}_{2}$ & $1$      & $17/12$   & $0.97917$               & $1.02083$ \\
    ${s}_{3}$ & $1$      & $1$       & $1.02083$               & $1.02083$ \\
    ${s}_{4}$ & $7/12$   & $7/12$    & $1.06250$               & $1.06250$ \\
    ${s}_{5}$ & $1/6$    & $1/6$     & $1.10417$               & $1.10417$ \\
    \multicolumn{3}{l}{${\alpha }^{*} = 2.08333$} & 
    \multicolumn{2}{r}{${\lambda }_{ {\alpha }^{*} } = 1.12083$} \\
    \hline
    \end{tabular}
    \caption{${\alpha }^{*} = {\alpha }^{*}_{mv} (1)$}
    \label{tab: Example 1 with random zeta -a}
    \end{subtable}
    \begin{subtable}[t]{0.495\linewidth}\centering
    \begin{tabular}{ccccr}
    \hline 
              & $\zeta $ & ${Y}_{*}$ & ${X}_{ {\alpha }^{*} }$ & ${\lambda }_{ {\alpha }^{*} } - \frac{\zeta }{\theta }$ \\
    \hline
    ${s}_{1}$ & $1$      & $29/15$   & $0.91250$               & $1.00583$ \\ 
    ${s}_{2}$ & $1$      & $27/20$   & $0.97083$               & $1.00583$ \\ 
    ${s}_{3}$ & $1$      & $1$       & $1.02917$               & $1.00583$ \\ 
    ${s}_{4}$ & $7/12$   & $7/12$    & $1.08750$               & $1.04750$ \\ 
    ${s}_{5}$ & $1/5$    & $1/5$     & $1.14583$               & $1.08583$ \\ 
    \multicolumn{3}{l}{${\alpha }^{*} = 2.91667$} & 
    \multicolumn{2}{r}{${\lambda }_{ {\alpha }^{*} } = 1.10583$} \\
    \hline
    \end{tabular}    
    \caption{${\alpha }^{*} \ne {\alpha }^{*}_{mv} (1)$}
    \label{tab: Example 1 with random zeta -b}
    \end{subtable}
    \caption{Example 1 with random $\zeta $.}
    \label{tab: Example 1 with random zeta}
\end{table}

Using the linear system \cref{eq: linear system given by KKT :eq}, we can easily make numerical experiments with different deterministic $\zeta $.
The results are summarized in \Cref{tab: Example 1 with deterministic zeta}.
For $\zeta \le 1/6$, the results are the same as those in \Cref{tab: Example 1 with deterministic zeta -a}, as previously noted.
For $\zeta > 0.75$, no solution exists.
As $\zeta $ increases, two clear monotonic trends emerge: The SMMV optimal investment ratio increases, while the truncation level decreases. These trends can be explained by the following mechanism:
Larger values of $\zeta $ place greater weight on extreme payoffs, since the marginal utility per unit of such outcomes is at least $\zeta $. 
Consequently, the SMMV optimal investment strategy becomes increasingly aggressive with rising $\zeta $, leading to a greater dispersion in terminal wealth ${X}_{ {\alpha }^{*} }$. 
This, in turn, causes a decline in the truncation level ${\lambda }_{ {\alpha }^{*} } - \frac{\zeta }{\theta }$.

\begin{table}[H]
    \centering
    \begin{subtable}[t]{0.495\linewidth}\centering
    \begin{tabular}{cclr}
    \hline 
              & ${Y}_{*}$ & $\beta $ & ${X}_{ {\alpha }^{*} }$ \\
    \hline
    ${s}_{1}$ & $11/6$    & $0$      & $0.93750$ \\ 
    ${s}_{2}$ & $17/12$   & $0$      & $0.97917$ \\
    ${s}_{3}$ & $1$       & $0$      & $1.02083$ \\
    ${s}_{4}$ & $7/12$    & $0$      & $1.06250$ \\
    ${s}_{5}$ & $1/6$     & $0$      & $1.10417$ \\
    \multicolumn{2}{l}{${\alpha }^{*} = 2.08333$} & 
    \multicolumn{2}{r}{${\lambda }_{ {\alpha }^{*} } - \frac{\zeta }{\theta } \equiv 1.10417$} \\
    \hline
    \end{tabular}
    \caption{Constant $\zeta \equiv 1/6$}
    \label{tab: Example 1 with deterministic zeta -a}
    \end{subtable}
    \begin{subtable}[t]{0.495\linewidth}\centering
    \begin{tabular}{cclr}
    \hline 
              & ${Y}_{*}$ & $\beta $  & ${X}_{ {\alpha }^{*} }$ \\
    \hline
    ${s}_{1}$ & $1.84706$ & $0$       & $0.93603$ \\ 
    ${s}_{2}$ & $1.42059$ & $0$       & $0.97868$ \\
    ${s}_{3}$ & $0.99412$ & $0$       & $1.02132$ \\
    ${s}_{4}$ & $0.56765$ & $0$       & $1.06397$ \\
    ${s}_{5}$ & $0.2$     & $0.05882$ & $1.10662$ \\
    \multicolumn{2}{l}{${\alpha }^{*} = 2.13235$} &
    \multicolumn{2}{r}{${\lambda }_{ {\alpha }^{*} } - \frac{\zeta }{\theta } \equiv 1.10074$} \\
    \hline
    \end{tabular}
    \caption{$\zeta \equiv 0.2$}
    \end{subtable} \\ 
    \vspace{1em}
    \begin{subtable}[t]{0.495\linewidth}\centering
    \begin{tabular}{cclr}
    \hline 
              & ${Y}_{*}$ & $\beta $  & ${X}_{ {\alpha }^{*} }$ \\
    \hline
    ${s}_{1}$ & $1.95$    & $0$       & $0.925$ \\ 
    ${s}_{2}$ & $1.45$    & $0$       & $0.975$ \\
    ${s}_{3}$ & $0.95$    & $0$       & $1.025$ \\
    ${s}_{4}$ & $0.45$    & $0$       & $1.075$ \\
    ${s}_{5}$ & $0.45$    & $0.5$     & $1.125$ \\
    \multicolumn{2}{l}{${\alpha }^{*} = 2.5$} &
    \multicolumn{2}{r}{${\lambda }_{ {\alpha }^{*} } - \frac{\zeta }{\theta } \equiv 1.075$} \\
    \hline
    \end{tabular}
    \caption{$\zeta \equiv 0.45$}
    \end{subtable} 
    \begin{subtable}[t]{0.495\linewidth}\centering
    \begin{tabular}{cclr}
    \hline 
              & ${Y}_{*}$ & $\beta $  & ${X}_{ {\alpha }^{*} }$ \\
    \hline
    ${s}_{1}$ & $2.03846$ & $0$       & $0.91346$ \\ 
    ${s}_{2}$ & $1.46154$ & $0$       & $0.97115$ \\
    ${s}_{3}$ & $0.88462$ & $0$       & $1.02885$ \\
    ${s}_{4}$ & $0.5$     & $0.19231$ & $1.08654$ \\
    ${s}_{5}$ & $0.5$     & $0.76923$ & $1.14423$ \\
    \multicolumn{2}{l}{${\alpha }^{*} = 2.88462$} &
    \multicolumn{2}{r}{${\lambda }_{ {\alpha }^{*} } - \frac{\zeta }{\theta } \equiv 1.06731$} \\
    \hline
    \end{tabular}
    \caption{$\zeta \equiv 0.5$}
    \label{tab: Example 1 with deterministic zeta -d}
    \end{subtable} \\
    \vspace{1em}
    \begin{subtable}[t]{0.495\linewidth}\centering
    \begin{tabular}{cclr}
    \hline 
              & ${Y}_{*}$ & $\beta $  & ${X}_{ {\alpha }^{*} }$ \\
    \hline
    ${s}_{1}$ & $2.21538$ & $0$       & $0.89038$ \\ 
    ${s}_{2}$ & $1.48462$ & $0$       & $0.96346$ \\
    ${s}_{3}$ & $0.75385$ & $0$       & $1.03654$ \\
    ${s}_{4}$ & $0.6$     & $0.57692$ & $1.10962$ \\
    ${s}_{5}$ & $0.6$     & $1.30769$ & $1.18269$ \\
    \multicolumn{2}{l}{${\alpha }^{*} = 3.65385$} &
    \multicolumn{2}{r}{${\lambda }_{ {\alpha }^{*} } - \frac{\zeta }{\theta } \equiv 1.05192$} \\
    \hline
    \end{tabular}
    \caption{$\zeta \equiv 0.6$}
    \end{subtable}
    \begin{subtable}[t]{0.495\linewidth}\centering
    \begin{tabular}{cclr}
    \hline 
              & ${Y}_{*}$ & $\beta $  & ${X}_{ {\alpha }^{*} }$ \\
    \hline
    ${s}_{1}$ & $7/3$     & $0$       & $0.87500$ \\ 
    ${s}_{2}$ & $3/2$     & $0$       & $0.95833$ \\
    ${s}_{3}$ & $2/3$     & $0$       & $1.04167$ \\
    ${s}_{4}$ & $2/3$     & $5/6$     & $1.12500$ \\
    ${s}_{5}$ & $2/3$     & $5/3$     & $1.20833$ \\
    \multicolumn{2}{l}{${\alpha }^{*} = 4.16667$} &
    \multicolumn{2}{r}{${\lambda }_{ {\alpha }^{*} } - \frac{\zeta }{\theta } \equiv 1.04167$} \\
    \hline
    \end{tabular}
    \caption{$\zeta \equiv 2/3$}
    \end{subtable} \\
    \vspace{1em}
    \begin{subtable}[t]{0.495\linewidth}\centering
    \begin{tabular}{cclr}
    \hline 
              & ${Y}_{*}$ & $\beta $  & ${X}_{ {\alpha }^{*} }$ \\
    \hline
    ${s}_{1}$ & $2.7$     & $0$       & $0.775$ \\ 
    ${s}_{2}$ & $1.2$     & $0$       & $0.925$ \\
    ${s}_{3}$ & $0.7$     & $1$       & $1.075$ \\
    ${s}_{4}$ & $0.7$     & $2.5$     & $1.225$ \\
    ${s}_{5}$ & $0.7$     & $4$       & $1.375$ \\
    \multicolumn{2}{l}{${\alpha }^{*} = 7.5$} &
    \multicolumn{2}{r}{${\lambda }_{ {\alpha }^{*} } - \frac{\zeta }{\theta } \equiv 0.975$} \\
    \hline
    \end{tabular}
    \caption{$\zeta \equiv 0.7$}
    \end{subtable}
    \begin{subtable}[t]{0.495\linewidth}\centering
    \begin{tabular}{cclr}
    \hline 
              & ${Y}_{*}$ & $\beta $  & ${X}_{ {\alpha }^{*} }$ \\
    \hline
    ${s}_{1}$ & $3.25$    & $0$       & $0.625$ \\ 
    ${s}_{2}$ & $0.75$    & $0$       & $0.875$ \\
    ${s}_{3}$ & $0.75$    & $2.5$     & $1.125$ \\
    ${s}_{4}$ & $0.75$    & $5$       & $1.375$ \\
    ${s}_{5}$ & $0.75$    & $7.5$     & $1.625$ \\
    \multicolumn{2}{l}{${\alpha }^{*} = 12.5$} &
    \multicolumn{2}{r}{${\lambda }_{ {\alpha }^{*} } - \frac{\zeta }{\theta } \equiv 0.875$} \\
    \hline
    \end{tabular}
    \caption{$\zeta \equiv 0.75$}
    \end{subtable}
    \caption{Example 1 with several deterministic $\zeta $.}
    \label{tab: Example 1 with deterministic zeta}
\end{table}

Another interesting observation is that for $\zeta = 1/6, 0.45, 2/3, 0.75$, there exists a state of nature $\omega $ such that ${X}_{ {\alpha }^{*} } ( \omega ) + \frac{\zeta }{\theta } = {\lambda }_{ {\alpha }^{*} }$.
For these specific $\zeta $, since ${X}_{ {\alpha }^{*} } + \frac{\zeta }{\theta } \le {\lambda }_{ {\alpha }^{*} }$ is equivalent to ${Y}^{*} \ge \zeta $,
one can mirror the proof of \Cref{thm: Lagrange multiplier as solution} to obtain
\begin{equation*}
{\alpha }^{*} = \frac{ \zeta \mathbb{E} [R] + ( 1 - \zeta ) \mathbb{E} [ R | {Y}_{*} \ge \zeta ] - r }{\theta \mathbb{P} ( {Y}_{*} \ge \zeta ) \Var [ R | {Y}_{*} \ge \zeta ] }.
\end{equation*}
In contrast, the expression for ${\alpha }^{*}$ in \Cref{thm: Lagrange multiplier as solution} is not applicable when $\zeta \equiv 0.75$, as $\Var [ R | {Y}_{*} > \zeta ] = 0$.

\begin{table}[H]
    \centering
    \begin{tabular}{lccccccc}
    \hline
    $R ( {s}_{5} )$          & $1.05$   & $1.10$   & $1.15$   & $1.20$   & $1.50$   & $2$      & $3$  \\
    \hline
    ${\alpha }^{*}_{mv} (1)$ & $2.0833$ & $1.3574$ & $0.9174$ & $0.6747$ & $0.2465$ & $0.1175$ & $0.0572$ \\
    $\zeta = 0$              & $2.0833$ & $1.8382$ & $1.8382$ & $1.8382$ & $1.8382$ & $1.8382$ & $1.8382$ \\
    $\zeta = 0.01$           & $2.0833$ & $1.8695$ & $1.8860$ & $1.9026$ & $2.0018$ & $2.1673$ & $2.4982$ \\
    $\zeta = 0.02$           & $2.0833$ & $1.9007$ & $1.9338$ & $1.9669$ & $2.1654$ & $2.4963$ & $3.1581$ \\
    $\zeta = 0.05$           & $2.0833$ & $1.9945$ & $2.0772$ & $2.1599$ & $2.6563$ & $3.4835$ & $5.9856$ \\
    $\zeta = 0.1$            & $2.0833$ & $2.1507$ & $2.3162$ & $2.4816$ & $3.4743$ & $6.2019$ & $20.625$ \\
    $\zeta = 1/6$            & $2.0833$ & $2.3591$ & $2.6348$ & $2.9105$ & $5.3686$ & $13.542$ & ---      \\ 
    $\zeta = 0.2$            & $2.1323$ & $2.4632$ & $2.7941$ & $3.1250$ & $6.6346$ & $28.750$ & ---      \\
    $\zeta = 0.45$           & $2.5000$ & $4.0144$ & $5.5288$ & $7.8125$ & ---      & ---      & ---      \\
    $\zeta = 0.5$            & $2.8846$ & $4.5673$ & $6.2500$ & $15.625$ & ---      & ---      & ---      \\
    $\zeta = 0.6$            & $3.6539$ & $8.7500$ & $20.000$ & ---      & ---      & ---      & ---      \\
    $\zeta = 2/3$            & $4.1667$ & $16.667$ & ---      & ---      & ---      & ---      & ---      \\
    $\zeta = 0.7$            & $7.5000$ & ---      & ---      & ---      & ---      & ---      & ---      \\
    $\zeta = 0.75$           & $12.500$ & ---      & ---      & ---      & ---      & ---      & ---      \\
    \multicolumn{8}{l}{---: there exists no solution.} \\
    \hline
    \end{tabular} 
    \caption{\centering SMMV optimal investment ratios with different $R ( {s}_{5} )$ and deterministic $\zeta $, \\
             where the second and the third columns correspond to Examples 1 and 2, respectively.}
    \label{tab: Examples 1 and 2}
\end{table}

Next, we increase the payoff $R ( {s}_{5} )$ to include the Examples 1 and 2 with different deterministic values of $\zeta $, and present the SMMV optimal investment ratios in \Cref{tab: Examples 1 and 2}.
For ease of comparison, we also display the result for ${\alpha }_{mv} (1)$ and $\zeta = 0$ (i.e.,  ${\alpha }_{mmv} (1)$ in \cite[Section 6]{Maccheroni-Marinacci-Rustichini-Taboga-2009}) in \Cref{tab: Examples 1 and 2}.
Interestingly, in the vast majority of cases, the SMMV optimal investment ratio is increasing in both the chosen $\zeta $ and the payoff $R ( {s}_{5} )$. 
However, this monotonicity in the payoff $R ( {s}_{5} )$ does not necessarily hold when $R ( {s}_{5} )$ increases from $1.05$ to $1.10$.
For instance, when $\zeta = 0.02$, ${\alpha }^{*}$ has a sudden jump from $2.0833$ to $1.8766$ at $R ( {s}_{5} ) \approx 1.0636$.
This is because ${Y}_{*}$ switches from the interior solution for \cref{eq: minimization problem in SP :eq} to a corner solution, which adds one more equation regarding ${s}_{5}$ to \cref{eq: linear system given by KKT :eq}.
For $R ( {s}_{5} ) \ge 1.0636$, the aforementioned monotonicity in $R ( {s}_{5} )$ holds. 

\begin{table}[htbp!]
    \centering
    \begin{tabular}{ccccccccc}
    \hline
    $\zeta ( \cdot ) \equiv \zeta $ & $0$      & $0.01$   & $0.02$   & $0.05$   & $0.1$    & $0.2$    & $0.3$    & $0.4$ \\
    \hline
    ${\alpha }^{*}_{1}$             & $4.2980$ & $4.3793$ & $4.4597$ & $4.7007$ & $5.1025$ & $9.2422$ & $14.840$ & $20.438$ \\
    ${\alpha }^{*}_{2}$             & $3.0423$ & $3.1030$ & $3.1638$ & $3.3459$ & $3.6495$ & $6.2969$ & $9.8359$ & $13.375$ \\  
    ${\lambda }_{ \vec{\alpha }^{*} } - \frac{\zeta }{\theta }$ 
                                    & $1.1316$ & $1.1308$ & $1.1301$ & $1.1280$ & $1.1243$ & $1.1101$ & $1.0929$ & $1.0756$ \\ 
    \hline
    \end{tabular} 
    \caption{SMMV optimal investment ratio in Example 3 with different deterministic $\zeta $.}
    \label{tab: Example 3}
\end{table}

Finally, we extend the discussion to investigate a portfolio selection problem involving a second risky asset, whose SMMV optimal investment ratio is denoted by ${\alpha }^{*}_{2}$. 
In this context, we seek to determine $\vec{\alpha }^{*} = ( {\alpha }^{*}_{1}, {\alpha }^{*}_{2} )$ for Example 3 across various deterministic values of $\zeta $, and summarize the results in \Cref{tab: Example 3}.
Notably, for the constant $\zeta \ge 0.4133$ in this example, there exists no solution. 
In this example, it is obvious that as $\zeta $ increases, the SMMV optimal investment ratio for each asset increases, and the truncation level decreases.
Indeed, since ${\alpha }^{*}_{1} > {\alpha }^{*}_{2}$ --- indicating that the investment in the first asset plays a dominant role as shown in \Cref{tab: Example 3} --- 
the mechanism described in Example 1 remains valid to explain the two observations for monotonicity.

\section{Continuous-time dynamic portfolio management}
\label{sec: Continuous-time dynamic portfolio management}

\subsection{Model and problem formulation}

In this section, we study the portfolio selection problem in a continuous-time stochastic control framework with a preassigned finite time-horizon $T$ and the SMMV preference ${V}_{ \theta, \zeta }$.
Let us formulate this problem on the complete filtered probability basis $( \Omega, \mathcal{F}, \mathbb{F}, \mathbb{P} )$,
where $\mathbb{F} := \{ \mathcal{F}_{t} \}_{ t \in [ 0,T ] }$ is the right-continuous, completed natural filtration generated by a one-dimensional standard Brownian motion $\{ {W}_{t} \}_{ t \in [ 0,T ] }$. 
For the conditional expectation operator, we write $\mathbb{E}_{t} [ \cdot ] := \mathbb{E} [ \cdot | \mathcal{F}_{t} ]$ for brevity.
Without any loss of generality, we assume that $\mathcal{F}_{0} = \{ \varnothing, \Omega \}$ and $\mathcal{F}_{T} = \mathcal{F}$. 
For ease of reference, below we list the notation of spaces that are used in the sequel.
\begin{itemize}
\item $\mathbb{L}^{2}_{\mathbb{F}} ( 0,T; \mathbb{L}^{2} ( \Omega ) )$ denotes the set of all $\mathbb{F}$-progressively measurable processes $f: [ 0,T ] \times \Omega \to \mathbb{R}$ such that 
      $\mathbb{E} [ \int_{0}^{T} | f ( t, \cdot ) |^{2} dt ] < \infty $.
\item ${C}_{\mathbb{F}} ( 0,T; \mathbb{L}^{2} ( \Omega ) )$ denotes the set of all $\mathbb{P}$-a.s. continuous processes $f \in \mathbb{L}^{2}_{\mathbb{F}} ( 0,T; \mathbb{L}^{2} ( \Omega ) )$ such that 
      $\mathbb{E} [ \sup_{ t \in [ 0,T ] } | f ( t, \cdot ) |^{2} ] < \infty $.
\item $\mathbb{L}^{2}_{\mathbb{F}} ( 0,T ; \mathbb{L}^{2} ( \Omega; {C}^{p,q} ( \mathbb{R} \times [ 0, + \infty ) ) ) )$ denotes the set of all random fields $f: [ 0,T ] \times \Omega \times \mathbb{R} \times [ 0, + \infty ) \to \mathbb{R}$ 
      such that $f ( \cdot, \cdot, x, z ) \in \mathbb{L}^{2} ( 0,T; \mathbb{L}^{2} ( \Omega ) )$ and $f ( t, \omega, x, z )$ is $p$ (resp. $q$) times continuously differentiable in $x$ (resp. $z$) for every $t$ and $\mathbb{P}$-a.e. $\omega $.
\item ${C}_{\mathbb{F}} ( 0,T; \mathbb{L}^{2} ( \Omega; {C}^{p,q} ( \mathbb{R} \times [ 0, + \infty ) ) ) )$ denotes the set of all 
      $f \in \mathbb{L}^{2}_{\mathbb{F}} ( 0,T ; \mathbb{L}^{2} ( \Omega; {C}^{p,q} ( \mathbb{R} \times [ 0, + \infty ) ) ) )$ such that $f ( \cdot, \cdot, x,z ) \in {C}_{\mathbb{F}} ( 0,T; \mathbb{L}^{2} ( \Omega ) )$.
\end{itemize}
To keep the exposition concise, hereafter we omit the statement of sample path $\omega $, ``a.e. $t$'' and ``$\mathbb{P}$-a.s.'', 
and write ${f}_{t} = f ( t, \omega )$, $f ( t,x ) = f ( t, \omega, x )$ and $f ( t,x,z ) = f ( t, \omega, x,z )$, unless otherwise mentioned.
  
Referring to \cite[Section 6.8]{Yong-Zhou-1999} (see also \cite{Zhou-Li-2000}) and \cite[Section 4.5.1]{Shreve-2004}, 
we consider a simple Black-Scholes market, which includes a bond and a stock with random parameters. 
For every epoch $t$, let ${r}_{t} \in \mathbb{R}$ denote the instantaneous yield rate of the bond,
and $( {\sigma }_{t}, {\vartheta }_{t} ) \in \mathbb{R}_{+} \times \mathbb{R}_{+}$ respectively denote the volatility rate and the market price of risk for the stock.
Thus, the price processes of these two assets satisfy the following stochastic differential equations (SDEs):
\begin{equation*}
\left\{ \begin{aligned}
& d {B}_{t} = {B}_{t} {r}_{t} dt, && {B}_{0} > 0 \quad \text{(for the bond)}; \\
& d {S}_{t} = {S}_{t} {r}_{t} dt + {S}_{t} {\sigma }_{t} ( d {W}_{t} + {\vartheta }_{t} dt ), && {S}_{0} > 0 \quad \text{(for the stock)}.
\end{aligned} \right.
\end{equation*}
Then, it is supposed to maximize ${V}_{ \theta, \zeta } ( {X}^{\pi }_{T} )$ by choosing an appropriate dynamic stock investment amount process $\pi \in \mathbb{L}^{2}_{\mathbb{F}} ( 0,T; \mathbb{L}^{2} ( \Omega ) )$, 
subject to the following controlled SDE for the wealth process ${X}^{\pi }$:
\begin{equation}\label{eq: wealth dynamics :eq}
d {X}^{\pi }_{t} = ( {X}^{\pi }_{t} - {\pi }_{t} ) \frac{ d {B}_{t} }{ {B}_{t} } + {\pi }_{t} \frac{ d {S}_{t} }{ {S}_{t} }
                 = {X}^{\pi }_{t} {r}_{t} dt + {\pi }_{t} {\sigma }_{t} ( d {W}_{t} + {\vartheta }_{t} dt ), \quad {X}^{\pi }_{0} = {x}_{0}.
\end{equation}

Before proceeding, we should note that the aforementioned financial market is complete and has a unique risk-neutral probability measurable $\mathbb{Q}$ given by the Radon-Nikod\'ym derivative 
\begin{equation*}
\frac{ d \mathbb{Q} }{ d \mathbb{P} } \Big|_{\mathcal{F}_{T}} = {\Lambda }_{T} := {e}^{ - \int_{0}^{T} {\vartheta }_{v} d {W}_{v} - \frac{1}{2} \int_{0}^{T} | {\vartheta }_{v} |^{2} dv }.
\end{equation*}
Let ${\Lambda }_{t} := \mathbb{E}_{t} [ {\Lambda }_{T} ]$ and ${D}_{t} := {\Lambda }_{t} \exp ( - \int_{0}^{t} {r}_{v} dv )$ (i.e. the so-called state-price density process).
In addition, we assume that $r, \sigma, \vartheta \in {C}_{\mathbb{F}} ( 0,T; \mathbb{L}^{2} ( \Omega ) )$ are essentially bounded, 
and introduce the following assumption, for which the economic insights can be found in the next subsection.

\begin{assumption}\label{ass: SMMV-MV comparison}
$\zeta \le {D}_{T} / \mathbb{E} [ {D}_{T} ]$, $\mathbb{P}$-a.s.
\end{assumption}

Due to the presence of $\zeta \in \mathbb{L}^{2} ( \Omega )$, the method in \cite{Strub-Li-2020}, which shows that the optimal strategies for MMV and MV preferences coincide, is no longer applicable to compare the optimal strategies for SMMV and MV preferences.
For example, this method employs the embedding technique in \cite{Li-Ng-2000,Zhou-Li-2000}, reducing the MV problem to minimizing $\mathbb{E} [ ( {X}^{\pi }_{T} - \lambda )^{2} ]$,
where $\lambda = \frac{1}{\theta } + \mathbb{E} [ {X}^{*}_{T} ]$ with ${X}^{*}$ corresponding to the MV optimal strategy ${\pi }^{*}$. 
It is obvious that ${X}^{*}_{T} \in \mathcal{G}_{\theta }$ if and only if ${X}^{*}_{T} \le \lambda $.
However, ${X}^{*}_{T} \le \lambda - \frac{\zeta }{\theta }$ is needed for ${X}^{*}_{T} \in \mathcal{G}_{\theta, \zeta }$, 
and the ``target-anchoring'' strategy in the proof of \cite[Theorem 1]{Strub-Li-2020} cannot realize such a target ${X}^{\pi }_{T} = \lambda - \frac{\zeta }{\theta }$, $\mathbb{P}$-a.s., for random $\zeta $.
To resolve this issue, we proceed with the following proposition, 
which does not only show the equivalence between ${X}^{*}_{T} \in \mathcal{G}_{\theta, \zeta }$ and \Cref{ass: SMMV-MV comparison},
but also set a stepping stone to validate the coincidence among the SMMV, MMV and MV optimal strategies.

\begin{proposition}\label{lem: MV optimal portfolio in monotonicity domain}
${X}^{*}_{T} = \mathbb{E} [ {X}^{*}_{T} ] + \frac{1}{\theta } ( 1 - \frac{ {D}_{T} }{ \mathbb{E} [ {D}_{T} ] } )$ realizes the maximum for the MV problem formulated by 
${U}_{\theta } ( {X}^{*}_{T} ) = \max_{ \pi \in \mathbb{L}^{2}_{\mathbb{F}} ( 0,T; \mathbb{L}^{2} ( \Omega ) ) } {U}_{\theta } ( {X}^{\pi }_{T} )$.
Consequently, ${X}^{*}_{T} \in \mathcal{G}_{ \theta, \zeta }$ if and only if \Cref{ass: SMMV-MV comparison} holds. 
\end{proposition}

\begin{proof}
See \Cref{pf-lem: MV optimal portfolio in monotonicity domain}.
Notably, to prove this \Cref{lem: MV optimal portfolio in monotonicity domain},
we conduct a convex duality analysis of the MV problem under a slightly general framework in \Cref{app: Convex duality analysis to the MV problem},
which is also applicable for MV problems in an incomplete market with partial information.
\end{proof}

Now we state the main result of this subsection, showing that the SMMV, MMV and MV optimal strategies coincide.

\begin{theorem}\label{thm: identical solution of SMMV and MV}
Suppose that \Cref{ass: SMMV-MV comparison} holds.
Then, the MV optimal strategy ${\pi }^{*}$ maximizes ${V}_{ \theta, \zeta } ( {X}^{\pi }_{T} )$ over all $\pi \in \mathbb{L}^{2}_{\mathbb{F}} ( 0,T; \mathbb{L}^{2} ( \Omega ) )$, 
implying ${V}_{ \theta, \zeta } ( {X}^{*}_{T} ) = {U}_{\theta } ( {X}^{*}_{T} )$.
\end{theorem}

\begin{proof}
See \Cref{pf-thm: identical solution of SMMV and MV}.
\end{proof}

Our results in the continuous-time setting are established under \Cref{ass: SMMV-MV comparison}.
However, analogous to the aforementioned single-period situations where $\mathbb{P} ( \zeta \le d \mathbb{Q} / d \mathbb{P} ) < 1$ and no solution exists (see \Cref{rem: no static solution}), 
the SMMV dynamic portfolio problems become ill-posed when \Cref{ass: SMMV-MV comparison} does not hold. 
This will be demonstrated in the following subsections.

\subsection{Microeconomics analysis for \texorpdfstring{\Cref{ass: SMMV-MV comparison}}{Assumption 4.1}: single-period binomial framework}

In this subsection, we employ a single-period binomial framework, with a slight abuse of notation, to intuitively show that the solution does not exist unless \Cref{ass: SMMV-MV comparison} holds.
Before proceeding, we should note that the SMMV portfolio selection problem can be reduced to
maximizing ${V}_{ \theta, \zeta } ( {X}^{\pi }_{T} )$ subject to ${X}^{\pi }_{T} \in \mathbb{L}^{2} ( \Omega )$ with $\mathbb{E} [ {X}^{\pi }_{T} {D}_{T} ] = {x}_{0}$,
analogous to the proof of \Cref{lem: MV optimal portfolio in monotonicity domain} in \Cref{pf-lem: MV optimal portfolio in monotonicity domain}.
In other words, the portfolio selection problem is equivalent to the terminal state choosing problem,
provided that the market is complete.

Let us consider the outcomes of a single coin toss (as in \Cref{exp: illustrate MV MMV SMMV}) with $\mathbb{P} (head) = {p}_{1}$ and $\mathbb{P} (tail) = {p}_{2}$ under the real-world probability measure.
Without loss of generality, ${p}_{1} \ne {p}_{2}$ is allowed. 
Suppose that the wealth starts at the initial value ${x}_{0}$, and then becomes $X(head) = {x}_{1}$ or $X(tail) = {x}_{2}$ at the maturity.
In addition, we write $\zeta ( head ) = {\zeta }_{1}$ and $\zeta ( tail ) = {\zeta }_{2}$, and denote the accumulating factor by ${e}^{r}$, where $r$ is also a random variable with $r (head) = {r}_{1}$ and $r (tail) = {r}_{2}$.
Then, one can find the risk-neutral probability measure $\mathbb{Q}$, under which $\mathbb{Q} (head) = {q}_{1}$ and $\mathbb{Q} (tail) = {q}_{2}$, 
by $1 = {q}_{1} + {q}_{2}$ and ${x}_{0} = {x}_{1} {e}^{ - {r}_{1} } {q}_{1} + {x}_{2} {e}^{ - {r}_{2} } {q}_{2}$.
Consequently, \Cref{ass: SMMV-MV comparison} degenerates to
\begin{equation*}
{\zeta }_{1} \le \frac{ {e}^{ - {r}_{1} } \frac{ {q}_{1} }{ {p}_{1} } }{ {e}^{ - {r}_{1} } {q}_{1} + {e}^{ - {r}_{2} } {q}_{2} }, \quad
{\zeta }_{2} \le \frac{ {e}^{ - {r}_{2} } \frac{ {q}_{2} }{ {p}_{2} } }{ {e}^{ - {r}_{1} } {q}_{1} + {e}^{ - {r}_{2} } {q}_{2} }.
\end{equation*}

Since ${V}_{ \theta, \zeta } (X)$ can be re-expressed as a utility function $u ( {x}_{1}, {x}_{2} )$,
the SMMV problem can be formulated as
\begin{equation*}
\max_{ ( {x}_{1}, {x}_{2} ) \in \mathbb{R}^{2} } u ( {x}_{1}, {x}_{2} ), \quad s.t. \quad {x}_{1} {e}^{ - {r}_{1} } {q}_{1} + {x}_{2} {e}^{ - {r}_{2} } {q}_{2} = {x}_{0}.
\end{equation*}
This is an analogue to the consumer behavior model in microeconomics; 
that is, one is supposed to choose the amount $( {x}_{1}, {x}_{2} )$ for the ``goods'' $X ( head )$ and $X ( tail )$ with their ``prices'' ${e}^{ - {r}_{1} } {q}_{1}$ and ${e}^{ - {r}_{2} } {q}_{2}$,
subject to the budget constraint ${x}_{0}$, to maximize the utility $u ( {x}_{1}, {x}_{2} )$.
So the equi-marginal principle suggests that at the maximum point,
\begin{equation*}
\frac{ \partial u / \partial {x}_{1} }{ {e}^{ - {r}_{1} } {q}_{1} } = \frac{ \partial u / \partial {x}_{2} }{ {e}^{ - {r}_{2} } {q}_{2} }.
\end{equation*}

Now, corresponding to the case that \Cref{ass: SMMV-MV comparison} is violated, we assume, without loss of generality, that 
${\zeta }_{1} > \frac{ {e}^{ - {r}_{1} } {q}_{1} / {p}_{1} }{ {e}^{ - {r}_{1} } {q}_{1} + {e}^{ - {r}_{2} } {q}_{2} }$, 
i.e. $\frac{ {\zeta }_{1} {p}_{1} }{ {e}^{ - {r}_{1} } {q}_{1} } > \frac{ 1 - {\zeta }_{1} {p}_{1} }{ {e}^{ - {r}_{2} } {q}_{2} }$.
Let us start at the point $( {x}_{1}, {x}_{2} )$ satisfying $( {x}_{1} + \frac{ {\zeta }_{1} }{\theta } - {x}_{2} - \frac{ {\zeta }_{2} }{\theta } ) {p}_{2} \ge \frac{1}{\theta }$.
By the second assertion in \Cref{lem: identity and truncation}, we know that 
\begin{equation*}
{\lambda }_{ \theta, \zeta, X } = {x}_{2} + \frac{ {\zeta }_{2} }{\theta } + \frac{ 1 - {\zeta }_{1} {p}_{1} - {\zeta }_{2} {p}_{2} }{ \theta {p}_{2} } 
                                = {x}_{2} + \frac{ 1 - {\zeta }_{1} {p}_{1} }{ \theta {p}_{2} },
\end{equation*}
and hence \Cref{thm: Gateaux differentiability and equivalent expressions of SMMV} gives
\begin{align*}
& \frac{ \partial u }{ \partial {x}_{1} } = \mathbb{E} [ {1}_{\{ head \}} d {V}_{ \theta, \zeta } (X) ] = {\zeta }_{1} {p}_{1}, \\
& \frac{ \partial u }{ \partial {x}_{2} } = \mathbb{E} [ {1}_{\{ tail \}} d {V}_{ \theta, \zeta } (X) ] = \theta ( {\lambda }_{ \theta, \zeta, X } - {x}_{2} ) {p}_{2} = 1 - {\zeta }_{1} {p}_{1}.
\end{align*}
Given $\frac{ {\zeta }_{1} {p}_{1} }{ {e}^{ - {r}_{1} } {q}_{1} } > \frac{ 1 - {\zeta }_{1} {p}_{1} }{ {e}^{ - {r}_{2} } {q}_{2} }$, it immediately follows that
\begin{equation*}
\frac{ \partial u / \partial {x}_{1} }{ {e}^{ - {r}_{1} } {q}_{1} } > \frac{ \partial u / \partial {x}_{2} }{ {e}^{ - {r}_{2} } {q}_{2} }, \quad i.e. \quad
\frac{ \partial u / \partial {x}_{1} }{ \partial u / \partial {x}_{2} } > \frac{ {e}^{ - {r}_{1} } {q}_{1} }{ {e}^{ - {r}_{2} } {q}_{2} }.
\end{equation*}
This implies that at all points $( {x}_{1}, {x}_{2} )$ satisfying $( {x}_{1} + \frac{ {\zeta }_{1} }{\theta } - {x}_{2} - \frac{ {\zeta }_{2} }{\theta } ) {p}_{2} \ge \frac{1}{\theta }$,
the marginal substitution rate of ``good $X(head)$'' for ``good $X(tail)$'' is always larger than their ``price'' ratio.
As a result, the rational agent is always willing to reduce the amount ${x}_{2}$ for every additional increase of ${x}_{1}$, which implies that any finite ${x}_{1}$ is not optimal.

In conclusion, when \Cref{ass: SMMV-MV comparison} is violated, the marginal utilities across certain states become equalized but excessively high, causing the optimal allocation to concentrate entirely on those states. Intuitively, this is analogous to the substitution between goods in microeconomics being replaced here by substitution between states of nature. For instance, between a ``good'' and a ``bad'' state. Increasing investment corresponds to earning more in the good state but losing more in the bad one. As long as the marginal gain in the good state always outweighs the marginal loss in the bad state, the investor's optimal decision collapses to an all-in position. 

Extending the model beyond this case would require a more sophisticated monotonicity adjustment mechanism, which lies beyond the scope of this paper. Our focus is to highlight the limitation of MMV preferences, propose a feasible and strictly monotone alternative, and clearly delineate the theoretical and practical domain where SMMV preferences are meaningful. The exploration of more general continuous-time extensions of SMMV is left for future research.

\subsection{Characterizations for the solution}

In the previous subsections, our method is indeed straightforward and differs from the most existing studies using some stochastic control techniques to address the continuous-time portfolio problems 
(including \cite{Trybula-Zawisza-2019,Shen-Zou-2022,Li-Guo-Tian-2023,Li-Guo-2021,Li-Liang-Pang-2022,Li-Liang-Pang-2023} for the MMV preference).
To facilitate comparison and establish potential frameworks, 
in the following three subsections we characterize the optimal SMMV strategy by employing the dynamic programming principle, the martingale convex duality analysis, and the embedding method, respectively.

\subsubsection{BSPDE characterization}

Due to \cref{eq: measure optimization form of SMMV}, the portfolio problem with the SMMV preference ${V}_{\theta, \zeta }$ is equivalent to finding a maximum point ${\pi }^{*}$ for 
\begin{equation}\label{eq: max-min problem for DP :eq}
\sup_{ \pi \in \mathbb{L}^{2}_{\mathbb{F}} ( 0,T; \mathbb{L}^{2} ( \Omega ) ) } \inf_{ Z \in \mathbb{L}^{2}_{+} ( \Omega ), \mathbb{E} [Z] = 1 }
\bigg\{ \mathbb{E} [ {X}^{\pi }_{T} \zeta ]
      + \kappa \mathbb{E} \Big[ \Big( {X}^{\pi }_{T} + \frac{\zeta }{\theta } \Big) Z \Big]
      + \frac{ {\kappa }^{2} }{ 2 \theta } \mathbb{E} [ {Z}^{2} ] \bigg\},
\end{equation}
subject to \cref{eq: wealth dynamics :eq}.
Inspired by \cite{Trybula-Zawisza-2019} and \cite{Li-Guo-2021}, in order to formulate the aforementioned max-min problem in the framework of stochastic control,
we denote by $( {X}^{ t,x, \pi }, {Z}^{ t,z, \gamma } )$ the $\mathbb{F}$-adapted solution of the following controlled SDEs corresponding to the control pair $( \pi, \gamma )$:
\begin{equation}\label{eq: controlled SDEs :eq}
\left\{ \begin{aligned}
& d {X}_{s} = {X}_{s} {r}_{s} ds + {\pi }_{s} {\sigma }_{s} ( d {W}_{s} + {\vartheta }_{s} ds ), \\
& d {Z}_{s} = {\gamma }_{s} d {W}_{s}, \\
& ( {X}_{t}, {Z}_{t} ) = ( x,z ) \in \mathbb{R} \times [ 0, + \infty ).
\end{aligned} \right.
\end{equation}
For ease of statement, we let ${\Gamma }^{t,z}$ denote the set of all the admissible $\gamma \in \mathbb{L}^{2}_{\mathbb{F}} ( 0,T; \mathbb{L}^{2} ( \Omega ) )$ such that ${Z}^{ t,z, \gamma }_{T} \in \mathbb{L}^{2}_{+} ( \Omega )$.
Notably, $Z = 0$ is the absorbing state for ${Z}^{ t,z, \gamma }$, as it is a non-negative continuous $( \mathbb{F}, \mathbb{P} )$-martingale.
As a consequence, any $\gamma \in {\Gamma }^{t,0}$ satisfies $\mathbb{E} [ \int_{t}^{T} | {\gamma }_{s} |^{2} ds ] = 0$. 
In addition, we introduce the objective functional
\begin{equation}\label{eq: objective functional for DP :eq}
{J}^{ \pi, \gamma } ( t,x,z ) := \mathbb{E}_{t} \bigg[ {X}^{ t,x, \pi }_{T} \zeta 
                                                     + \kappa \Big( {X}^{ t,x, \pi }_{T} + \frac{\zeta }{\theta } \Big) {Z}^{ t,z, \gamma }_{T} 
                                                     + \frac{ {\kappa }^{2} }{ 2 \theta } | {Z}^{ t,z, \gamma }_{T} |^{2} \bigg].
\end{equation}
Then, the max-min problem given by \cref{eq: max-min problem for DP :eq} is reduced to finding the maximizer ${\pi }^{*}$ for 
\begin{equation}\label{eq: control problem primal :eq}
\sup_{ \pi \in \mathbb{L}^{2}_{\mathbb{F}} ( 0,T; \mathbb{L}^{2} ( \Omega ) ) } \inf_{ \gamma \in {\Gamma }^{0,1} } {J}^{ \pi, \gamma } ( 0, {x}_{0}, 1 ).
\end{equation}

By the dynamic programming principle, we aim to find $( {\pi }^{*}, {\gamma }^{*} ) \in \mathbb{L}^{2}_{\mathbb{F}} ( 0,T; \mathbb{L}^{2} ( \Omega ) ) \times {\Gamma }^{0,1}$ as a saddle point 
such that for any $( \pi, \gamma ) \in \mathbb{L}^{2}_{\mathbb{F}} ( 0,T; \mathbb{L}^{2} ( \Omega ) ) \times {\Gamma }^{0,1}$,
\begin{equation}\label{eq: saddle point definition :eq}
    {J}^{ \pi, {\gamma }^{*} } ( t, {X}^{ 0, {x}_{0}, {\pi }^{*} }_{t}, {Z}^{ 0,1, {\gamma }^{*} }_{t} ) 
\le {J}^{ {\pi }^{*}, {\gamma }^{*} } ( t, {X}^{ 0, {x}_{0}, {\pi }^{*} }_{t}, {Z}^{ 0,1, {\gamma }^{*} }_{t} ) 
\le {J}^{ {\pi }^{*}, \gamma } ( t, {X}^{ 0, {x}_{0}, {\pi }^{*} }_{t}, {Z}^{ 0,1, {\gamma }^{*} }_{t} ).
\end{equation}
In fact, \cref{eq: saddle point definition :eq} formulates a sequence of stochastic differential games, 
where one player (e.g., an investor) aims to maximize ${J}^{ \pi, \gamma }$ with its strategy $\pi $ 
and the other player (e.g., an incarnation of the market) aims to minimize ${J}^{ \pi, \gamma }$ with its strategy $\gamma $ at almost every epoch $t$.
The saddle point $( {\pi }^{*}, {\gamma }^{*} ) \in \mathbb{L}^{2}_{\mathbb{F}} ( 0,T; \mathbb{L}^{2} ( \Omega ) ) \times {\Gamma }^{0,1}$ 
is also a Nash equilibrium of the stochastic differential game corresponding to $t$.
Thus, by the method of continuous embedding, one may generate the stochastic differential games as
\begin{equation}\label{eq: control problem for DP :eq}
\esssup_{ \pi \in \mathbb{L}^{2}_{\mathbb{F}} ( 0,T; \mathbb{L}^{2} ( \Omega ) ) } \essinf_{ \gamma \in {\Gamma }^{t,z} } {J}^{ \pi, \gamma } ( t,x,z )
\end{equation}
for all $( t,x,z ) \in [ 0,T ) \times \mathbb{R} \times [ 0, + \infty )$.

Due to the presence of the absorbing level $Z = 0$, we are supposed to consider the ``boundary'' problem formulated by
\begin{equation*}
  \esssup_{ \pi \in \mathbb{L}^{2}_{\mathbb{F}} ( 0,T; \mathbb{L}^{2} ( \Omega ) ) } {J}^{ \pi, \gamma } ( t,x,0 )
= \esssup_{ \pi \in \mathbb{L}^{2}_{\mathbb{F}} ( 0,T; \mathbb{L}^{2} ( \Omega ) ) } \mathbb{E}_{t} [ {X}^{ t,x, \pi }_{T} \zeta ].
\end{equation*}
In general, this problem is not necessarily well-posed, because the objective functional is affine in $\pi $.
To avoid the potential ill-posedness, we purpose to characterize the SMMV optimal strategy through investigating the relaxed stochastic differential games formulated by
\begin{equation}\label{eq: relaxed SDGs :eq}
\esssup_{ \pi \in \mathbb{L}^{2}_{\mathbb{F}} ( 0,T; \mathbb{L}^{2} ( \Omega ) ) } \essinf_{ \gamma \in \mathbb{L}^{2}_{\mathbb{F}} ( 0,T; \mathbb{L}^{2} ( \Omega ) ) } {J}^{ \pi, \gamma } ( t,x,z ),
\quad ( t,x,z ) \in [ 0,T ) \times \mathbb{R} \times \mathbb{R}.
\end{equation}
Notably, the potential ill-posedness does not appear in the continuous-time MMV portfolio problems (i.e. $\zeta \equiv 0$), 
because \Cref{thm: optimal MV solution} implies that ${Z}^{ t,z, {\gamma }^{*} }$ will not hit the threshold $Z = 0$.

For brevity, we introduce two infinitesimal operators for $( \pi, \gamma ) \in \mathbb{R}^{2}$ and $\mathbb{R}$-valued function $f ( t,x,z )$ twice continuously differentiable in $( x,z )$ as follows:
\begin{align*}
    \mathcal{D}_{1}^{ \pi, \gamma } f ( t,x,z )
& = {f}_{x} ( t,x,z ) ( x {r}_{t} + \pi {\sigma }_{t} {\vartheta }_{t} ) \\
& \quad + \frac{1}{2} {f}_{xx} ( t,x,z ) ( \pi {\sigma }_{t} )^{2} + {f}_{xz} ( t,x,z ) \pi {\sigma }_{t} \gamma + \frac{1}{2} {f}_{zz} ( t,x,z ) {\gamma }^{2}, \\
    \mathcal{D}_{2}^{ \pi, \gamma } f ( t,x,z )
& = {f}_{x} ( t,x,z ) \pi {\sigma }_{t} + {f}_{z} ( t,x,z ) \gamma.
\end{align*}
Since the model parameters $( r, \sigma, \vartheta, \zeta )$ are random, we shall reduce the problem represented by \cref{eq: control problem for DP :eq} to solving a lower Isaacs BSPDE.
Referring to \cite[Sections XI.3, XI.4 and Theorem XI.5.1, pp. 377--383]{Fleming-Soner-2006}, we collect the results of problem reduction in the following verification theorem.

\begin{theorem}[verification theorem]\label{thm: verification theorem for relaxed SDGs}
Suppose that there exists a random field pair
\begin{equation*}
( \mathcal{V}, \Phi ) \in {C}_{\mathbb{F}} \big( 0,T; \mathbb{L}^{2} \big( \Omega; {C}^{2,2} \big( \mathbb{R} \times [ 0, + \infty ) \big) \big) \big)
            \times \mathbb{L}^{2}_{\mathbb{F}} \big( 0,T; \mathbb{L}^{2} \big( \Omega; {C}^{2,2} \big( \mathbb{R} \times [ 0, + \infty ) \big) \big) \big)
\end{equation*}
fulfilling the lower Isaacs BSPDE on $[ 0,T ) \times \mathbb{R} \times \mathbb{R}$:
\begin{equation}\label{eq: HJBI-BSPDE for SDGs :eq}
- d \mathcal{V} ( t,x,z )
= \esssup_{ \pi \in \mathbb{R} } \essinf_{ \gamma \in \mathbb{R} }
  \{ \mathcal{D}_{1}^{ \pi, \gamma } \mathcal{V} ( t,x,z ) + \mathcal{D}_{2}^{ \pi, \gamma } \Phi ( t,x,z ) \} dt
- \Phi ( t,x,z ) d {W}_{t}
\end{equation}
with the terminal condition on $\mathbb{R} \times \mathbb{R}$:
\begin{equation}\label{eq: terminal condition for SDGs :eq}
\mathcal{V} ( T,x,z ) = x \zeta + \kappa \Big( x + \frac{\zeta }{\theta } \Big) z + \frac{ {\kappa }^{2} }{ 2 \theta } {z}^{2},
\end{equation}
and the square-integrability condition
\begin{align}\label{eq: integrability condition for SDGs :eq}
& \mathbb{E} \bigg[ \sup_{ s \in [ t,T ] } | \mathcal{V} ( s, {X}^{ t,x, \pi }_{s}, {Z}^{ t,z, \gamma }_{s} ) | \\
\notag
& \quad           + \int_{t}^{T} \big( | \Phi ( s, {X}^{ t,x, \pi }_{s}, {Z}^{ t,z, \gamma }_{s} ) |^{2} 
                                   + | \mathcal{D}_{2}^{ {\pi }_{s}, {\gamma }_{s} } \mathcal{V} ( s, {X}^{ t,x, \pi }_{s}, {Z}^{ t,z, \gamma }_{s} ) |^{2} \big) ds \bigg] < \infty, \\
\notag
& \forall ( \pi, \gamma ) \in \mathbb{L}^{2}_{\mathbb{F}} ( 0,T; \mathbb{L}^{2} ( \Omega ) ) \times \mathbb{L}^{2}_{\mathbb{F}} ( 0,T; \mathbb{L}^{2} ( \Omega ) ), \quad ( t,x,z ) \in [ 0,T ) \times \mathbb{R} \times [ 0, + \infty ). 
\end{align}
If there exists a control pair $( {\pi }^{*}, {\gamma }^{*} ) \in \mathbb{L}^{2}_{\mathbb{F}} ( 0,T; \mathbb{L}^{2} ( \Omega ) ) \times \mathbb{L}^{2}_{\mathbb{F}} ( 0,T; \mathbb{L}^{2} ( \Omega ) )$ such that
\begin{equation}\label{eq: saddle point condition :eq}
\left\{ \begin{aligned}
&   \mathcal{D}_{1}^{ {\pi }_{s}, {\gamma }^{*}_{s} } \mathcal{V} ( s, {X}^{ t,x, \pi }_{s}, {Z}^{ t,z, {\gamma }^{*} }_{s} )
  + \mathcal{D}_{2}^{ {\pi }_{s}, {\gamma }^{*}_{s} } \Phi ( s, {X}^{ t,x, \pi }_{s}, {Z}^{ t,z, {\gamma }^{*} }_{s} )
\le \mathbb{H}^{ \mathcal{V}, \Phi } ( s, {X}^{ t,x, \pi }_{s}, {Z}^{ t,z, {\gamma }^{*} }_{s} ), \\
&   \mathcal{D}_{1}^{ {\pi }^{*}_{s}, {\gamma }_{s} } \mathcal{V} ( s, {X}^{ t,x, {\pi }^{*} }_{s}, {Z}^{ t,z, \gamma }_{s} )
  + \mathcal{D}_{2}^{ {\pi }^{*}_{s}, {\gamma }_{s} } \Phi ( s, {X}^{ t,x, {\pi }^{*} }_{s}, {Z}^{ t,z, \gamma }_{s} )
\ge \mathbb{H}^{ \mathcal{V}, \Phi } ( s, {X}^{ t,x, {\pi }^{*} }_{s}, {Z}^{ t,z, \gamma }_{s} ),
\end{aligned} \right.
\end{equation}
for any $( \pi, \gamma ) \in \mathbb{L}^{2}_{\mathbb{F}} ( 0,T; \mathbb{L}^{2} ( \Omega ) ) \times \mathbb{L}^{2}_{\mathbb{F}} ( 0,T; \mathbb{L}^{2} ( \Omega ) )$ and
\begin{equation*}
\mathbb{H}^{ \mathcal{V}, \Phi } ( t,x,z ) := 
\esssup_{ \pi \in \mathbb{R} } \essinf_{ \gamma \in \mathbb{R} } \{ \mathcal{D}_{1}^{ \pi, \gamma } \mathcal{V} ( t,x,z ) + \mathcal{D}_{2}^{ \pi, \gamma } \Phi ( t,x,z ) \},
\end{equation*}
then, ${J}^{ \pi, {\gamma }^{*} } ( t,x,z ) \le \mathcal{V} ( t,x,z ) \le {J}^{ {\pi }^{*}, \gamma } ( t,x,z )$, and hence
\begin{align*}
  \mathcal{V} ( t,x,z ) = {J}^{ {\pi }^{*}, {\gamma }^{*} } ( t,x,z )
& = \esssup_{ \pi \in \mathbb{L}^{2}_{\mathbb{F}} ( 0,T; \mathbb{L}^{2} ( \Omega ) ) } \essinf_{ \gamma \in \mathbb{L}^{2}_{\mathbb{F}} ( 0,T; \mathbb{L}^{2} ( \Omega ) ) } {J}^{ \pi, \gamma } ( t,x,z ) \\
& = \essinf_{ \gamma \in \mathbb{L}^{2}_{\mathbb{F}} ( 0,T; \mathbb{L}^{2} ( \Omega ) ) } \esssup_{ \pi \in \mathbb{L}^{2}_{\mathbb{F}} ( 0,T; \mathbb{L}^{2} ( \Omega ) ) } {J}^{ \pi, \gamma } ( t,x,z ).
\end{align*}
\end{theorem}

\begin{proof}
See \Cref{pf-thm: verification theorem for relaxed SDGs}.
\end{proof}

\begin{remark}\label{rem: upper Isaacs BSPDE}
Below we suppress arguments $( t,x,z )$ for notational simplicity.
For the Ansatz $\mathcal{V}_{xx} = 0$, $\mathcal{V}_{xz} \ne 0$ and $\mathcal{V}_{zz} > 0$, 
e.g., for the portfolio problem with the variational representation of MV preference ${U}_{\theta }$, one obtains
\begin{align*}
    \mathbb{H}^{ \mathcal{V}, \Phi }
& = \mathcal{V}_{x} x {r}_{t}
  + \esssup_{ \pi \in \mathbb{R} } \bigg\{ ( \mathcal{V}_{x} {\vartheta }_{t} + {\Phi }_{x} ) \pi {\sigma }_{t} 
                                         + \essinf_{ \gamma \in \mathbb{R} } \Big\{ ( \mathcal{V}_{xz} \pi {\sigma }_{t} + {\Phi }_{z} ) \gamma + \frac{1}{2} \mathcal{V}_{zz} {\gamma }^{2} \Big\} \bigg\} \\
& = \mathcal{V}_{x} x {r}_{t}
  + \esssup_{ \pi \in \mathbb{R} } \bigg\{ ( \mathcal{V}_{x} {\vartheta }_{t} + {\Phi }_{x} ) \pi {\sigma }_{t} - \frac{1}{2} \frac{ ( \mathcal{V}_{xz} \pi {\sigma }_{t} + {\Phi }_{z} )^{2} }{ \mathcal{V}_{zz} } \bigg\} \\
& = \mathcal{V}_{x} x {r}_{t}
  - \frac{1}{ \mathcal{V}_{xz} } ( \mathcal{V}_{x} {\vartheta }_{t} + {\Phi }_{x} ) {\Phi }_{z}
  + \frac{ \mathcal{V}_{zz} }{ 2 | \mathcal{V}_{xz} |^{2} } ( \mathcal{V}_{x} {\vartheta }_{t} + {\Phi }_{x} )^{2}.
\end{align*}
Let $\hat{\pi }$ and $\hat{\gamma }$ be the maximizer and the minimizer, respectively, which fulfill the following first-order derivative conditions
\begin{equation}\label{eq: optimality condition of HJBI-BSPDE for SDGs :eq}
\vec{0} = \begin{pmatrix} \mathcal{V}_{xx} & \mathcal{V}_{xz} \\ \mathcal{V}_{xz} & \mathcal{V}_{zz} \end{pmatrix}
          \binom{ \hat{\pi } {\sigma }_{t} }{ \hat{\gamma } }
        + \binom{ \mathcal{V}_{x} {\vartheta }_{t} + {\Phi }_{x} }{ {\Phi }_{z} }.
\end{equation}
Moreover, $\mathcal{D}_{1}^{ \pi, \hat{\gamma } } \mathcal{V} + \mathcal{D}_{2}^{ \pi, \hat{\gamma } } \Phi \le \mathbb{H}^{ \mathcal{V}, \Phi } \le \mathcal{D}_{1}^{ \hat{\pi }, \gamma } \mathcal{V} + \mathcal{D}_{2}^{ \hat{\pi }, \gamma } \Phi $
follows from
\begin{equation*}
\left\{ \begin{aligned}
  \mathcal{D}_{1}^{ \pi, \hat{\gamma } } \mathcal{V} + \mathcal{D}_{2}^{ \pi, \hat{\gamma } } \Phi - \mathbb{H}^{ \mathcal{V}, \Phi }
& = ( \mathcal{V}_{x} {\vartheta }_{t} + {\Phi }_{x} + \mathcal{V}_{xz} \hat{\gamma } ) ( \pi - \hat{\pi } ) {\sigma }_{t} = 0, \\
  \mathcal{D}_{1}^{ \hat{\pi }, \gamma } \mathcal{V} + \mathcal{D}_{2}^{ \hat{\pi }, \gamma } \Phi - \mathbb{H}^{ \mathcal{V}, \Phi }
& = ( \mathcal{V}_{xz} \hat{\pi } {\sigma }_{t} + {\Phi }_{z} + \mathcal{V}_{zz} \hat{\gamma } ) ( \gamma - \hat{\gamma } ) + \frac{1}{2} \mathcal{V}_{zz} ( \gamma - \hat{\gamma } )^{2} \ge 0.
\end{aligned} \right.
\end{equation*}
In other words, $( \hat{\pi }, \hat{\gamma } )$ is a saddle point of $\mathcal{D}_{1}^{ \pi, \gamma } \mathcal{V} + \mathcal{D}_{2}^{ \pi, \gamma } \Phi $, implying
\begin{equation*}
  \essinf_{ \gamma \in \mathbb{R} } \esssup_{ \pi \in \mathbb{R} } \{ \mathcal{D}_{1}^{ \pi, \gamma } \mathcal{V} ( t,x,z ) + \mathcal{D}_{2}^{ \pi, \gamma } \Phi ( t,x,z ) \} 
= \mathbb{H}^{ \mathcal{V}, \Phi } ( t,x,z ).
\end{equation*}
Hence, \cref{eq: HJBI-BSPDE for SDGs :eq} is equivalent to the so-called upper Isaacs BSPDE, and then it can be named as a ``stochastic Hamilton-Jacobi-Bellman-Isaacs (HJBI) equation''.
\end{remark}

Given the lower Isaacs BSPDE \cref{eq: HJBI-BSPDE for SDGs :eq},
we now investigate the feedback control scheme of the saddle point $( {\pi }^{*}, {\gamma }^{*} )$ for \cref{eq: relaxed SDGs :eq}, i.e.
${\pi }^{*}_{s} = \hat{\pi } ( s, {X}^{ t,x, {\pi }^{*} }_{s}, {Z}^{ t,z, {\gamma }^{*} }_{s} )$ and ${\gamma }^{*}_{s} = \hat{\gamma } ( s, {X}^{ t,x, {\pi }^{*} }_{s}, {Z}^{ t,z, {\gamma }^{*} }_{s} )$
with the random fields $\hat{\pi }, \hat{\gamma }: [ 0,T ] \times \Omega \times \mathbb{R} \times \mathbb{R} \to \mathbb{R}$ mentioned in the above \Cref{rem: upper Isaacs BSPDE}.
Before proceeding to find the explicit expression of $( \hat{\pi }, \hat{\gamma } )$, 
we are supposed to introduce the following $( \mathbb{F}, \mathbb{P} )$-martingale representations due to the randomness of $( \zeta, r, \vartheta )$:
\begin{align}\label{eq: martingale representation of zeta :eq}
\mathbb{E}_{t} [ \zeta ] & = \mathbb{E} [ \zeta ] + \int_{0}^{t} {\eta }_{s} d {W}_{s}, \\
\label{eq: martingale representation of D :eq}
\mathbb{E}_{t} [ {D}_{T} ] & = \mathbb{E} [ {D}_{T} ] + \int_{0}^{t} {\eta }^{(1)}_{s} d {W}_{s}, \\
\label{eq: martingale representation of D2 :eq}
\mathbb{E}_{t} [ | {D}_{T} |^{2} ] & = \mathbb{E} [ | {D}_{T} |^{2} ] + \int_{0}^{t} {\eta }^{(2)}_{s} d {W}_{s}.
\end{align}
In addition, we let 
\begin{equation*}
{\xi }^{(1)}_{t} := - {\vartheta }_{t} - \frac{ {\eta }^{(1)}_{t} }{ \mathbb{E}_{t} [ {D}_{T} ] } \qquad
{\xi }^{(2)}_{t} := 2 {\vartheta }_{t} + \frac{ {\eta }^{(2)}_{t} }{\mathbb{E}_{t} [ | {D}_{T} |^{2} ] },
\end{equation*}
such that applying It\^o's rule yields
\begin{align*}
 & d \frac{ {D}_{t} }{ \mathbb{E}_{t} [ {D}_{T} ] } 
= \bigg( \frac{ | {\eta }^{(1)}_{t} |^{2}  }{ | \mathbb{E}_{t} [ {D}_{T} ] |^{2} } + {\vartheta }_{t} \frac{ {\eta }^{(1)}_{t} }{ \mathbb{E}_{t} [ {D}_{T} ] } - {r}_{t} \bigg) \frac{ {D}_{t} }{ \mathbb{E}_{t} [ {D}_{T} ] } dt 
+ {\xi }^{(1)}_{t} \frac{ {D}_{t} }{ \mathbb{E}_{t} [ {D}_{T} ] } d {W}_{t}, \\
& d \frac{ \mathbb{E}_{t} [ | {D}_{T} |^{2} ] }{ | {D}_{t} |^{2} }
= \bigg( 2 {r}_{t} + 3 | {\vartheta }_{t} |^{2} + 2 {\vartheta }_{t} \frac{ {\eta }^{(2)}_{t} }{ \mathbb{E}_{t} [ | {D}_{T} |^{2} ] } \bigg) \frac{ \mathbb{E}_{t} [ | {D}_{T} |^{2} ] }{ | {D}_{t} |^{2} } dt 
+ {\xi }^{(2)}_{t} \frac{ \mathbb{E}_{t} [ | {D}_{T} |^{2} ] }{ | {D}_{t} |^{2} } d {W}_{t}.
\end{align*}
Notably, if the parameters $( r, \vartheta )$ are deterministic, then $( {\xi }^{(1)}, {\xi }^{(2)} )$ both vanish.
For ease of reference, we summarize the main results in the following theorem.

\begin{theorem}\label{thm: optimal MV solution}
The saddle point $( {\pi }^{*}, {\gamma }^{*} )$ for \cref{eq: relaxed SDGs :eq} can be characterized by the following properties:
\begin{itemize}
\item (feedback control)
      \begin{equation}\label{eq: optimal MV portfolio :eq}
      \left\{ \begin{aligned}
      & \hat{\pi } ( t,x,z ) = \frac{1}{ {\sigma }_{t} } 
                               \bigg( ( \mathbb{E}_{t} [ \zeta ] + \kappa z ) \frac{ \mathbb{E}_{t} [ | {D}_{T} |^{2} ] }{ \theta {D}_{t} \mathbb{E}_{t} [ {D}_{T} ] } ( {\vartheta }_{t} - {\xi }^{(1)}_{t} - {\xi }^{(2)}_{t} ) 
                                    - x {\xi }^{(1)}_{t} \bigg) \\
      & \hat{\gamma } ( t,x,z ) = - \frac{1}{\kappa } \Big( ( \mathbb{E}_{t} [ \zeta ] + \kappa z ) ( {\vartheta }_{t} + {\xi }^{(1)}_{t} ) + {\eta }_{t} \Big);
      \end{aligned} \right.
      \end{equation}
\item (state processes)
      \begin{equation}\label{eq: optimal MV state pair :eq}
      \left\{ \begin{aligned}
      & {X}^{ t,x, {\pi }^{*} }_{s} = x \frac{ {D}_{t} \mathbb{E}_{s} [ {D}_{T} ] }{ {D}_{s} \mathbb{E}_{t} [ {D}_{T} ] }
                                    + ( \mathbb{E}_{t} [ \zeta ] + \kappa z ) 
                                      \frac{ \mathbb{E}_{s} [ {D}_{T} ] \mathbb{E}_{t} [ | {D}_{T} |^{2} ] - \mathbb{E}_{s} [ | {D}_{T} |^{2} ] \mathbb{E}_{t} [ {D}_{T} ] }{ \theta {D}_{s} | \mathbb{E}_{t} [ {D}_{T} ] |^{2} }, \\
      & {Z}^{ t,z, {\gamma }^{*} }_{s} = \frac{1}{\kappa } \mathbb{E}_{s} \bigg[ ( \mathbb{E}_{t} [ \zeta ] + \kappa z ) \frac{ {D}_{T} }{ \mathbb{E}_{t} [ {D}_{T} ] } - \zeta \bigg];
      \end{aligned} \right.
      \end{equation}
\item (value random field)
      \begin{equation}\label{eq: MV value random field :eq}
        \mathcal{V} ( t,x,z ) 
      = x ( \mathbb{E}_{t} [ \zeta ] + \kappa z ) \frac{ {D}_{t} }{ \mathbb{E}_{t} [ {D}_{T} ] }
      + \frac{1}{ 2 \theta } ( \mathbb{E}_{t} [ \zeta ] + \kappa z )^{2} \frac{ \mathbb{E}_{t} [ | {D}_{T} |^{2} ] }{ | \mathbb{E}_{t} [ {D}_{T} ] |^{2} }
      - \frac{1}{ 2 \theta } \mathbb{E}_{t} [ | \zeta |^{2} ].
      \end{equation}
\end{itemize}
Furthermore, if \Cref{ass: SMMV-MV comparison} holds, then $( {\pi }^{*}, {\gamma }^{*} )$ characterized by \cref{eq: optimal MV portfolio :eq}--\cref{eq: MV value random field :eq} is also a saddle point for 
\cref{eq: control problem primal :eq}, \cref{eq: saddle point definition :eq} and \cref{eq: control problem for DP :eq} with $z \ge \frac{1}{\kappa } \mathbb{E}_{t} [ \frac{ {D}_{T} }{ \mathbb{E} [ {D}_{T} ] } - \zeta ]$.
\end{theorem}

\begin{proof}
See \Cref{pf-thm: optimal MV solution}.
\end{proof}

Corresponding to $( {\pi }^{*}, {\gamma }^{*} )$ characterized by \Cref{thm: optimal MV solution},
\begin{equation*}
{X}^{ 0, {x}_{0}, {\pi }^{*} }_{t} = {x}_{0} \frac{ \mathbb{E}_{t} [ {D}_{T} ] }{ {D}_{t} \mathbb{E} [ {D}_{T} ] }
                                   + \frac{ \mathbb{E}_{t} [ {D}_{T} ] \mathbb{E} [ | {D}_{T} |^{2} ] - \mathbb{E}_{t} [ | {D}_{T} |^{2} ] \mathbb{E} [ {D}_{T} ] }{ \theta {D}_{t} | \mathbb{E} [ {D}_{T} ] |^{2} }, \quad
{Z}^{ 0,1, {\gamma }^{*} }_{t} = \frac{1}{\kappa } \mathbb{E}_{t} \bigg[ \frac{ {D}_{T} }{ \mathbb{E} [ {D}_{T} ] } - \zeta \bigg],
\end{equation*}
implying that
\begin{equation*}
{X}^{ 0, {x}_{0}, {\pi }^{*} }_{T} = \frac{ {x}_{0} }{ \mathbb{E} [ {D}_{T} ] }
                                   + \frac{ \Var [ {D}_{T} ] }{ \theta | \mathbb{E} [ {D}_{T} ] |^{2} }
                                   + \frac{1}{\theta } \bigg( 1 - \frac{ {D}_{T} }{ \mathbb{E} [ {D}_{T} ] } \bigg), \quad
{Z}^{ 0,1, {\gamma }^{*} }_{T} = \frac{1}{\kappa } \bigg( \frac{ {D}_{T} }{ \mathbb{E} [ {D}_{T} ] } - \zeta \bigg).
\end{equation*}
In comparison with \Cref{lem: MV optimal portfolio in monotonicity domain}, here we provide a much more clear expression for ${X}^{*}_{T} = {X}^{ 0, {x}_{0}, {\pi }^{*} }_{T}$.

\subsubsection{Duality characterization}

To tackle the issues stemming from the absence of saddle point or well-posedness for \cref{eq: control problem primal :eq}, 
we purpose to investigate a sequence of auxiliary problems (indexed by $( \rho, c ) \in \mathbb{R}_{+} \times \mathbb{L}^{2} ( \Omega )$) formulated by
\begin{equation}\label{eq: approximate control problem :eq}
\sup_{ \pi \in \mathbb{L}^{2}_{\mathbb{F}} ( 0,T; \mathbb{L}^{2} ( \Omega ) ) } \inf_{ \gamma \in {\Gamma }^{0,1} } \Big\{ {J}^{ \pi, \gamma } ( 0, {x}_{0}, 1 ) - \frac{\rho }{2} \mathbb{E} [ ( {X}^{ 0, {x}_{0}, \pi }_{T} - c )^{2} ] \Big\}.
\end{equation}
Obviously,
\begin{align*}
& \sup_{ \pi \in \mathbb{L}^{2}_{\mathbb{F}} ( 0,T; \mathbb{L}^{2} ( \Omega ) ) } \inf_{ \gamma \in {\Gamma }^{0,1} } {J}^{ \pi, \gamma } ( 0, {x}_{0}, 1 ) \\
& \ge \sup_{ ( \rho, c ) \in \mathbb{R}_{+} \times \mathbb{R} } \sup_{ \pi \in \mathbb{L}^{2}_{\mathbb{F}} ( 0,T; \mathbb{L}^{2} ( \Omega ) ) } \inf_{ \gamma \in {\Gamma }^{0,1} } 
      \Big\{ {J}^{ \pi, \gamma } ( 0, {x}_{0}, 1 ) - \frac{\rho }{2} \mathbb{E} [ ( {X}^{ 0, {x}_{0}, \pi }_{T} - c )^{2} ] \Big\} \\
& \ge \sup_{ \pi \in \mathbb{L}^{2}_{\mathbb{F}} ( 0,T; \mathbb{L}^{2} ( \Omega ) ) } \sup_{ \rho \in \mathbb{R}_{+} } \inf_{ \gamma \in {\Gamma }^{0,1} } 
      \Big\{ {J}^{ \pi, \gamma } ( 0, {x}_{0}, 1 ) - \frac{\rho }{2} \inf_{ c \in \mathbb{R} } \mathbb{E} [ ( {X}^{ 0, {x}_{0}, \pi }_{T} - c )^{2} ] \Big\} \\
&   = \sup_{ \pi \in \mathbb{L}^{2}_{\mathbb{F}} ( 0,T; \mathbb{L}^{2} ( \Omega ) ) } \sup_{ \rho \in \mathbb{R}_{+} } \inf_{ \gamma \in {\Gamma }^{0,1} } 
      \Big\{ {J}^{ \pi, \gamma } ( 0, {x}_{0}, 1 ) - \frac{\rho }{2} \Var [ {X}^{ 0, {x}_{0}, \pi }_{T} ] \Big\} \\
&   = \sup_{ \pi \in \mathbb{L}^{2}_{\mathbb{F}} ( 0,T; \mathbb{L}^{2} ( \Omega ) ) } \inf_{ \gamma \in {\Gamma }^{0,1} } {J}^{ \pi, \gamma } ( 0, {x}_{0}, 1 ),
\end{align*}
where the third line formulates a portfolio problem with a mixed SMMV-MV objective functional.
The presence of the arbitrarily chosen $c \in \mathbb{L}^{2} ( \Omega )$ helps our asymptotic analysis.
In particular, in the case where \Cref{ass: SMMV-MV comparison} is violated, we can show the ill-posedness of \cref{eq: control problem primal :eq} by
\begin{align*}
& \sup_{ \pi \in \mathbb{L}^{2}_{\mathbb{F}} ( 0,T; \mathbb{L}^{2} ( \Omega ) ) } \inf_{ \gamma \in {\Gamma }^{0,1} } {J}^{ \pi, \gamma } ( 0, {x}_{0}, 1 ) \\
& \ge \lim_{ \rho \downarrow 0 } \sup_{ \pi \in \mathbb{L}^{2}_{\mathbb{F}} ( 0,T; \mathbb{L}^{2} ( \Omega ) ) } \inf_{ \gamma \in {\Gamma }^{0,1} } 
      \Big\{ {J}^{ \pi, \gamma } ( 0, {x}_{0}, 1 ) - \frac{\rho }{2} \mathbb{E} [ ( {X}^{ 0, {x}_{0}, \pi }_{T} - c )^{2} ] \Big\}.
\end{align*}

For notational simplicity, let us introduce the auxiliary function (with $\zeta $ treated as a constant):
\begin{equation*}
U ( x,z ) := x \zeta + \kappa \Big( x + \frac{\zeta }{\theta } \Big) z + \frac{ {\kappa }^{2} }{ 2 \theta } {z}^{2} - \frac{\rho }{2} ( x - c )^{2},
\end{equation*}
so that ${J}^{ \pi, \gamma } ( 0, {x}_{0}, 1 ) - \frac{\rho }{2} \mathbb{E} [ ( {X}^{ 0, {x}_{0}, \pi }_{T} - c )^{2} ] = \mathbb{E} [ U ( {X}^{ 0, {x}_{0}, \pi }_{T}, {Z}^{ 0,1, \gamma }_{T} ) ]$.
Obviously, $z = \frac{\theta }{\kappa } ( \frac{h}{\kappa } - x - \frac{\zeta }{\theta } )_{+}$ is the unique minimizer for minimizing $U ( x,z ) - z h$ on $z \in [ 0, + \infty )$, and hence
\begin{equation*} 
  \min_{ z \in [ 0, + \infty ) } \{ U ( x,z ) - z h \}
= x \zeta - \frac{\theta }{2} \Big| \Big( \frac{h}{\kappa } - x - \frac{\zeta }{\theta } \Big)_{+} \Big|^{2} - \frac{\rho }{2} ( x - c )^{2}.
\end{equation*} 
Consequently, the pair $( \mathcal{X} ( y,h ), \mathcal{Z} ( y,h ) )$ given by the following continuous functions (where the dependence of $( \omega, \zeta, \kappa, \theta, \rho, c )$ for $\mathcal{X}$ is suppressed for notational simplicity):
\begin{align*}
& \mathcal{X} ( y,h ) := \bigg( \frac{1}{ \rho + \theta } \Big( \frac{\theta }{\kappa } h - y + \rho c \Big) \bigg) \vee \bigg( \frac{1}{\rho } ( \zeta + \rho c - y ) \bigg) 
                       = \frac{1}{\rho } \bigg( \frac{ \rho \theta ( \frac{h}{\kappa } - c - \frac{y}{\theta } ) }{ \rho + \theta } \vee ( \zeta - y ) \bigg) + c, \\
& \mathcal{Z} ( y,h ) := \frac{1}{\kappa } \bigg( \frac{\theta }{ \rho + \theta } \Big( \frac{\rho }{\kappa } h + y - \rho c \Big) - \zeta \bigg)_{+}
                       = \frac{1}{\kappa } \bigg( \frac{ \rho \theta ( \frac{h}{\kappa } - c - \frac{y}{\theta } ) }{ \rho + \theta } - ( \zeta - y ) \bigg)_{+},
\end{align*} 
is the unique saddle point for 
\begin{equation*}
\tilde{U} ( y,h ) := \max_{ x \in \mathbb{R} } \min_{ z \in [ 0, + \infty ) } \{ U ( x,z ) - ( x y + z h ) \} = \min_{ z \in [ 0, + \infty ) } \max_{ x \in \mathbb{R} } \{ U ( x,z ) - ( x y + z h ) \},
\end{equation*}
since it uniquely fulfills the first-order derivative conditions:
\begin{equation}\label{eq: FOC in MCD :eq}
\mathcal{Z} = \frac{\theta }{\kappa } \Big( \frac{h}{\kappa } - \mathcal{X} - \frac{\zeta }{\theta } \Big)_{+}, \quad
0 = \zeta + \kappa \mathcal{Z} - \rho ( \mathcal{X} - c ) - y.
\end{equation}
Therefore, for any $( x,z ) \in \mathbb{R} \times [ 0, + \infty )$ and $( y,h ) \in \mathbb{R}^{2}$, we have
\begin{equation}\label{eq: saddle point in duality :eq}
    U \big( x, \mathcal{Z} ( y,h ) \big) - \big( x y + h \mathcal{Z} ( y,h ) \big) 
\le \tilde{U} ( y,h )
\le U \big( \mathcal{X} ( y,h ), z \big) - \big( y \mathcal{X} ( y,h ) + z h \big).
\end{equation}
Then, the following theorem consolidates the main results of the duality analysis.

\begin{theorem}\label{thm: duality characterization}
On the one hand, there exists a unique pair $( \bar{y}, \bar{h} ) \in \mathbb{R}^{2}$ such that
\begin{equation}\label{eq: y-h system :eq}
\mathbb{E} [ {D}_{T} \mathcal{X} ( \bar{y} {D}_{T}, \bar{h} ) ] = {x}_{0}, \quad 
\mathbb{E} [ \mathcal{Z} ( \bar{y} {D}_{T}, \bar{h} ) ] = 1, 
\end{equation}
and then there exists a unique pair $( \bar{\pi }, \bar{\gamma } ) \in \mathbb{L}^{2}_{\mathbb{F}} ( 0,T; \mathbb{L}^{2} ( \Omega ) ) \times {\Gamma }^{0,1}$ such that
\begin{equation}\label{eq: minimax equality conditions :eq}
{X}^{ 0, {x}_{0}, \bar{\pi } }_{T} = \mathcal{X} ( \bar{y} {D}_{T}, \bar{h} ), \quad
{Z}^{ 0,1, \bar{\gamma } }_{T} = \mathcal{Z} ( \bar{y} {D}_{T}, \bar{h} ). 
\end{equation}
On the other hand, for any $( \pi, \gamma ) \in \mathbb{L}^{2}_{\mathbb{F}} ( 0,T; \mathbb{L}^{2} ( \Omega ) ) \times {\Gamma }^{0,1}$ and $( y,h ) \in \mathbb{R}^{2}$,
\begin{equation}\label{eq: minimax inequalities :eq}
\left\{ \begin{aligned}
&   \mathbb{E} \big[ U \big( {X}^{ 0, {x}_{0}, \pi }_{T}, \mathcal{Z} ( y {D}_{T}, h ) \big) \big]
\le \mathbb{E} [ \tilde{U} ( y {D}_{T}, h ) ] + y {x}_{0} + h \mathbb{E} [ \mathcal{Z} ( y {D}_{T}, h ) ], \\
&   \mathbb{E} \big[ U \big( \mathcal{X} ( y {D}_{T}, h ), {Z}^{ 0,1, \gamma }_{T} \big) \big]
\ge \mathbb{E} [ \tilde{U} ( y {D}_{T}, h ) ] + y \mathbb{E} [ {D}_{T} \mathcal{X} ( y {D}_{T}, h ) ] + h.
\end{aligned} \right.
\end{equation}
Therefore, $( \bar{\pi }, \bar{\gamma } )$ is the saddle point for \cref{eq: approximate control problem :eq}, and $( \bar{y}, \bar{h} )$ is the saddle point for
\begin{equation}\label{eq: minimax equality in duality :eq}
  \sup_{ h \in \mathbb{R} } \inf_{ y \in \mathbb{R} } \{ \mathbb{E} [ \tilde{U} ( y {D}_{T}, h ) ] + y {x}_{0} + h \}
= \inf_{ y \in \mathbb{R} } \sup_{ h \in \mathbb{R} } \{ \mathbb{E} [ \tilde{U} ( y {D}_{T}, h ) ] + y {x}_{0} + h \}.
\end{equation}
\end{theorem}

\begin{proof}
See \Cref{pf-thm: duality characterization}.
\end{proof}

The above duality analysis is also applicable to $\rho = 0$.
By \cref{eq: FOC in MCD :eq}, one obtains
\begin{equation*}
\mathcal{X} ( y,h ) = \frac{h}{\kappa } - \frac{y}{\theta }, \quad
\mathcal{Z} ( y,h ) = \frac{ \zeta - y }{\kappa }, 
\end{equation*}
with the domain $\{ ( \omega, y, h ): y \ge \zeta ( \omega ) \}$.
Under the Ansatz $\zeta \le \bar{y} {D}_{T}$, $\mathbb{P}$-a.s., \cref{eq: y-h system :eq} turns into
\begin{equation*}
\left\{ \begin{aligned}
& \mathbb{E} \Big[ {D}_{T} \Big( \frac{ \bar{h} }{\kappa } - \frac{ \bar{y} }{\theta } {D}_{T} \Big) \Big] = {x}_{0}, \\
& \mathbb{E} [ \bar{y} {D}_{T} - \zeta ] = \kappa, \\
\end{aligned} \right. \quad i.e. \quad
\left\{ \begin{aligned}
& \bar{y} = \frac{1}{ \mathbb{E} [ {D}_{T} ] }, \\
& \bar{h} = \kappa \bigg( \frac{ {x}_{0} }{ \mathbb{E} [ {D}_{T} ] } + \frac{ \Var [ {D}_{T} ] }{ \theta | \mathbb{E} [ {D}_{T} ] |^{2} } + \frac{1}{\theta } \bigg),
\end{aligned} \right.
\end{equation*}
and hence \Cref{ass: SMMV-MV comparison} is presupposed.
Consequently, \cref{eq: minimax equality conditions :eq} reproduces the result in the previous subsection, as the following:
\begin{equation*}
\left\{ \begin{aligned}
& {X}^{ 0, {x}_{0}, \bar{\pi } }_{T} = \frac{ \bar{h} }{\kappa } - \frac{ \bar{y} }{\theta } {D}_{T}
                                     = \frac{ {x}_{0} }{ \mathbb{E} [ {D}_{T} ] } + \frac{ \Var [ {D}_{T} ] }{ \theta | \mathbb{E} [ {D}_{T} ] |^{2} } + \frac{1}{\theta } \bigg( 1 - \frac{ {D}_{T} }{ \mathbb{E} [ {D}_{T} ] } \bigg), \\
& {Z}^{ 0,1, \bar{\gamma } }_{T} = \frac{1}{\kappa } ( \bar{y} {D}_{T} - \zeta ) = \frac{1}{\kappa } \bigg( \frac{ {D}_{T} }{ \mathbb{E} [ {D}_{T} ] } - \zeta \bigg).
\end{aligned} \right.
\end{equation*}
Notably, if we merely let $\rho $ tend to zero, 
$c \le \frac{ \bar{h} }{\kappa } - \frac{ \bar{y} }{\theta } {D}_{T}$, $\mathbb{P}$-a.e. on $\{ \omega: \zeta ( \omega ) = \frac{ {D}_{T} ( \omega ) }{ \mathbb{E} [ {D}_{T} ] } \}$, is presupposed to arrive at the above result for $\rho = 0$.
This condition can be realized, due to the arbitrariness of $c \in \mathbb{L}^{2} ( \Omega )$.

As a result, we can show by contradiction that there does not exist such a limiting pair $( \bar{y}, \bar{h} )$ satisfying \cref{eq: y-h system :eq} with $\rho \downarrow 0$ and some sufficiently small $c$, when \Cref{ass: SMMV-MV comparison} is violated.
Indeed, by the dominated convergence theorem, \cref{eq: y-h system :eq} produces $\bar{y} \mathbb{E} [ {D}_{T} ] \le \mathbb{E} [ ( \bar{y} {D}_{T} ) \vee \zeta ] = 1$, 
and hence for $c \le \frac{ \bar{h} }{\kappa } - \frac{ \bar{y} }{\theta } {D}_{T}$, $\mathbb{P}$-a.s., one obtains
\begin{equation*}
    {x}_{0} - \mathbb{E} [ {D}_{T} c ]
\ge \liminf_{ \rho \downarrow 0 } \mathbb{E} \bigg[ {D}_{T} \bigg( \Big( \frac{ \bar{h} }{\kappa } - c - \frac{ \bar{y} }{\theta } {D}_{T} \Big) \vee \frac{ \zeta - \bar{y} {D}_{T} }{\rho } \bigg) \bigg]
\ge \liminf_{ \rho \downarrow 0 } \frac{1}{\rho } \mathbb{E} \bigg[ {D}_{T} \bigg( \zeta - \frac{ {D}_{T} }{ \mathbb{E} [ {D}_{T} ] } \bigg)_{+} \bigg].
\end{equation*}
The last limit inferior is finite, only if $( \zeta - \frac{ {D}_{T} }{ \mathbb{E} [ {D}_{T} ] } )_{+} = 0$, $\mathbb{P}$-a.s., namely, \Cref{ass: SMMV-MV comparison} holds.
Hence, \Cref{ass: SMMV-MV comparison} is necessary for the existence of the saddle points $( \bar{\pi }, \bar{\gamma } )$ for \cref{eq: control problem primal :eq}.

\subsubsection{Multi-stage-optimization characterization}

In the previous subsection, the existence of such a pair $( \bar{y}, \bar{h} )$ indicates the existence of the saddle point for \cref{eq: approximate control problem :eq}.
So in this subsection, we purpose to investigate the following two-stage optimization problem:
\begin{equation}\label{eq: interchanged approximate control problem :eq}
\text{minimizing} \quad \sup_{ \pi \in \mathbb{L}^{2}_{\mathbb{F}} ( 0,T; \mathbb{L}^{2} ( \Omega ) ) } \bigg\{ {J}^{ \pi, \gamma } ( 0, {x}_{0}, 1 ) - \frac{\rho }{2} \mathbb{E} [ ( {X}^{ 0, {x}_{0}, \pi }_{T} - c )^{2} ] \bigg\} \quad
\text{s.t.} \quad \gamma \in {\Gamma }^{0,1}.
\end{equation} 
Before proceeding, we introduce the following auxiliary functionals for the sake of brevity:
\begin{align*}
{F}_{ \rho, j } ( \gamma ) := \frac{j}{ 2 \mathbb{E} [ | {D}_{T} |^{2} ] } \mathbb{E} [ {D}_{T} ( \kappa {Z}^{ 0,1, \gamma }_{T} + \zeta ) ] - \frac{\rho }{ \mathbb{E} [ | {D}_{T} |^{2} ] } ( {x}_{0} - \mathbb{E} [ {D}_{T} c ] ),
\quad \rho \ge 0, ~ j = 1,2.
\end{align*}
By the martingale convex duality method as used in the previous subsection, we have
\begin{align*}
& \sup_{ \pi \in \mathbb{L}^{2}_{\mathbb{F}} ( 0,T; \mathbb{L}^{2} ( \Omega ) ) } \bigg\{ {J}^{ \pi, \gamma } ( 0, {x}_{0}, 1 ) - \frac{\rho }{2} \mathbb{E} [ ( {X}^{ 0, {x}_{0}, \pi }_{T} - c )^{2} ] \bigg\} \\
& = \sup_{ \pi \in \mathbb{L}^{2}_{\mathbb{F}} ( 0,T; \mathbb{L}^{2} ( \Omega ) ) }
    \mathbb{E} \Big[ {X}^{ 0, {x}_{0}, \pi }_{T} ( \kappa {Z}^{ 0,1, \gamma }_{T} + \zeta ) - \frac{\rho }{2} ( {X}^{ 0, {x}_{0}, \pi }_{T} - c )^{2} \Big] 
  + \frac{1}{ 2 \theta } \mathbb{E} [ ( \kappa {Z}^{ 0,1, \gamma }_{T} + \zeta )^{2} - {\zeta }^{2} ] \\
& = \inf_{ y \in \mathbb{R} } \Big\{ \mathbb{E} \Big[ c ( \kappa {Z}^{ 0,1, \gamma }_{T} + \zeta - y {D}_{T} ) + \frac{1}{ 2 \rho } ( \kappa {Z}^{ 0,1, \gamma }_{T} + \zeta - y {D}_{T} )^{2} \Big] + y {x}_{0} \Big\} \\
& \quad + \frac{1}{ 2 \theta } \mathbb{E} [ ( \kappa {Z}^{ 0,1, \gamma }_{T} + \zeta )^{2} - {\zeta }^{2} ] \\
& = \frac{1}{\rho } \bigg( \frac{ \theta + \rho }{ 2 \theta } \mathbb{E} [ ( \kappa {Z}^{ 0,1, \gamma }_{T} + \zeta )^{2} ] 
                         + \rho \mathbb{E} [ c ( \kappa {Z}^{ 0,1, \gamma }_{T} + \zeta ) ]
                         - {F}_{ \rho, 1 } ( \gamma ) \mathbb{E} [ {D}_{T} ( \kappa {Z}^{ 0,1, \gamma }_{T} + \zeta ) ] \bigg) \\
& \quad - \frac{ \rho ( {x}_{0} - \mathbb{E} [ {D}_{T} c ] )^{2} }{ 2 \mathbb{E} [ | {D}_{T} |^{2} ] } - \frac{1}{ 2 \theta } \mathbb{E} [ {\zeta }^{2} ],
\end{align*}
where the infimum and the supremum are realized by
\begin{equation*}
{y}^{\dag } = {F}_{ \rho, 2 } ( \gamma ), \quad 
{X}^{ 0, {x}_{0}, {\pi }^{\dag } }_{T} = \frac{1}{\rho } ( \kappa {Z}^{ 0,1, \gamma }_{T} + \zeta - {y}^{\dag } {D}_{T} ) + c, \quad \text{respectively}.
\end{equation*}

For any $\rho > 0$, \cref{eq: interchanged approximate control problem :eq} is reduced to the minimization problem formulated by 
\begin{equation}\label{eq: minimization problem in two-stage :eq}
\inf_{ \gamma \in {\Gamma }^{0,1} } \bigg\{ \frac{ \theta + \rho }{ 2 \theta } \mathbb{E} [ ( \kappa {Z}^{ 0,1, \gamma }_{T} + \zeta )^{2} ] 
                                          + \rho \mathbb{E} [ c ( \kappa {Z}^{ 0,1, \gamma }_{T} + \zeta ) ] 
                                          - {F}_{ \rho, 1 } ( \gamma ) \mathbb{E} [ {D}_{T} ( \kappa {Z}^{ 0,1, \gamma }_{T} + \zeta ) ] \bigg\}.
\end{equation} 
Let ${\gamma }^{**}$ denote the solution for this minimization problem. 
Moreover, to apply the embedding method pioneered by \cite{Li-Ng-2000} (see also \cite{Zhou-Li-2000} and \cite[Theorem 6.8.2, p. 338]{Yong-Zhou-1999} for the continuous-time framework),
we let $\bar{\Gamma }^{0,1} (w)$ be the set of all minimizers for
\begin{equation}\label{eq: embedding problem :eq}
\inf_{ \gamma \in {\Gamma }^{0,1} } \bigg\{ \frac{ \theta + \rho }{ 2 \theta } \mathbb{E} [ ( \kappa {Z}^{ 0,1, \gamma }_{T} + \zeta )^{2} ] 
                                          + \rho \mathbb{E} [ c ( \kappa {Z}^{ 0,1, \gamma }_{T} + \zeta ) ]
                                          - w \mathbb{E} [ {D}_{T} ( \kappa {Z}^{ 0,1, \gamma }_{T} + \zeta ) ] \bigg\}.
\end{equation}
Then, we isolate the key result of applying the embedding method in the following \Cref{thm: embedding method}. 

\begin{theorem}\label{thm: embedding method}
${\gamma }^{**} \in \bar{\Gamma }^{0,1} ( {w}^{**} )$ with ${w}^{**} = {F}_{ \rho, 2 } ( {\gamma }^{**} )$.
\end{theorem}

\begin{proof}
See \Cref{pf-thm: embedding method}.
\end{proof}

By virtue of \Cref{thm: embedding method}, we are supposed to solve the problem \cref{eq: embedding problem :eq} to find the solution for \cref{eq: interchanged approximate control problem :eq}.
Applying the martingale convex duality method again, we have
\begin{align*}
& \inf_{ \gamma \in {\Gamma }^{0,1} } \mathbb{E} \bigg[ \frac{ \theta + \rho }{ 2 \theta } ( \kappa {Z}^{ 0,1, \gamma }_{T} + \zeta )^{2} - ( w {D}_{T} - \rho c ) ( \kappa {Z}^{ 0,1, \gamma }_{T} + \zeta ) \bigg] \\
& = \sup_{ h \in \mathbb{R} } \bigg\{ \mathbb{E} \bigg[ \frac{ \theta + \rho }{ 2 \theta } | \zeta |^{2} - ( w {D}_{T} - \rho c ) \zeta 
                                                      - \frac{\theta }{ 2 ( \theta + \rho ) } \Big| \Big( w {D}_{T} - \rho c + \frac{h}{\kappa } - \frac{ \theta + \rho }{\theta } \zeta \Big)_{+} \Big|^{2} \bigg] + h \bigg\},
\end{align*}
whereby ${h}^{\dag } (w)$ and ${\gamma }^{\dag } (w)$ given by
\begin{equation*}
\left\{ \begin{aligned}
& 1 = \frac{\theta }{ \kappa ( \theta + \rho ) } \mathbb{E} \Big[ \Big( w {D}_{T} - \rho c + \frac{1}{\kappa } {h}^{\dag } (w) - \frac{ \theta + \rho }{\theta } \zeta \Big)_{+} \Big], \\
& {Z}^{ 0,1, {\gamma }^{\dag } (w) }_{T} = \frac{\theta }{ \kappa ( \theta + \rho ) } \Big( w {D}_{T} - \rho c + \frac{1}{\kappa } {h}^{\dag } (w) - \frac{ \theta + \rho }{\theta } \zeta \Big)_{+}
\end{aligned} \right.
\end{equation*}
realize the supremum and the infimum, respectively.
The existence and uniqueness of ${h}^{\dag } (w)$ is straightforward.
Due to the uniqueness of ${h}^{\dag } (w)$ and the martingale representation theorem, one obtains $\bar{\Gamma }^{0,1} (w) = \{ {\gamma }^{\dag } (w) \}$.
Therefore, ${w}^{**} = {F}_{ \rho, 2 } ( {\gamma }^{**} ) = {F}_{ \rho, 2 } ( {\gamma }^{\dag } ( {w}^{**} ) )$.
We should note that the existence of saddle point $( \bar{\pi }, \bar{\gamma } )$ provided by \Cref{thm: duality characterization} merely implies the existence of ${\gamma }^{**}$ and ${w}^{**}$, even if the saddle point is unique.
We need the following \Cref{lem: uniqueness of w} to guarantee that the solution of $w = {F}_{ \rho, 2 } ( {\gamma }^{\dag } (w) )$ is necessarily the desired ${w}^{**}$. 

\begin{proposition}\label{lem: uniqueness of w}
$w = {F}_{ \rho, 2 } ( {\gamma }^{\dag } (w) )$ has a unique solution ${w}^{\S }$, and ${w}^{\S } \ge \frac{ \mathbb{E} [ {D}_{T} \zeta ] - \rho ( {x}_{0} - c \mathbb{E} [ {D}_{T} ] ) }{ \mathbb{E} [ | {D}_{T} |^{2} ] }$.
\end{proposition}

\begin{proof}
See \Cref{pf-lem: uniqueness of w}.
\end{proof}

So far, we have already derived the analytical solution for \cref{eq: interchanged approximate control problem :eq}.
That is, the minimizer $\bar{\gamma } = {\gamma }^{\dag } ( {w}^{\S } )$ is given by the martingale representation (namely, the It\^o's representation)
\begin{equation*}
\frac{\theta }{ \kappa ( \theta + \rho ) } \Big( {w}^{\S } {D}_{T} - \rho c + \frac{1}{\kappa } {h}^{\dag } ( {w}^{\S } ) - \frac{ \theta + \rho }{\theta } \zeta \Big)_{+} = 1 + \int_{0}^{T} \bar{\gamma }_{t} d {W}_{t}.
\end{equation*}
Moreover, corresponding to the ``optimal'' portfolio strategy $\bar{\pi }$ (for \cref{eq: interchanged approximate control problem :eq} at the minimum point $\bar{\gamma }$), we attain
\begin{equation*}
  {X}^{ 0, {x}_{0}, \bar{\pi } }_{T} 
= \frac{1}{\rho } \bigg( \frac{\theta }{ \theta + \rho } \Big( {w}^{\S } {D}_{T} - \rho c + \frac{1}{\kappa } {h}^{\dag } ( {w}^{\S } ) - \frac{ \theta + \rho }{\theta } \zeta \Big)_{+} - ( {w}^{\S } {D}_{T} - \rho c - \zeta ) \bigg),
\end{equation*}
implying that $\mathbb{E} [ {X}^{ 0, {x}_{0}, \bar{\pi } }_{T} ] = \frac{1}{\rho } ( 1 - {w}^{\S } \mathbb{E} [ {D}_{T} ] ) + \mathbb{E} [c]$.
Therefore, the existence of the saddle point $( \bar{\pi }, \bar{\gamma } )$ at the limiting case where $\rho $ tends to $0$ results in ${w}^{\S } = \frac{1}{ \mathbb{E} [ {D}_{T} ] } + O ( \rho )$.

In particular, under the Ansatz ${w}^{\S } {D}_{T} - \rho c + \frac{1}{\kappa } {h}^{\dag } ( {w}^{\S } ) - \frac{ \theta + \rho }{\theta } \zeta \ge 0$, one obtains
\begin{align*}
{w}^{\S } 
          & = \frac{1}{ \mathbb{E} [ {D}_{T} ] } 
            - \rho \frac{ \frac{ \Var [ {D}_{T} ] }{ \mathbb{E} [ {D}_{T} ] } + \theta ( {x}_{0} - \mathbb{E} [ {D}_{T} ] \mathbb{E} [ c ] ) }
                        { \rho \mathbb{E} [ | {D}_{T} |^{2} ] + \theta | \mathbb{E} [ {D}_{T} ] |^{2} }
            - {\rho }^{2} \frac{ {x}_{0} - \mathbb{E} [ {D}_{T} c ] }{ \rho \mathbb{E} [ | {D}_{T} |^{2} ] + \theta | \mathbb{E} [ {D}_{T} ] |^{2} }, \\
{h}^{\dag } ( {w}^{\S } ) & 
                            = \rho \bigg( \frac{1}{\theta } + \frac{ \Var [ {D}_{T} ] }{ \theta \mathbb{E} [ {D}_{T} ] } + \frac{ {x}_{0} }{ \mathbb{E} [ {D}_{T} ] } \bigg)
                            + {\rho }^{2} \frac{ ( {x}_{0} - \mathbb{E} [ {D}_{T} c ] ) \mathbb{E} [ {D}_{T} ] }{ \rho \mathbb{E} [ | {D}_{T} |^{2} ] + \theta | \mathbb{E} [ {D}_{T} ] |^{2} },
\end{align*} 
as $\rho \downarrow 0$. 
Substituting these limiting results into the aforementioned Ansatz implies that \Cref{ass: SMMV-MV comparison} is presupposed in this limiting case.
Moreover, under the aforementioned Ansatz, we have
\begin{align*}
& {Z}^{ 0,1, \bar{\gamma } }_{T} 
  = \frac{\theta }{ \kappa ( \theta + \rho ) } \Big( {w}^{\S } {D}_{T} - \rho c + \frac{1}{\kappa } {h}^{\dag } ( {w}^{\S } ) - \frac{ \theta + \rho }{\theta } \zeta \Big)
  \to \frac{1}{\kappa } \Big( \frac{ {D}_{T} }{ \mathbb{E} [ {D}_{T} ] } - \zeta \Big), \\
& {X}^{ 0, {x}_{0}, \bar{\pi } }_{T} 
  = \frac{ \theta {h}^{\dag } ( {w}^{\S } ) }{ \kappa \rho ( \theta + \rho ) } - \frac{ {w}^{\S } {D}_{T} }{ \theta + \rho } + \frac{ \rho c }{ \theta + \rho } 
  \to \frac{ {x}_{0} }{ \mathbb{E} [ {D}_{T} ] } + \frac{ \Var [ {D}_{T} ] }{ \theta | \mathbb{E} [ {D}_{T} ] |^{2} } + \frac{1}{\theta } \bigg( 1 - \frac{ {D}_{T} }{ \mathbb{E} [ {D}_{T} ] } \bigg),
\end{align*}
in both the sense of almost sure convergence and mean-square convergence, as $\rho \downarrow 0$.

\section{Concluding remark}
\label{sec: Concluding remark}

To address the limitations of conventional MMV preference, we modify the application of Fenchel conjugate and introduce a class of SMMV preferences. 
The associated portfolio selection problems are studied in both single-period and continuous-time settings.

In the single-period SMMV portfolio problem, we provide the gradient condition that is sufficient and necessary for optimality, followed by the exploration on the existence and uniqueness of its solution.
Moreover, by comparing the solutions of the static MMV and SMMV problems, 
we find that in the case with only one risky asset, the sign of the optimal SMMV portfolio strategy can be determined by the sign of the optimal MV portfolio strategy.
Furthermore, we reduce the SMMV portfolio problem to solving a linear system, which is a significantly simpler algorithm compared to the approach in literature, 
and conduct some numerical experiments to illustrate the single-period static portfolio selection under SMMV preferences.

In the continuous-time SMMV portfolio problem, we identify the equivalent condition for the existence of solution, along with the notable result that the optimal strategies for SMMV and MV preferences coincide. 
To characterize the solution, different approaches are employed, such as the stochastic control techniques (including the dynamic programming principle and the martingale convex duality method) and the penalty function method. 
In particular, we formulate a sequence of auxiliary problems by incorporating some quadratic penalty functions to the SMMV objective functional, and apply the embedding method to solve them.
These discussions provide a potential study framework for future research on SMMV dynamic portfolio problems.
For practical applications, the solutions of these auxiliary problems can be treated as the ``near-optimal'' SMMV strategies, unless the equivalent condition holds.

\appendix
\section{Convex duality analysis to the MV problem}
\label{app: Convex duality analysis to the MV problem}

For any $\mathbb{P}$-complete sub-$\sigma $-field of $\mathcal{G} \subset \mathcal{F}$,
let $\mathbb{L}^{2}_{\mathcal{G}} ( \Omega )$ (resp. $\mathbb{L}^{2}_{\mathcal{G}} ( \Omega; \mathbb{R}_{+} )$) denote the set of all $\mathcal{G}$-measurable random variables $X: \Omega \to \mathbb{R}$ (resp. $X: \Omega \to \mathbb{R}_{+}$) 
such that $\mathbb{E} [ |X|^{2} ] < \infty $.
We isolate the main results for reference in the following \Cref{lem: general mean-variance optimization}, and put the related convex duality analysis in its proof.

\begin{lemma}\label{lem: general mean-variance optimization}
Suppose that $R \in \mathbb{L}^{2} ( \Omega )$, $\mathbb{D} \subset \mathbb{L}^{2} ( \Omega )$ is a non-trivial closed convex cone such that $\mathbb{D}_{1} := \{ Y \in \mathbb{D}: \mathbb{E} [ Y | \mathcal{G} ] = 1 \} \ne \varnothing $, 
and $\mathcal{G} \subset \mathcal{F}$ is a $\mathbb{P}$-complete sub-$\sigma $-field.
Then, 
\begin{align}\label{eq: MV upper bound :eq}
& \esssup_{ X \in \mathbb{X} ( R, \mathbb{D}, \mathcal{G} ) } 
  \bigg\{ \mathbb{E} [ X | \mathcal{G} ] - \frac{\theta }{2} \mathbb{E} \big[ ( X - \mathbb{E} [ X | \mathcal{G} ] )^{2} \big| \mathcal{G} \big] \bigg\} \\
\notag
& = \frac{1}{ 2 \theta } ( 1 - \mathbb{E} [ {R}^{2} | \mathcal{G} ] )
  + \esssup_{ c \in \mathbb{L}^{2}_{\mathcal{G}} ( \Omega ) } \essinf_{ y \in \mathbb{L}^{2}_{\mathcal{G}} ( \Omega; \mathbb{R}_{+} ) } 
    \bigg\{ ( 1 - y ) c - \frac{y}{\theta } + \frac{1}{ 2 \theta } \essinf_{ Y \in \mathbb{D}_{1} } \mathbb{E} [ ( y Y + \theta R )^{2} | \mathcal{G} ] \bigg\}, 
\end{align}
where $\mathbb{X} ( R, \mathbb{D}, \mathcal{G} ) := \{ X \in \mathbb{L}^{2} ( \Omega ): \mathbb{E} [ ( X - R ) Y | \mathcal{G} ] \le 0, \forall Y \in \mathbb{D} \}$.
In particular, the maximizer ${X}^{*}$ for the first line of \cref{eq: MV upper bound :eq} satisfies ${X}^{*} = \mathbb{E} [ {X}^{*} | \mathcal{G} ] + \frac{1}{\theta } ( 1 - {Y}^{*} )$ and $\mathbb{E} [ ( {X}^{*} - R ) {Y}^{*} | \mathcal{G} ] = 0$,
where ${Y}^{*}$ maximizes $\mathbb{E} [ ( Y + \theta R )^{2} | \mathcal{G} ]$ over all $Y \in \mathbb{D}_{1}$.
\end{lemma}

\begin{proof}
Let us proceed with the following minimization problem indexed by $c \in \mathbb{R}$:
\begin{equation}\label{eq: quadratic minimization problem :eq}
minimizing \quad \mathbb{E} \Big[ \frac{\theta }{2} \Big( X - c - \frac{1}{\theta } \Big)^{2} \Big| \mathcal{G} \Big], \quad s.t. \quad X \in \mathbb{X} ( R, \mathbb{D}, \mathcal{G} ).
\end{equation}
Notably, this minimization problem is defined on a (weakly) closed convex subset of reflexive Hilbert space $\mathbb{L}^{2} ( \Omega )$,
and $\mathbb{X} ( R, \mathbb{D}, \mathcal{G} ) \cap \{ X \in \mathbb{L}^{2} ( \Omega ): \mathbb{E} [ {X}^{2} ] \le t \}$ with sufficiently large $t$ weakly compact (according to Kakutani's Theorem). 
We refer to the infinite-dimensional version of the Weierstrass Theorems (e.g., \cite[Theorems 2.3.4 and 2.3.5, p. 56]{Bobylev-Emelyanov-Korovin-1999}), and conclude the existence of the solution.
For brevity, we introduce $u (x) := \frac{\theta }{2} ( x - c - \frac{1}{\theta } )^{2}$ with its Fenchel conjugate
\begin{equation*}
\tilde{u} (y) := \min_{ x \in \mathbb{R} } \{ u(x) + xy \} = \Big( c + \frac{1}{\theta } \Big) y - \frac{1}{ 2 \theta } {y}^{2}, 
\end{equation*}
whereby $x = I(y) := \tilde{u}' (y) = c + \frac{1}{\theta } ( 1 - y )$ realizes the minimum.
Conversely, $u(x) = \max_{ y \in \mathbb{R} } \{ \tilde{u} (y) - xy \}$, whereby $y = - u'(x)$ realizes the maximum.
To show the methodology and potential inspiration for fairly general problems, 
we state our argument with using the twice continuously differentiable convex-concave pair $( u, \tilde{u} )$ of Fenchel duality and $I = \tilde{u}'$ as the inverse function of $- u'$, rather than their specific expressions.
By straightforward calculation, for every $X \in \mathbb{L}^{2} ( \Omega )$ we have
\begin{equation*}
  \mathbb{E} [ u(X) | \mathcal{G} ]
= \esssup_{ Y \in \mathbb{L}^{2} ( \Omega ) } 
  \bigg\{ \essinf_{ Z \in \mathbb{L}^{2} ( \Omega ) } \mathbb{E} [ u(Z) + Y Z | \mathcal{G} ] - \mathbb{E} [ X Y | \mathcal{G} ] \bigg\}
= \esssup_{ Y \in \mathbb{L}^{2} ( \Omega ) } \mathbb{E} [ \tilde{u} (Y) - X Y | \mathcal{G} ],
\end{equation*}
which is also recognized as an application of Fenchel-Moreau theorem (see \cite[Theorem 13.37]{Bauschke-Combettes-2017}). 
Consequently, $\mathbb{E} [ u(X) | \mathcal{G} ] \ge \mathbb{E} [ \tilde{u} (Y) - X Y | \mathcal{G} ] \ge \mathbb{E} [ \tilde{u} (Y) - R Y | \mathcal{G} ]$ 
for any $( X,Y ) \in \mathbb{X} ( R, \mathbb{D}, \mathcal{G} ) \times \mathbb{D}$,
which results in
\begin{equation*}
\essinf_{ X \in \mathbb{X} ( R, \mathbb{D}, \mathcal{G} ) } \mathbb{E} [ u(X) | \mathcal{G} ]
\ge \esssup_{ Y \in \mathbb{D} } \mathbb{E} [ \tilde{u} (Y) - R Y | \mathcal{G} ].
\end{equation*}

Let $\mathbb{D}_{0} := \{ Y \in \mathbb{D}: \mathbb{E} [ ( I(Y) - R ) Y | \mathcal{G} ] = 0 \}$, so that $\mathbb{E} [ \tilde{u} (Y) - R Y | \mathcal{G} ] = \mathbb{E} [ u ( I (Y) ) | \mathcal{G} ]$ for all $Y \in \mathbb{D}_{0}$.
Notably, for any $Y \in \mathbb{D}$, there exists a unique ${y}_{Y} \in \mathbb{L}^{2}_{\mathcal{G}} ( \Omega; \mathbb{R}_{+} )$ such that $\mathbb{E} [ ( I ( {y}_{Y} Y ) - R ) Y | \mathcal{G} ] = 0$ due to the monotonicity of $I$, 
and hence ${y}_{Y} Y \in \mathbb{D}_{0}$.
This implies that $\mathbb{X} ( R, \mathbb{D}, \mathcal{G} ) = \mathbb{X} ( R, \mathbb{D}_{0}, \mathcal{G} )$.
For a given point ${Y}^{\dag } \in \mathbb{D}_{0}$, we have the following three equivalent conditions that are analogous to \cite[Section 9]{Karatzas-Lehoczky-Shreve-Xu-1991}:
\begin{itemize}
\item[(C1)] Feasibility of ${Y}^{\dag }$: $I ( {Y}^{\dag } ) \in \mathbb{X} ( R, \mathbb{D}, \mathcal{G} )$. 
\item[(C2)] Most-favorability of ${Y}^{\dag }$: ${Y}^{\dag }$ maximizes $\mathbb{E} [ u ( I(Y) ) | \mathcal{G} ]$ over all $Y \in \mathbb{D}_{0}$.
\item[(C3)] Dual optimality of ${Y}^{\dag }$: ${Y}^{\dag }$ maximizes $\mathbb{E} [ \tilde{u} (Y) - R Y | \mathcal{G} ]$ over all $Y \in \mathbb{D}$.
\end{itemize}
At first, we show (C1) $\Rightarrow $ (C2). 
According to the definition of Fenchel conjugate, 
\begin{equation*}
    \mathbb{E} [ u(X) | \mathcal{G} ] 
\ge \mathbb{E} [ \tilde{u} (Y) - X Y | \mathcal{G} ] 
\ge \mathbb{E} [ \tilde{u} (Y) - R Y | \mathcal{G} ] 
  = \mathbb{E} \big[ u \big( I (Y) \big) \big| \mathcal{G} \big], 
~ \forall ( X,Y ) \in \mathbb{X} ( R, \mathbb{D}, \mathcal{G} ) \times \mathbb{D}_{0}.
\end{equation*}
If (C1) holds, then (C2) immediately arises from
\begin{equation*}
    \mathbb{E} \big[ u \big( I ( {Y}^{\dag } ) \big) \big| \mathcal{G} \big] 
  = \essinf_{ X \in \mathbb{X} ( R, \mathbb{D}, \mathcal{G} ) } \mathbb{E} [ u(X) | \mathcal{G} ]
\ge \mathbb{E} \big[ u \big( I (Y) \big) \big| \mathcal{G} \big], 
\quad \forall Y \in \mathbb{D}_{0}.
\end{equation*}
Next, we show (C2) $\Rightarrow $ (C3). 
It follows from (C2) that
\begin{equation*}
    \mathbb{E} [ \tilde{u} ( {Y}^{\dag } ) - R {Y}^{\dag } | \mathcal{G} ] 
  = \mathbb{E} \big[ u \big( I ( {Y}^{\dag } ) \big) \big| \mathcal{G} \big] 
\ge \mathbb{E} \big[ u \big( I (Y) \big) \big| \mathcal{G} \big]
  = \mathbb{E} [ \tilde{u} (Y) - R Y | \mathcal{G} ], \quad \forall Y \in \mathbb{D}_{0}.
\end{equation*}
As $\mathbb{E} [ \tilde{u} ( y Y ) - y R Y | \mathcal{G} ]$ is concave in $y$ and $d \mathbb{E} [ \tilde{u} ( y Y ) - y R Y | \mathcal{G} ] / dy = \mathbb{E} [ I ( y Y ) Y - R Y | \mathcal{G} ]$, we have
\begin{equation*}
\mathbb{E} [ \tilde{u} ( {y}_{Y} Y ) - {y}_{Y} R Y | \mathcal{G} ] \ge \mathbb{E} [ \tilde{u} (Y) - R Y | \mathcal{G} ], \quad \forall Y \in \mathbb{D}.
\end{equation*}
As ${y}_{Y} Y \in \mathbb{D}_{0}$, (C3) immediately arises.
Finally, we show (C3) $\Rightarrow $ (C1) by the convex variation method.
In fact, (C3) implies that
\begin{align*}
0 & \ge \lim_{ \varepsilon \downarrow 0 }
        \frac{d}{d \varepsilon } \mathbb{E} \big[ \tilde{u} \big( {Y}^{\dag } + \varepsilon ( Y - {Y}^{\dag } ) \big) 
                                                - R \big( {Y}^{\dag } + \varepsilon ( Y - {Y}^{\dag } ) \big) \big| \mathcal{G} \big] \\
  &   = \mathbb{E} \big[ \big( I ( {Y}^{\dag } ) - R \big) ( Y - {Y}^{\dag } ) \big| \mathcal{G} \big]
      = \mathbb{E} \big[ \big( I ( {Y}^{\dag } ) - R \big) Y \big| \mathcal{G} \big], \quad \forall Y \in \mathbb{D},
\end{align*}
which immediately results in (C1).

If ${X}^{\dag } \in \mathbb{X} ( R, \mathbb{D}, \mathcal{G} )$ is the minimizer for \cref{eq: quadratic minimization problem :eq}, 
then applying the first-order derivative optimality condition with $\{ R + x ( {X}^{\dag } - R ) \}_{ x \ge 0 } \subset \mathbb{X} ( R, \mathbb{D}, \mathcal{G} )$ gives
\begin{equation*}
0 = \frac{d}{dx} \mathbb{E} \big[ u \big( R + x ( {X}^{\dag } - R ) \big) \big| \mathcal{G} \big] \big|_{x=1}
  = \mathbb{E} [ ( {X}^{\dag } - R ) u' ( {X}^{\dag } ) | \mathcal{G} ].
\end{equation*}
Applying the convex variation method yields
\begin{align*}
0 & \le \lim_{ \varepsilon \downarrow 0 } \frac{d}{d \varepsilon } \mathbb{E} \big[ u \big( {X}^{\dag } + \varepsilon ( X - {X}^{\dag } ) \big) \big| \mathcal{G} \big] \\
  &   = \mathbb{E} [ ( X - {X}^{\dag } ) u' ( {X}^{\dag } ) | \mathcal{G} ]
      = \mathbb{E} [ ( X - R ) u' ( {X}^{\dag } ) | \mathcal{G} ], \quad \forall X \in \mathbb{X} ( R, \mathbb{D}, \mathcal{G} ).
\end{align*}
Hence, $- u' ( {X}^{\dag } ) \in \mathbb{D}_{0}$.
Conversely, provided that $- u' ( {X}^{\dag } ) \in \mathbb{D}_{0}$, one obtains
\begin{equation*}
    \mathbb{E} [ u ( {X}^{\dag } ) | \mathcal{G} ]
  = \mathbb{E} \big[ \tilde{u} \big( - u' ( {X}^{\dag } ) \big) + R u' ( {X}^{\dag } ) \big| \mathcal{G} \big]
\le \esssup_{ Y \in \mathbb{D} } \mathbb{E} [ \tilde{u} (Y) - R Y | \mathcal{G} ] 
\le \mathbb{E} [ u(X) | \mathcal{G} ]
\end{equation*}
for any $X \in \mathbb{X} ( R, \mathbb{D}, \mathcal{G} )$,
where the inequalities both arise from
\begin{equation*}
    \mathbb{E} [ \tilde{u} (Y) - R Y | \mathcal{G} ] 
\le \mathbb{E} [ u(X) + (X-R) Y | \mathcal{G} ] 
\le \mathbb{E} [ u(X) | \mathcal{G} ], 
\quad \forall ( X,Y ) \in \mathbb{X} ( R, \mathbb{D}, \mathcal{G} ) \times \mathbb{D}.
\end{equation*}
In conclusion, for ${X}^{\dag } = I ( {Y}^{\dag } ) \in \mathbb{X} ( R, \mathbb{D}, \mathcal{G} )$ and ${Y}^{\dag } = - u' ( {X}^{\dag } ) \in \mathbb{D}_{0}$, we obtain
\begin{equation*}
  \mathbb{E} [ u ( {X}^{\dag } ) | \mathcal{G} ] 
= \essinf_{ X \in \mathbb{X} ( R, \mathbb{D}, \mathcal{G} ) } \mathbb{E} [ u(X) | \mathcal{G} ]
= \esssup_{ Y \in \mathbb{D} } \mathbb{E} [ \tilde{u} (Y) - R Y | \mathcal{G} ]
= \mathbb{E} [ \tilde{u} ( {Y}^{\dag } ) - R {Y}^{\dag } | \mathcal{G} ]
\end{equation*}
and $\mathbb{E} [ ( {X}^{\dag } - R ) {Y}^{\dag } | \mathcal{G} ] = 0$.

To prove \Cref{lem: general mean-variance optimization}, let us proceed with the following equalities:
\begin{align*}
&   \esssup_{ X \in \mathbb{X} ( R, \mathbb{D}, \mathcal{G} ) } 
    \bigg\{ \mathbb{E} [ X | \mathcal{G} ] - \frac{\theta }{2} \mathbb{E} \big[ ( X - \mathbb{E} [ X | \mathcal{G} ] )^{2} \big| \mathcal{G} \big] \bigg\} \\
& = \esssup_{ X \in \mathbb{X} ( R, \mathbb{D}, \mathcal{G} ) } 
    \bigg\{ \mathbb{E} [ X | \mathcal{G} ] - \frac{\theta }{2} \essinf_{ c \in \mathbb{L}^{2}_{\mathcal{G}} ( \Omega ) } \mathbb{E} [ ( X - c )^{2} | \mathcal{G} ] \bigg\} \\
& = \esssup_{ c \in \mathbb{L}^{2}_{\mathcal{G}} ( \Omega ) } \esssup_{ X \in \mathbb{X} ( R, \mathbb{D}, \mathcal{G} ) } 
    \mathbb{E} \bigg[ X - \frac{\theta }{2} ( X - c )^{2} \bigg| \mathcal{G} \bigg] \\
& = \frac{1}{ 2 \theta } 
  + \esssup_{ c \in \mathbb{L}^{2}_{\mathcal{G}} ( \Omega ) }
    \bigg\{ c - \essinf_{ X \in \mathbb{X} ( R, \mathbb{D}, \mathcal{G} ) } 
                \mathbb{E} \bigg[ \frac{\theta }{2} \Big( X - c - \frac{1}{\theta } \Big)^{2} \bigg| \mathcal{G} \bigg] \bigg\}.
\end{align*}
This problem reduction is indeed equivalent to the embedding technique in \cite{Li-Ng-2000,Zhou-Li-2000}; see also \cite{Pedersen-Peskir-2017}.
Denote the maximizer and ($X$-dependent) minimizer in the above second line by ${X}^{*}$ and ${c}^{*} (X)$, respectively; and denote the maximizer and ($c$-dependent) minimizer in the above third line by ${c}^{*}$ and ${X}^{*} (c)$, respectively.
According to the max-max problem given in the above third line with the joint strict concavity of $X - \frac{\theta }{2} ( X - c )^{2}$ w.r.t. $( X,c )$, 
one can conclude that the maximum point is unique, and hence $( {c}^{*}, {X}^{*} ( {c}^{*} ) ) = ( {c}^{*} ( {X}^{*} ), {X}^{*} )$.
This implies that ${c}^{*} = {c}^{*} ( {X}^{*} ) = \mathbb{E} [ {X}^{*} | \mathcal{G} ] = \mathbb{E} [ {X}^{*} ( {c}^{*} ) | \mathcal{G} ]$.

From the previous duality analysis, it follows that
\begin{align*}
&   \essinf_{ X \in \mathbb{X} ( R, \mathbb{D}, \mathcal{G} ) } 
    \mathbb{E} \Big[ \frac{\theta }{2} \Big( X - c - \frac{1}{\theta } \Big)^{2} \Big| \mathcal{G} \Big] \\
& = \esssup_{ Y \in \mathbb{D} } 
    \mathbb{E} \Big[ Y \Big( c + \frac{1}{\theta } - R \Big) - \frac{1}{ 2 \theta } {Y}^{2} \Big| \mathcal{G} \Big] \\
& = \esssup_{ ( y,Y ) \in \mathbb{L}^{2}_{\mathcal{G}} ( \Omega; \mathbb{R}_{+} ) \times \mathbb{D}_{1} } 
    \mathbb{E} \Big[ y Y \Big( c + \frac{1}{\theta } - R \Big) - \frac{1}{ 2 \theta } {y}^{2} {Y}^{2} \Big| \mathcal{G} \Big] \\
& = \frac{1}{ 2 \theta } \mathbb{E} [ {R}^{2} | \mathcal{G} ]
  + \esssup_{ y \in \mathbb{L}^{2}_{\mathcal{G}} ( \Omega; \mathbb{R}_{+} ) } 
    \bigg\{ y \Big( c + \frac{1}{\theta } \Big) - \frac{1}{ 2 \theta } \essinf_{ Y \in \mathbb{D}_{1} } \mathbb{E} [ ( y Y + \theta R )^{2} | \mathcal{G} ] \bigg\}.
\end{align*}
Hence, \cref{eq: MV upper bound :eq} holds.
Moreover, ${X}^{*} (c) = c + \frac{1}{\theta } - \frac{1}{\theta } {y}^{*} (c) {Y}^{*} ( {y}^{*} (c) )$ and 
\begin{equation*}
\mathbb{E} [ ( {X}^{*} (c) - R ) {y}^{*} (c) {Y}^{*} ( {y}^{*} (c) ) | \mathcal{G} ] = 0,
\end{equation*} 
where ${y}^{*} (c)$ maximizes
\begin{equation*}
y \Big( c + \frac{1}{\theta } \Big) - \frac{1}{ 2 \theta } \essinf_{ Y \in \mathbb{D}_{1} } \mathbb{E} [ ( y Y + \theta R )^{2} | \mathcal{G} ]
\end{equation*}
over all $\mathbb{P}$-a.s. non-negative $y \in \mathbb{L}^{2}_{\mathcal{G}} ( \Omega )$ and ${Y}^{*} (y)$ minimizes $\mathbb{E} [ ( y Y + \theta R )^{2} | \mathcal{G} ]$ over all $Y \in \mathbb{D}_{1}$.
For the maximum point ${c}^{*}$, applying the envelope theorem and the first-order derivative optimality condition, one obtains
\begin{equation*}
0 = \frac{d}{dc} \essinf_{ y \in \mathbb{L}^{2}_{\mathcal{G}} ( \Omega; \mathbb{R}_{+} ) } 
                 \Big\{ ( 1 - y ) c - \frac{y}{\theta } + \frac{1}{ 2 \theta } \essinf_{ Y \in \mathbb{D}_{1} } \mathbb{E} [ ( y Y + \theta R )^{2} | \mathcal{G} ] \Big\} \Big|_{ c = {c}^{*} }
  = 1 - {y}^{*} ( {c}^{*} ).
\end{equation*}
Therefore, ${X}^{*} = {X}^{*} ( {c}^{*} ) = \mathbb{E} [ {X}^{*} | \mathcal{G} ] + \frac{1}{\theta } ( 1 - {Y}^{*} (1) )$ and $\mathbb{E} [ ( {X}^{*} - R ) {Y}^{*} (1) | \mathcal{G} ] = 0$. 
This completes the proof.
\end{proof}

\section{Proof of lemmas, theorems and propositions}
\label{app: Proof}

\subsection{Proof of \texorpdfstring{\Cref{lem: identity and truncation}}{Lemma 2.8}}
\label{pf-lem: identity and truncation}

We start by combining the properties of ${U}_{\theta }$ mentioned in \Cref{subsec: MMV preference revisited}, from which one obtains
\begin{equation*}
\partial {U}_{\theta } (f) = \{ 1 - \theta ( f - \mathbb{E} [f] ) \} = \argmin_{ Y \in \mathbb{L}^{2} ( \Omega ) } \{ \mathbb{E} [ Y f ] - {U}_{\theta }^{*} (Y) \}.
\end{equation*}
Given the strict concavity of ${U}_{\theta }^{*}$, the minimizer of the functional $\mathbb{E} [ Y f ] - {U}_{\theta }^{*} (Y)$ is unique.
Consequently, if and only if $1 - \theta ( f - \mathbb{E} [f] ) \in \mathbb{L}^{2}_{ \zeta + } ( \Omega )$, 
it realizes the minimum on the right-hand side of \cref{eq: definition of SMMV :eq}, which is equivalent to the equality ${V}_{ \theta, \zeta } (f) = {U}_{\theta } (f)$.
This proves the first assertion.
Then, due to 
\begin{equation*}
  \lambda - \mathbb{E} \Big[ \Big( f + \frac{\zeta }{\theta } \Big) \wedge \lambda \Big]
= \mathbb{E} \Big[ \Big( \lambda - f - \frac{\zeta }{\theta } \Big)_{+} \Big]
= \int_{ - \infty }^{\lambda } \mathbb{P} \Big( f + \frac{\zeta }{\theta } \le s \Big) ds
\end{equation*}
for any $\lambda \in \mathbb{R}$, the second assertion arises from
\begin{equation*}
\frac{ 1 - \zeta }{\theta } \ge \int_{ - \infty }^{\lambda } \mathbb{P} \Big( f + \frac{\zeta }{\theta } \le s \Big) ds - \frac{ \zeta - \mathbb{E} [ \zeta ] }{\theta }
                            \ge f \wedge \Big( \lambda - \frac{\zeta }{\theta } \Big) - \mathbb{E} \Big[ f \wedge \Big( \lambda - \frac{\zeta }{\theta } \Big) \Big], \quad \forall \lambda \le {\lambda }_{ f, \theta, \zeta },
\end{equation*}
where the first inequality is given by \cref{eq: lambda-equation :eq}, and
\begin{equation*}
\frac{ 1 - \mathbb{E} [ \zeta ] }{\theta } = {\lambda }_{ f, \theta, \zeta } - \mathbb{E} \Big[ \Big( f + \frac{\zeta }{\theta } \Big) \wedge {\lambda }_{ f, \theta, \zeta } \Big] 
                                         \ge {\lambda }_{ f, \theta, \zeta } - \mathbb{E} \Big[ f + \frac{\zeta }{\theta } \Big].
\end{equation*}
Then the ``if'' direction of the third assertion is also proved.
As for the ``only if'' direction of the third assertion, since $f + \frac{\zeta }{\theta } \le \mathbb{E} [f] + \frac{1}{\theta }$, $\mathbb{P}$-a.s., we have
\begin{equation*}
    \frac{ 1 - \mathbb{E} [ \zeta ] }{\theta } 
\ge \esssup \Big\{ f + \frac{\zeta }{\theta } \Big\} - \int_{ - \infty }^{ \esssup \{ f + \frac{\zeta }{\theta } \} } t d \mathbb{P} \Big( f + \frac{\zeta }{\theta } \le t \Big)
  = \int_{ - \infty }^{ \esssup \{ f + \frac{\zeta }{\theta } \} } \mathbb{P} \Big( f + \frac{\zeta }{\theta } \le s \Big) ds.
\end{equation*}
This inequality combined with \cref{eq: lambda-equation :eq} implies that $\esssup \{ f + \frac{\zeta }{\theta } \} \le {\lambda }_{ f, \theta, \zeta }$, which leads to $f + \frac{\zeta }{\theta } \le {\lambda }_{ f, \theta, \zeta }$, $\mathbb{P}$-a.s.
Then the ``only if'' direction of the last assertion also follows.
Finally, we assume that $f \in \mathbb{L}^{2} ( \Omega )$ and $f \wedge ( \lambda - \frac{\zeta }{\theta } ) \in \mathcal{G}_{ \theta, \zeta }$ for some $\lambda > {\lambda }_{ f, \theta, \zeta }$.
Notably, ${\lambda }_{ f \wedge ( \lambda - \frac{\zeta }{\theta } ), \theta, \zeta } \le {\lambda }_{ f, \theta, \zeta }$, with equality holding for all $\lambda \ge {\lambda }_{ f, \theta, \zeta }$, 
since
\begin{align*}
    \frac{ 1 - \mathbb{E} [ \zeta ] }{\theta }
&   = \int_{ - \infty }^{ {\lambda }_{ f, \theta, \zeta } } \mathbb{P} \Big( f + \frac{\zeta }{\theta } \le s \Big) ds
    = \mathbb{E} \bigg[ \int_{\mathbb{R}} {1}_{\{ f + \frac{\zeta }{\theta } \le s \le {\lambda }_{ f, \theta, \zeta } \}} ds \bigg] \\
& \le \mathbb{E} \bigg[ \int_{\mathbb{R}} {1}_{\{ ( f + \frac{\zeta }{\theta } ) \wedge \lambda \le s \le {\lambda }_{ f, \theta, \zeta } \}} ds \bigg]
    = \int_{ - \infty }^{ {\lambda }_{ f, \theta, \zeta } } \mathbb{P} \bigg( f \wedge \Big( \lambda - \frac{\zeta }{\theta } \Big) + \frac{\zeta }{\theta } \le s \bigg) ds,
\end{align*}
with equality holding for $\lambda \ge {\lambda }_{ f, \theta, \zeta }$.
Intuitively speaking, when $\lambda \ge {\lambda }_{ f, \theta, \zeta }$, 
$f \wedge ( \lambda - \frac{\zeta }{\theta } ) + \frac{\zeta }{\theta } $ and $f + \frac{\zeta }{\theta } $ have the same distribution on $( - \infty, {\lambda }_{ f, \theta, \zeta } ]$,
the solution of \cref{eq: lambda-equation :eq} remains unchanged even if $f$ is replaced by $f \wedge ( \lambda - \frac{\zeta }{\theta } )$ therein.
So $\esssup \{ ( f + \frac{\zeta }{\theta } ) \wedge \lambda \} \le {\lambda }_{ f \wedge ( \lambda - \frac{\zeta }{\theta } ), \theta, \zeta } = {\lambda }_{ f, \theta, \zeta } < \lambda $ follows from the third assertion.
As a consequence, $\esssup \{ f + \frac{\zeta }{\theta } \} \le {\lambda }_{ f, \theta, \zeta }$, which proves the ``if'' direction of the last assertion.

\subsection{Proof of \texorpdfstring{\Cref{thm: Gateaux differentiability and equivalent expressions of SMMV}}{Theorem 2.9}}
\label{pf-thm: Gateaux differentiability and equivalent expressions of SMMV}

Since ${V}_{ \theta, \zeta }$ (as a point-wise infimum of some affine functions) is concave and upper semi-continuous, the Fenchel conjugate of ${V}_{ \theta, \zeta }$ is also concave and upper semi-continuous.
It follows from ${V}_{ \theta, \zeta } \ge {U}_{\theta }$ that
\begin{equation*}
{V}_{ \theta, \zeta }^{*} (Y) := \inf_{ f \in \mathbb{L}^{2} ( \Omega ) } \{ \mathbb{E} [ Y f ] - {V}_{ \theta, \zeta } (f) \} 
                             \le \inf_{ f \in \mathbb{L}^{2} ( \Omega ) } \{ \mathbb{E} [ Y f ] - {U}_{\theta } (f) \} = {U}_{\theta }^{*} (Y), 
\quad \forall Y \in \mathbb{L}^{2} ( \Omega ),
\end{equation*} 
which implies that ${V}_{ \theta, \zeta }^{*} (Y) = - \infty $ for $\mathbb{E} [Y] \ne 1$.
For $Y \in \mathbb{L}^{2}_{ \zeta + } ( \Omega )$,
the converse inequality ${V}_{ \theta, \zeta }^{*} (Y) \ge {U}_{\theta }^{*} (Y)$ follows from
$\mathbb{E} [ Y f ] - {V}_{ \theta, \zeta } (f) \ge {U}_{\theta }^{*} (Y)$ given by \cref{eq: definition of SMMV :eq} for any $f \in \mathbb{L}^{2} ( \Omega )$.
In terms of $Y \notin \mathbb{L}^{2}_{ \zeta + } ( \Omega )$, there exists $\varepsilon > 0$ such that $\mathbb{P} ( Y \le \zeta - \varepsilon ) > 0$,
and then for $f = c {1}_{\{ Y \le \zeta - \varepsilon \}} \ge 0 \in \mathcal{G}_{ \zeta, \theta  }$ with $c \in \mathbb{R}_{+}$ we have
\begin{equation*}
\left\{ \begin{aligned}
        & \mathbb{E} [ ( Y - \zeta ) f ] \le - c \varepsilon \mathbb{P} ( Y \le \zeta - \varepsilon ) \downarrow - \infty, \quad as \quad c \uparrow \infty; \\
        & {V}_{ \theta, \zeta } (f) - \mathbb{E} [ \zeta f ] \ge {V}_{ \theta, \zeta } (0) = {U}_{\theta } (0) = 0,
        \end{aligned} \right.
\end{equation*}
which implies that ${V}_{ \theta, \zeta }^{*} (Y) = - \infty $. 
Summing up, one obtains
\begin{equation*}
{V}_{ \theta, \zeta }^{*} (Y) 
= \left\{ \begin{aligned} 
  & - \frac{1}{2 \theta } ( \mathbb{E} [ {Y}^{2} ] - 1 ), && if \ Y \in \mathbb{L}^{2}_{ \zeta + } ( \Omega ) ~and~ \mathbb{E} [Y] = 1; \\
  & - \infty, && otherwise. 
  \end{aligned} \right.
\end{equation*}
Applying Fenchel-Moreau theorem (cf. \cite[Theorem 13.37]{Bauschke-Combettes-2017}) to ${V}_{ \theta, \zeta }$ yields
\begin{equation*}
{V}_{ \theta, \zeta } (f) = \inf_{ Y \in \mathbb{L}^{2} ( \Omega ) } \{ \mathbb{E} [ Y f ] - {V}_{ \theta, \zeta }^{*} (Y) \}, \quad \forall f \in \mathbb{L}^{2} ( \Omega ).
\end{equation*}
Notably, the above statement also arises from artificially assigning ${U}_{\theta }^{*} (Y) = - \infty $ for any $Y \notin \mathbb{L}^{2}_{ \zeta + } ( \Omega )$ in \cref{eq: definition of SMMV :eq};
however, the Fenchel conjugation of $( {V}_{ \theta, \zeta }, {V}_{ \theta, \zeta }^{*} )$ gives some additional information about superdifferential $\partial {V}_{ \theta, \zeta }$ as follows.
In fact, $Y \in \partial {V}_{ \theta, \zeta } (f)$ is equivalent to the following statements:
\begin{itemize}
\item ${V}_{ \theta, \zeta } (g) \le {V}_{ \theta, \zeta } (f) + \mathbb{E} [ Y ( g - f ) ]$ for any $g \in \mathbb{L}^{2} ( \Omega )$;
\item $\mathbb{E} [ Y f ] - {V}_{ \theta, \zeta } (f) \le \inf_{ g \in \mathbb{L}^{2} ( \Omega ) } \{ \mathbb{E} [ Y g ] - {V}_{ \theta, \zeta } (g) \} \equiv {V}_{ \theta, \zeta }^{*} (Y)$;
\item ${V}_{ \theta, \zeta } (f) = \mathbb{E} [ Y f ] - {V}_{ \theta, \zeta }^{*} (Y)$.
\end{itemize}
Consequently,
\begin{equation}\label{eq: superdifferential of SMMV :eq}
\partial {V}_{ \theta, \zeta } (f) 
= \argmin_{ Y \in \mathbb{L}^{2} ( \Omega ) } \{ \mathbb{E} [ Y f ] - {V}_{ \theta, \zeta }^{*} (Y) \}
= \argmin_{ Y \in \mathbb{L}^{2}_{ \zeta + } ( \Omega ), \mathbb{E}^{P} [Y] = 1 } \Big\{ \mathbb{E} [ Y f ] + \frac{1}{2 \theta } ( \mathbb{E} [ {Y}^{2} ] - 1 ) \Big\}
\end{equation}
is at most a singleton.
If the minimizer exists, then ${V}_{ \theta, \zeta }$ is G\^ateaux differentiable according to \cite[Proposition 1.8, p. 5]{Phelps-1993},
and $d {V}_{ \theta, \zeta } (f)$ realizes the minimum on the right-hand side of \cref{eq: definition of SMMV :eq}, which aligns with the result obtained by heuristically applying envelope theorem.

Now we solve the minimization problem in \cref{eq: superdifferential of SMMV :eq} by Lagrange duality method (noting that Lagrange multiplier method is also applicable). 
Let us proceed with the following min-max inequality:
\begin{align*}
      \inf_{ Y \in \mathbb{L}^{2}_{ \zeta + } ( \Omega ), \mathbb{E} [Y] = 1 } \mathbb{E} \Big[ f Y + \frac{1}{ 2 \theta } {Y}^{2} \Big] 
&   = \inf_{ Y \in \mathbb{L}^{2}_{ \zeta + } ( \Omega ) } \sup_{ \lambda \in \mathbb{R} } \Big\{ \mathbb{E} \Big[ f Y + \frac{1}{2 \theta } {Y}^{2} - \lambda Y \Big] + \lambda \Big\} \\
& \ge \sup_{ \lambda \in \mathbb{R} } \bigg\{ \inf_{ Y \in \mathbb{L}^{2}_{ \zeta + } ( \Omega ) } \mathbb{E} \Big[ f Y + \frac{1}{2 \theta } {Y}^{2} - \lambda Y \Big] + \lambda \bigg\} \\
&   = \sup_{ \lambda \in \mathbb{R} } \bigg\{ \mathbb{E} \bigg[ \inf_{ Y \ge \zeta } \Big\{ \frac{1}{2 \theta } {Y}^{2} + ( f - \lambda ) Y \Big\} \bigg] + \lambda \bigg\},
\end{align*}
for which the unique minimizer $Y = \zeta + \theta ( \lambda - f - \frac{\zeta }{\theta } )_{+} \in \mathbb{L}^{2}_{ \zeta + } ( \Omega )$ can be easily seen from
\begin{equation*} 
  \inf_{ Y \ge \zeta } \Big\{ \frac{1}{2 \theta } {Y}^{2} + ( f - \lambda ) Y \Big\}
= \inf_{ Y - \zeta \ge 0 } \Big\{ \frac{1}{2 \theta } ( Y - \zeta )^{2} - \Big( \lambda - f - \frac{\zeta }{\theta } \Big) ( Y - \zeta ) \Big\} + {\zeta }^{2} - \zeta ( \lambda - f ).
\end{equation*}
Consequently,
\begin{align*}
& \argmax_{ \lambda \in \mathbb{R} } \bigg\{ \inf_{ Y \in \mathbb{L}^{2}_{ \zeta + } ( \Omega ) } \mathbb{E} \Big[ f Y + \frac{1}{2 \theta } {Y}^{2} - \lambda Y \Big] + \lambda \bigg\} \\
& = \argmax_{ \lambda \in \mathbb{R} } \bigg\{ \lambda ( 1 - \mathbb{E} [ \zeta ] )
                                             - \frac{\theta }{2 } \mathbb{E} \Big[ \Big( \lambda - f - \frac{\zeta }{\theta } \Big)^{2} {1}_{\{ f + \frac{\zeta }{\theta } \le \lambda \}} \Big] \bigg\} \\
& = \argmax_{ \lambda \in \mathbb{R} } \bigg\{ \lambda ( 1 - \mathbb{E} [ \zeta ] )
                                             - \frac{\theta }{2} \int_{ - \infty }^{\lambda } ( \lambda - s )^{2} d \mathbb{P} \Big( f + \frac{\zeta }{\theta } \le s \Big) \bigg\} \\
& = \argmax_{ \lambda \in \mathbb{R} } \bigg\{ \lambda ( 1 - \mathbb{E} [ \zeta ] )
                                             - \theta \int_{ - \infty }^{\lambda } ( \lambda - s ) \mathbb{P} \Big( f + \frac{\zeta }{\theta } \le s \Big) ds \bigg\},
\end{align*}
where the last equality results from
\begin{align*}
  \frac{1}{2} \int_{ - \infty }^{\lambda } ( \lambda - t )^{2} d \mathbb{P} \Big( f + \frac{\zeta }{\theta } \le t \Big)
& = \int_{ - \infty }^{\lambda } d \mathbb{P} \Big( f + \frac{\zeta }{\theta } \le t \Big) \int_{t}^{\lambda } ( \lambda - s ) ds \\
& = \int_{ - \infty }^{\lambda } ( \lambda - s ) ds \int_{ - \infty }^{s} d \mathbb{P} \Big( f + \frac{\zeta }{\theta } \le t \Big). 
\end{align*}
By the first-order derivative conditions $0 = 1 - \mathbb{E} [ \zeta ] - \theta \int_{ - \infty }^{\lambda } \mathbb{P} ( f + \frac{\zeta }{\theta } \le s ) ds$,
of which the right-hand side is decreasing in $\lambda $ and strictly decreasing on $( \essinf \{ f + \frac{\zeta }{\theta } \}, + \infty )$,
we arrive at the unique maximizer ${\lambda }_{ f, \theta, \zeta }$ given by \cref{eq: lambda-equation :eq}.
Therefore, by
\begin{equation*}
\left\{ \begin{aligned}
&   \inf_{ Y \in \mathbb{L}^{2}_{ \zeta + } ( \Omega ), \mathbb{E} [Y] = 1 } \mathbb{E} \Big[ f Y + \frac{1}{ 2 \theta } {Y}^{2} \Big] 
\ge \mathbb{E} \Big[ f Y + \frac{1}{ 2 \theta } {Y}^{2} \Big] \Big|_{ Y = \zeta + \theta ( {\lambda }_{ f, \theta, \zeta } - f - \frac{\zeta }{\theta } )_{+} }, \\
& \mathbb{E} \Big[ \zeta + \theta \Big( {\lambda }_{ f, \theta, \zeta } - f - \frac{\zeta }{\theta } \Big)_{+} \Big] 
= \mathbb{E} [ \zeta ] + \theta \bigg( {\lambda }_{ f, \theta, \zeta } - \mathbb{E} \Big[ \Big( f + \frac{\zeta }{\theta } \Big) \wedge {\lambda }_{ f, \theta, \zeta } \Big] \bigg) = 1,
\end{aligned} \right.
\end{equation*}
we conclude that $Y = \zeta + \theta ( {\lambda }_{ f, \theta, \zeta } - f - \frac{\zeta }{\theta } )_{+}$ is the unique minimizer for \cref{eq: superdifferential of SMMV :eq},
and hence $d {V}_{ \theta, \zeta } (f) = \zeta + \theta ( {\lambda }_{ f, \theta, \zeta } - f - \frac{\zeta }{\theta } )_{+}$.
Furthermore, one can obtain
\begin{align*}
  {V}_{ \theta, \zeta } (f) 
= \max_{ \lambda \in \mathbb{R} } \bigg\{ \lambda ( 1 - \mathbb{E} [ \zeta ] )
                                        - \theta \int_{ - \infty }^{\lambda } ( \lambda - s ) \mathbb{P} \Big( f + \frac{\zeta }{\theta } \le s \Big) ds \bigg\}
+ \mathbb{E} [ f \zeta ] + \frac{1}{ 2 \theta } \mathbb{E} [ {\zeta }^{2} ] - \frac{1}{ 2 \theta }, 
\end{align*}
and then we immediately arrive at \cref{eq: direct calculation form of SMMV :eq}.
Then, the second line of our desired expression for ${V}_{ \theta, \zeta }$ follows, as
\begin{align*}
& \theta \int_{ - \infty }^{ {\lambda }_{ f, \theta, \zeta } } s \mathbb{P} \Big( f + \frac{\zeta }{\theta } \le s \Big) ds \\
& = \theta \int_{ - \infty }^{ {\lambda }_{ f, \theta, \zeta } } s ds \int_{ - \infty }^{s} d \mathbb{P} \Big( f + \frac{\zeta }{\theta } \le t \Big) \\
& = \frac{\theta }{2} \int_{ - \infty }^{ {\lambda }_{ f, \theta, \zeta } } ( {\lambda }_{ f, \theta, \zeta }^{2} - {s}^{2} ) d \mathbb{P} \Big( f + \frac{\zeta }{\theta } \le s \Big) \\
& = \frac{\theta }{2} {\lambda }_{ f, \theta, \zeta }^{2} 
  - \frac{\theta }{2} {\lambda }_{ f, \theta, \zeta }^{2} \mathbb{P} \Big( f + \frac{\zeta }{\theta } < {\lambda }_{ f, \theta, \zeta } \Big)
  - \frac{\theta }{2} \int_{ - \infty }^{ {\lambda }_{ f, \theta, \zeta } } {s}^{2} d \mathbb{P} \Big( f + \frac{\zeta }{\theta } \le s \Big) \\
& = \frac{\theta }{2} \bigg( \frac{ 1 - \mathbb{E} [ \zeta ] }{\theta } + \mathbb{E} \Big[ \Big( f + \frac{\zeta }{\theta } \Big) \wedge {\lambda }_{ f, \theta, \zeta } \Big] \bigg)^{2}
  - \frac{\theta }{2} \mathbb{E} \bigg[ \Big| \Big( f + \frac{\zeta }{\theta } \Big) \wedge {\lambda }_{ f, \theta, \zeta } \Big|^{2} \bigg] \\
& = \frac{\theta }{2} \Big( \frac{ 1 - \mathbb{E} [ \zeta ] }{\theta } \Big)^{2} 
  - \mathbb{E} [ \zeta ] \Big( {\lambda }_{ f, \theta, \zeta } - \frac{ 1 - \mathbb{E} [ \zeta ] }{\theta } \Big) 
  + {U}_{\theta } \bigg( \Big( f + \frac{\zeta }{\theta } \Big) \wedge {\lambda }_{ f, \theta, \zeta } \bigg) \\
& = \frac{ 1 - | \mathbb{E} [ \zeta ] |^{2} }{ 2 \theta } - {\lambda }_{ f, \theta, \zeta } \mathbb{E} [ \zeta ] 
  + {U}_{\theta } \bigg( \Big( f + \frac{\zeta }{\theta } \Big) \wedge {\lambda }_{ f, \theta, \zeta } \bigg),
\end{align*}
where the fourth and fifth equalities both arise from \cref{eq: lambda-equation :eq}.
Alternatively, proceeding with the above fourth equality, we substitute 
\begin{align*}
  \theta \int_{ - \infty }^{ {\lambda }_{ f, \theta, \zeta } } s \mathbb{P} \Big( f + \frac{\zeta }{\theta } \le s \Big) ds
& = \frac{\theta }{2} \bigg( \frac{1}{\theta } + \mathbb{E} \Big[ f \wedge \Big( {\lambda }_{ f, \theta, \zeta } - \frac{\zeta }{\theta } \Big) \Big] \bigg)^{2}
  - \frac{\theta }{2} \mathbb{E} \bigg[ \Big| f \wedge \Big( {\lambda }_{ f, \theta, \zeta } - \frac{\zeta }{\theta } \Big) + \frac{\zeta }{\theta } \Big|^{2} \bigg] \\
& = \frac{ 1 - \mathbb{E} [ {\zeta }^{2} ] }{ 2 \theta } 
  + {U}_{\theta } \bigg( f \wedge \Big( {\lambda }_{ f, \theta, \zeta } - \frac{\zeta }{\theta } \Big) \bigg)
  - \mathbb{E} \bigg[ \bigg( f \wedge \Big( {\lambda }_{ f, \theta, \zeta } - \frac{\zeta }{\theta } \Big) \bigg) \zeta \bigg]
\end{align*}
into the right-hand side of \cref{eq: direct calculation form of SMMV :eq} and then obtain \cref{eq: modificed MV form of SMMV :eq}, which leads to the last line of our desired result.

\subsection{Proof of \texorpdfstring{\Cref{thm: minimality and maximality of SMMV}}{Proposition 2.10}}
\label{pf-thm: minimality and maximality of SMMV}

For the minimality, we assume by contradiction that $V(f) < {V}_{ \theta, \zeta } (f)$, $V |_{ \mathcal{G}_{ \theta, \zeta } } = {U}_{\theta } |_{ \mathcal{G}_{ \theta, \zeta } }$ 
and ${\xi }_{f} \in \partial V (f) \cap \mathbb{L}^{2}_{\zeta + } ( \Omega ) \ne \varnothing$ for some $f \notin \mathcal{G}_{ \theta, \zeta }$. 
Then, it follows from the expression \cref{eq: modificed MV form of SMMV :eq} that
\begin{align*}
    {U}_{\theta } \bigg( f \wedge \Big( {\lambda }_{ f, \theta, \zeta } - \frac{\zeta }{\theta } \Big) \bigg)
& = {V}_{ \theta, \zeta } (f) - \mathbb{E} \bigg[ \bigg( f - f \wedge \Big( {\lambda }_{ f, \theta, \zeta } - \frac{\zeta }{\theta } \Big) \bigg) \zeta \bigg] \\
& > V(f) - \mathbb{E} \bigg[ \bigg( f - f \wedge \Big( {\lambda }_{ f, \theta, \zeta } - \frac{\zeta }{\theta } \Big) \bigg) \zeta \bigg] \\
& \ge V(f) - \mathbb{E} \bigg[ \bigg( f - f \wedge \Big( {\lambda }_{ f, \theta, \zeta } - \frac{\zeta }{\theta } \Big) \bigg) {\xi }_{f} \bigg] \\
& \ge V \bigg( f \wedge \Big( {\lambda }_{ f, \theta, \zeta } - \frac{\zeta }{\theta } \Big) \bigg),
\end{align*}
which immediately contradicts $V |_{ \mathcal{G}_{ \theta, \zeta } } = {U}_{\theta } |_{ \mathcal{G}_{ \theta, \zeta } }$. 

In terms of the maximality, by \cref{eq: modificed MV form of SMMV :eq} and $f \wedge ( {\lambda }_{ f, \theta, \zeta } - \frac{\zeta }{\theta } ) \in \mathcal{G}_{ \theta, \zeta }$, we have
\begin{equation*}
{V}_{ \theta, \zeta } (f) \le \sup_{ g \in \mathcal{G}_{ \theta, \zeta }, g \le f } \{ {U}_{\theta } (g) + \mathbb{E} [ ( f - g ) \zeta ] \}.
\end{equation*}
On the other hand, since $d {U}_{\theta } (f) \in \partial {U}_{\theta } (f)$, $d {U}_{\theta } |_{ \mathcal{G}_{ \theta, \zeta } } = d {V}_{ \theta, \zeta } |_{ \mathcal{G}_{ \theta, \zeta } }$ and 
${\lambda }_{ f \wedge ( {\lambda }_{ f, \theta, \zeta } - \frac{\zeta }{\theta } ), \theta, \zeta } = {\lambda }_{ f, \theta, \zeta }$, we have
\begin{align*}
{U}_{\theta } (g) & \le {U}_{\theta } \bigg( f \wedge \Big( {\lambda }_{ f, \theta, \zeta } - \frac{\zeta }{\theta } \Big) \bigg) 
                      + \mathbb{E} \bigg[ \bigg( g - f \wedge \Big( {\lambda }_{ f, \theta, \zeta } - \frac{\zeta }{\theta } \Big) \bigg) 
                                          \bigg( \zeta + \theta \Big( {\lambda }_{ f, \theta, \zeta } - f - \frac{\zeta }{\theta } \Big)_{+} \bigg) \bigg] \\
                  & \le {U}_{\theta } \bigg( f \wedge \Big( {\lambda }_{ f, \theta, \zeta } - \frac{\zeta }{\theta } \Big) \bigg) 
                      + \mathbb{E} \bigg[ \bigg( g - f \wedge \Big( {\lambda }_{ f, \theta, \zeta } - \frac{\zeta }{\theta } \Big) \bigg) \zeta \bigg] 
\end{align*}
for any $g \in \mathcal{G}_{ \theta, \zeta }$ with $g \le f$.
Abstracting $\mathbb{E} [ ( g - f ) \zeta ]$ from both sides of the above inequality, 
then applying \cref{eq: modificed MV form of SMMV :eq} to the right-hand side and taking supremum on the left-hand side over all $g \in \mathcal{G}_{ \theta, \zeta }$ with $g \le f$,
we obtain
\begin{equation*}
\sup_{ g \in \mathcal{G}_{ \theta, \zeta }, g \le f } \{ {U}_{\theta } (g) + \mathbb{E} [ ( f - g ) \zeta ] \} \le {V}_{ \theta, \zeta } (f).
\end{equation*}
Summing up, we have proved the first desired equality. 
Moreover, the previous proof is also valid, if we extend the domain for $g$ to $\mathbb{L}^{2} ( \Omega )$ with $g \le f$. 
Therefore, the second desired equality holds.

\subsection{Proof of \texorpdfstring{\Cref{thm: MUA and SSD for SMMV}}{Proposition 2.12}}
\label{pf-thm: MUA and SSD for SMMV}

For the first property about more uncertainty averse,
the ``only if'' direction is obvious, due to ${V}_{ \theta, \zeta } (c) \equiv c$ and the monotonicity of ${V}_{ \theta, \zeta } (f)$ in $\theta $.
To check the ``if'' direction, we assume by contradiction that $\theta < \hat{\theta }$.
On the one hand, we choose $f \in \mathbb{L}^{2} ( \Omega )$ arbitrarily and $c = {V}_{ \theta, \zeta } (f)$,
so that ${V}_{ \theta, \zeta } (f) = {V}_{ \theta, \zeta } ( {V}_{ \theta, \zeta } (f) )$ leads to ${V}_{ \hat{\theta }, \zeta } (f) \ge {V}_{ \hat{\theta }, \zeta } ( {V}_{ \theta, \zeta } (f) ) = {V}_{ \theta, \zeta } (f)$.
On the other hand, one can find $f \in \mathbb{L}^{2} ( \Omega )$ such that $\Var [ f \wedge ( {\lambda }_{ f, \hat{\theta }, \zeta } - \frac{\zeta }{ \hat{\theta } } ) ] > 0$,
since $\{ f \le {\lambda }_{ f, \hat{\theta }, \zeta } - \frac{\zeta }{ \hat{\theta } } \}$ is not a $\mathbb{P}$-null subset due to the second assertion in \Cref{lem: identity and truncation},
$\mathbb{P} ( f \le {\lambda }_{ f, \hat{\theta }, \zeta } - \frac{\zeta }{ \hat{\theta } } ) = 1$ results in $f \in \mathcal{G}_{ \hat{\theta }, \zeta }$ due to the third assertion in \Cref{lem: identity and truncation},
and a non-trivial $\mathcal{F}$ implies that $\mathcal{G}_{ \hat{\theta }, \zeta } \subsetneq \mathbb{L}^{2} ( \Omega )$.
As a consequence, a contradiction arises from
\begin{align*}
  {V}_{ \theta, \zeta } (f) 
& \ge {U}_{\theta } \bigg( f \wedge \Big( {\lambda }_{ f, \hat{\theta }, \zeta } - \frac{\zeta }{\theta } \Big) \bigg)
    + \mathbb{E} \bigg[ \bigg( f - f \wedge \Big( {\lambda }_{ f, \hat{\theta }, \zeta } - \frac{\zeta }{\theta } \Big) \bigg) \zeta \bigg] \\
&   > {U}_{ \hat{\theta } } \bigg( f \wedge \Big( {\lambda }_{ f, \hat{\theta }, \zeta } - \frac{\zeta }{\theta } \Big) \bigg)
    + \mathbb{E} \bigg[ \bigg( f - f \wedge \Big( {\lambda }_{ f, \hat{\theta }, \zeta } - \frac{\zeta }{\theta } \Big) \bigg) \zeta \bigg] \\
&   = {V}_{ \hat{\theta }, \zeta } (f),
\end{align*}
where the first inequality and the last equality follow from the second assertion in \Cref{thm: minimality and maximality of SMMV} and \cref{eq: modificed MV form of SMMV :eq}, respectively.

Now we prove the second property about second-order stochastic dominance.
Firstly, ${\lambda }_{ g, \theta, \zeta } \le {\lambda }_{ f, \theta, \zeta }$ is a straightforward result of 
\begin{equation*}
  \int_{ - \infty }^{ {\lambda }_{ g, \theta, \zeta } } \mathbb{P} \Big( g + \frac{\zeta }{\theta } \le s \Big) ds
= \int_{ - \infty }^{ {\lambda }_{ f, \theta, \zeta } } \mathbb{P} \Big( f + \frac{\zeta }{\theta } \le s \Big) ds
\le \int_{ - \infty }^{ {\lambda }_{ f, \theta, \zeta } } \mathbb{P} \Big( g + \frac{\zeta }{\theta } \le s \Big) ds,
\end{equation*}
where the first equality follows from \cref{eq: lambda-equation :eq} and the last inequality follows from the second-order stochastic dominance condition, i.e.
\begin{equation*}
\int_{ - \infty }^{t} \mathbb{P} \Big( f + \frac{\zeta }{\theta } \le s \Big) ds \le \int_{ - \infty }^{t} \mathbb{P} \Big( g + \frac{\zeta }{\theta } \le s \Big) ds, \quad \forall t \in \mathbb{R}.
\end{equation*} 
It also follows from \cref{eq: lambda-equation :eq} that
\begin{align*}
&   \theta \int_{ - \infty }^{ {\lambda }_{ f, \theta, \zeta } } s \mathbb{P} \Big( f + \frac{\zeta }{\theta } \le s \Big) ds \\
& = \theta {\lambda }_{ f, \theta, \zeta } \int_{ - \infty }^{ {\lambda }_{ f, \theta, \zeta } } \mathbb{P} \Big( f + \frac{\zeta }{\theta } \le s \Big) ds
  - \theta \int_{ - \infty }^{ {\lambda }_{ f, \theta, \zeta } } \mathbb{P} \Big( f + \frac{\zeta }{\theta } \le s \Big) ds \int_{s}^{ {\lambda }_{ f, \theta, \zeta } } dt \\
& = {\lambda }_{ f, \theta, \zeta } ( 1 - \mathbb{E} [ \zeta ] ) 
  - \theta \int_{ - \infty }^{ {\lambda }_{ f, \theta, \zeta } } dt \int_{ - \infty }^{t} \mathbb{P} \Big( f + \frac{\zeta }{\theta } \le s \Big) ds.
\end{align*}
Hence, given the expression \cref{eq: direct calculation form of SMMV :eq} and the second-order stochastic dominance condition, one obtains 
\begin{align*}
& {V}_{ \theta, \zeta } (f) - {V}_{ \theta, \zeta } (g) \\
& = \mathbb{E} [ ( f - g ) \zeta ]
  + ( {\lambda }_{ f, \theta, \zeta } - {\lambda }_{ g, \theta, \zeta } ) ( 1 - \mathbb{E} [ \zeta ] ) 
  - \theta \int_{ {\lambda }_{ g, \theta, \zeta } }^{ {\lambda }_{ f, \theta, \zeta } } dt \int_{ - \infty }^{t} \mathbb{P} \Big( f + \frac{\zeta }{\theta } \le s \Big) ds \\
& \quad + \theta \int_{ - \infty }^{ {\lambda }_{ g, \theta, \zeta } } \bigg( \int_{ - \infty }^{t} \mathbb{P} \Big( g + \frac{\zeta }{\theta } \le s \Big) ds 
                                                                            - \int_{ - \infty }^{t} \mathbb{P} \Big( f + \frac{\zeta }{\theta } \le s \Big) ds \bigg) dt \\
& \ge \mathbb{E} [ ( f - g ) \zeta ]
    + \theta \int_{ {\lambda }_{ g, \theta, \zeta } }^{ {\lambda }_{ f, \theta, \zeta } } dt \int_{t}^{ {\lambda }_{ f, \theta, \zeta } } \mathbb{P} \Big( f + \frac{\zeta }{\theta } \le s \Big) ds \\
& \ge \mathbb{E} [ ( f - g ) \zeta ].
\end{align*}

\subsection{Proof of \texorpdfstring{\Cref{thm: optimality condition for SP}}{Theorem 3.1}}
\label{pf-thm: optimality condition for SP}

${V}_{ \theta, \zeta } ( {X}_{ \vec{\alpha } } )$ is jointly concave in $\vec{\alpha }$, since ${V}_{ \theta, \zeta }$ is concave and ${X}_{ \vec{\alpha } }$ is affine in $\vec{\alpha }$;
and it follows from \Cref{thm: Gateaux differentiability and equivalent expressions of SMMV} that
\begin{equation*} 
{V}_{ \theta, \zeta } ( \vec{X}_{ \vec{\alpha } } + \varepsilon \vec{h} ) - {V}_{ \theta, \zeta } ( \vec{X}_{ \vec{\alpha } } )
= \bigg\langle \vec{h}, \mathbb{E} \bigg[ ( \vec{R} - r \vec{1} ) \bigg( \zeta + \theta \Big( {\lambda }_{ \vec{\alpha } } - {X}_{ \vec{\alpha } } - \frac{\zeta }{\theta } \Big)_{+} \bigg) \bigg] \bigg\rangle \varepsilon + o ( \varepsilon )
\end{equation*}
for any $( \varepsilon, \vec{h} ) \in \mathbb{R}_{+} \times \mathbb{R}^{n}$.
Therefore, the following gradient condition,
\begin{align}\label{eq: gradient condition for SP :eq}
\vec{0} 
& = \mathbb{E} \bigg[ ( \vec{R} - r \vec{1} ) \bigg( \zeta + \theta \Big( {\lambda }_{ \vec{\alpha }^{*} } - {X}_{ \vec{\alpha }^{*} } - \frac{\zeta }{\theta } \Big)_{+} \bigg) \bigg] \\
\notag
& \equiv \mathbb{E} [ ( \vec{R} - r \vec{1} ) \zeta ]
       + \theta {\lambda }_{ \vec{\alpha }^{*} } \mathbb{P} \Big( {X}_{ \vec{\alpha }^{*} } + \frac{\zeta }{\theta } \le {\lambda }_{ \vec{\alpha }^{*} } \Big)
                                                 \mathbb{E} \Big[ \vec{R} - r \vec{1} \Big| {X}_{ \vec{\alpha }^{*} } + \frac{\zeta }{\theta } \le {\lambda }_{ \vec{\alpha }^{*} } \Big] \\
\notag
&\quad - \theta \mathbb{P} \Big( {X}_{ \vec{\alpha }^{*} } + \frac{\zeta }{\theta } \le {\lambda }_{ \vec{\alpha }^{*} } \Big)
                \mathbb{E} \Big[ ( \vec{R} - r \vec{1} ) \Big( {X}_{ \vec{\alpha }^{*} } + \frac{\zeta }{\theta } \Big) \Big| {X}_{ \vec{\alpha }^{*} } + \frac{\zeta }{\theta } \le {\lambda }_{ \vec{\alpha }^{*} } \Big],
\end{align}
is necessary and sufficient to realize the maximum in \cref{eq: SP :eq}.
Notably, the second equation in \cref{eq: optimality condition for SP :eq} is a re-expression of \cref{eq: lambda-equation :eq}, 
so the rest of this proof is to show the equivalence between \cref{eq: gradient condition for SP :eq} and the first equation in \cref{eq: optimality condition for SP :eq}.
Applying the iterated conditioning to \cref{eq: lambda-equation :eq} with $f = \vec{X}_{ \vec{\alpha }^{*} }$ yields
\begin{equation*}
  {\lambda }_{ \vec{\alpha }^{*} } \mathbb{P} \Big( {X}_{ \vec{\alpha }^{*} } + \frac{\zeta }{\theta } \le {\lambda }_{ \vec{\alpha }^{*} } \Big)
= \frac{\kappa }{\theta } + \mathbb{P} \Big( {X}_{ \vec{\alpha }^{*} } + \frac{\zeta }{\theta } \le {\lambda }_{ \vec{\alpha }^{*} } \Big)
                            \mathbb{E} \Big[ {X}_{ \vec{\alpha }^{*} } + \frac{\zeta }{\theta } \Big| {X}_{ \vec{\alpha }^{*} } + \frac{\zeta }{\theta } \le {\lambda }_{ \vec{\alpha }^{*} } \Big].
\end{equation*}
Substituting this result back into \cref{eq: gradient condition for SP :eq} and rearranging the terms, one obtains
\begin{align*}
& \mathbb{E} [ ( \vec{R} - r \vec{1} ) \zeta ]
+ \kappa \mathbb{E} \Big[ \vec{R} - r \vec{1} \Big| {X}_{ \vec{\alpha }^{*} } + \frac{\zeta }{\theta } \le {\lambda }_{ \vec{\alpha }^{*} } \Big] \\
& = \theta \mathbb{P} \Big( {X}_{ \vec{\alpha }^{*} } + \frac{\zeta }{\theta } \le {\lambda }_{ \vec{\alpha }^{*} } \Big)
           \mathbb{E} \Big[ ( \vec{R} - r \vec{1} ) \Big( {X}_{ \vec{\alpha }^{*} } + \frac{\zeta }{\theta } \Big) \Big| {X}_{ \vec{\alpha }^{*} } + \frac{\zeta }{\theta } \le {\lambda }_{ \vec{\alpha }^{*} } \Big] \\
& \quad - \theta \mathbb{P} \Big( {X}_{ \vec{\alpha }^{*} } + \frac{\zeta }{\theta } \le {\lambda }_{ \vec{\alpha }^{*} } \Big)
                 \mathbb{E} \Big[ {X}_{ \vec{\alpha }^{*} } + \frac{\zeta }{\theta } \Big| {X}_{ \vec{\alpha }^{*} } + \frac{\zeta }{\theta } \le {\lambda }_{ \vec{\alpha }^{*} } \Big]
                 \mathbb{E} \Big[ \vec{R} - r \vec{1} \Big| {X}_{ \vec{\alpha }^{*} } + \frac{\zeta }{\theta } \le {\lambda }_{ \vec{\alpha }^{*} } \Big] \\ 
& = \theta \mathbb{P} \Big( {X}_{ \vec{\alpha }^{*} } + \frac{\zeta }{\theta } \le {\lambda }_{ \vec{\alpha }^{*} } \Big)
           \Cov \Big[ \vec{R} - r \vec{1}, {X}_{ \vec{\alpha }^{*} } + \frac{\zeta }{\theta } \Big| {X}_{ \vec{\alpha }^{*} } + \frac{\zeta }{\theta } \le {\lambda }_{ \vec{\alpha }^{*} } \Big] \\
& = \mathbb{P} \Big( {X}_{ \vec{\alpha }^{*} } + \frac{\zeta }{\theta } \le {\lambda }_{ \vec{\alpha }^{*} } \Big)
    \Big( \Var \Big[ \vec{R} \Big| {X}_{ \vec{\alpha }^{*} } + \frac{\zeta }{\theta } \le {\lambda }_{ \vec{\alpha }^{*} } \Big] \theta \vec{\alpha }^{*}
        + \Cov \Big[ \vec{R}, \zeta \Big| {X}_{ \vec{\alpha }^{*} } + \frac{\zeta }{\theta } \le {\lambda }_{ \vec{\alpha }^{*} } \Big] \Big).
\end{align*}
Hence, \cref{eq: gradient condition for SP :eq} gives the first equation in \cref{eq: optimality condition for SP :eq}, and vice versa.

\subsection{Proof of \texorpdfstring{\Cref{lem: saddle point in SP}}{Proposition 3.2}}
\label{pf-lem: saddle point in SP}

The ``if'' direction is obvious according to the definition of saddle point.
To show the ``only if'' direction, we consider the value of \cref{eq: max-min problem for SP :eq} denoted by $\mathcal{M}$.
By comparing \cref{eq: SP :eq} with \cref{eq: modificed MV form of SMMV :eq}, 
in conjunction with the identities ${X}_{ \vec{\alpha }^{*} } \wedge ( {\lambda }_{ \vec{\alpha }^{*} } - \frac{\zeta }{\theta } ) = {\lambda }_{ \vec{\alpha }^{*} } - \frac{\zeta }{\theta } - \frac{\kappa }{\theta } {Z}_{*}$ 
and $( \theta {X}_{ \vec{\alpha }^{*} } + \zeta + \kappa {Z}_{*} - \theta {\lambda }_{ \vec{\alpha }^{*} } ) {Z}_{*} = 0$, one obtains
\begin{align*}
  \mathcal{M} 
& = \theta {U}_{\theta } \bigg( {X}_{ \vec{\alpha }^{*} } \wedge \Big( {\lambda }_{ \vec{\alpha }^{*} } - \frac{\zeta }{\theta } \Big) \bigg)
  + \theta \mathbb{E} \bigg[ \bigg( {X}_{ \vec{\alpha }^{*} } - {X}_{ \vec{\alpha }^{*} } \wedge \Big( {\lambda }_{ \vec{\alpha }^{*} } - \frac{\zeta }{\theta } \Big) \bigg) \zeta \bigg] 
  - \theta r \\
& = \mathbb{E} [ \theta {\lambda }_{ \vec{\alpha }^{*} } - \zeta - \kappa {Z}_{*} ]
  - \frac{1}{2} \Var [ \kappa {Z}_{*} + \zeta ]
  + \mathbb{E} [ ( \theta {X}_{ \vec{\alpha }^{*} } + \zeta + \kappa {Z}_{*} - \theta {\lambda }_{ \vec{\alpha }^{*} } ) \zeta ]
  - \theta r \\
& = \kappa \theta {\lambda }_{ \vec{\alpha }^{*} }
  - \frac{1}{2} - \frac{1}{2} \mathbb{E} [ ( \kappa {Z}_{*} + \zeta )^{2} ]
  + \mathbb{E} [ ( \theta {X}_{ \vec{\alpha }^{*} } + \zeta + \kappa {Z}_{*} ) \zeta ] 
  - \theta r \\
& = \frac{1}{2} \mathbb{E} [ ( \kappa {Z}_{*} + \zeta )^{2} ] - \frac{1}{2} + \theta \langle \vec{\alpha }^{*}, \mathbb{E} [ ( \vec{R} - r \vec{1} ) ( \kappa {Z}_{*} + \zeta ) ] \rangle.
\end{align*}
According to the gradient condition \cref{eq: gradient condition for SP :eq}, i.e. $\mathbb{E} [ ( \vec{R} - r \vec{1} ) ( \kappa {Z}_{*} + \zeta ) ] = 0$, we further attain $2 \mathcal{M} = \mathbb{E} [ ( \kappa {Z}_{*} + \zeta )^{2} ] - 1$.
Consequently, the max-min inequality
\begin{align*}
2 \mathcal{M} & \le \min_{ Z \in \mathbb{L}^{2}_{+} ( \Omega ), \mathbb{E} [Z] = 1 } \max_{ \vec{\alpha } \in \mathbb{R}^{n} } 
                    \{ 2 \langle \theta \vec{\alpha }, \mathbb{E} [ ( \vec{R} - r \vec{1} ) ( \kappa Z + \zeta ) ] \rangle + \mathbb{E} [ ( \kappa Z + \zeta )^{2} ] - 1 \} \\
           & \equiv \min_{ Z \in \mathbb{L}^{2}_{+} ( \Omega ), ~ \mathbb{E} [Z] = 1, ~ \mathbb{E} [ ( \vec{R} - r \vec{1} ) ( \kappa Z + \zeta ) ] = \vec{0} } \{ \mathbb{E} [ ( \kappa Z + \zeta )^{2} ] - 1 \}
\end{align*}
for \cref{eq: max-min problem for SP :eq} holds with equality. 
Therefore, $( \vec{\alpha }^{*}, {Z}_{*} )$ is the saddle point for \cref{eq: max-min problem for SP :eq}.

\subsection{Proof of \texorpdfstring{\Cref{lem: express SMMV portfolio}}{Proposition 3.3}}
\label{pf-lem: express SMMV portfolio}

In order to remove the Lagrange multiplier $\mu \in \mathbb{R}$, we centralize the first equation in \cref{eq: KKT condition :eq} and attain
\begin{equation}\label{eq: centralized KKT equation :eq}
0 = \kappa {Z}_{*} + \zeta - 1 + \langle \theta \vec{\alpha }^{*}, \vec{R} - \mathbb{E} [R] \rangle - ( \beta - \mathbb{E} [ \beta ] ), \quad \mathbb{P}-a.s.
\end{equation}
Substituting \cref{eq: centralized KKT equation :eq} into the first line of the gradient condition \cref{eq: gradient condition for SP :eq} yields 
\begin{equation*}
  \mathbb{E} [ ( \vec{R} - r \vec{1} ) ]
= \theta \mathbb{E} \big[ ( \vec{R} - r \vec{1} ) ( \vec{R} - \mathbb{E} [ \vec{R} ] )^{\top } \big] \vec{\alpha }^{*}
- \mathbb{E} \big[ ( \vec{R} - r \vec{1} ) ( \beta - \mathbb{E} [ \beta ] ) \big]
= \theta \Var [ \vec{R} ] \vec{\alpha }^{*} - \Cov [ \vec{R}, \beta ],
\end{equation*}
which results in \cref{eq: express SMMV portfolio by MV portfolio :eq}.
Multiplying by $\beta $ and then taking expectation on both sides of \cref{eq: centralized KKT equation :eq}, in conjunction with applying the third line in \cref{eq: KKT condition :eq}, 
one obtains $0 = \mathbb{E} [ \beta ( \zeta - 1 ) ] + \langle \theta \vec{\alpha }^{*}, \Cov [ \vec{R}, \beta ] \rangle - \Var [ \beta ]$, i.e. \cref{eq: express SMMV portfilio by Lagrange multiplier :eq}.

Furthermore, as the first equation of \cref{eq: KKT condition :eq} gives $\beta = \kappa {Z}_{*} + ( \theta {X}_{ \vec{\alpha }^{*} } + \zeta ) - r \theta ( 1 - \langle \vec{\alpha }^{*}, \vec{1} \rangle ) - \mu $,
we have $\Var [ \beta ] = \Var [ \kappa {Z}_{*} + ( \theta {X}_{ \vec{\alpha }^{*} } + \zeta ) ] = {\theta }^{2} \Var [ {\lambda }_{ \vec{\alpha }^{*} } \vee ( {X}_{ \vec{\alpha }^{*} } + \frac{\zeta }{\theta } ) ]$.
If ${X}_{ \vec{\alpha }^{*} } \in \mathcal{G}_{ \theta, \zeta }$, then $\Var [ \beta ] = {\theta }^{2} \Var [ {\lambda }_{ \vec{\alpha }^{*} } ] = 0$, according to the third assertion in \Cref{lem: identity and truncation}.
Conversely, we suppose that $\Var [ \beta ] = 0$, and assume by contradiction that $\mathbb{P} ( {X}_{ \vec{\alpha }^{*} } + \frac{\zeta }{\theta } > {\lambda }_{ \vec{\alpha }^{*} } ) \in ( 0,1 )$
(noting that ${\lambda }_{ \vec{\alpha }^{*} } > \essinf ( {X}_{ \vec{\alpha }^{*} } + \zeta / \theta )$, see the second assertion in \Cref{lem: identity and truncation}).
By iterated conditioning formula, one obtains the decomposition
\begin{align*}
  \frac{1}{ {\theta }^{2} } \Var [ \beta ]
= \Var \Big[ {\lambda }_{ \vec{\alpha }^{*} } \vee \Big( {X}_{ \vec{\alpha }^{*} } + \frac{\zeta }{\theta } \Big) \Big]
& = \mathbb{E} \bigg[ \Var \Big[ {\lambda }_{ \vec{\alpha }^{*} } \vee \Big( {X}_{ \vec{\alpha }^{*} } + \frac{\zeta }{\theta } \Big) 
                                 \Big| {1}_{\{ {X}_{ \vec{\alpha }^{*} } + \frac{\zeta }{\theta } > {\lambda }_{ \vec{\alpha }^{*} } \}} \Big] \bigg] \\
& \quad + \Var \bigg[ \mathbb{E} \Big[ {\lambda }_{ \vec{\alpha }^{*} } \vee \Big( {X}_{ \vec{\alpha }^{*} } + \frac{\zeta }{\theta } \Big) 
                                       \Big| {1}_{\{ {X}_{ \vec{\alpha }^{*} } + \frac{\zeta }{\theta } > {\lambda }_{ \vec{\alpha }^{*} } \}} \Big] \bigg],
\end{align*}
and hence each term in this decomposition vanishes.
Since the first term vanishes, there exists some constant $c > {\lambda }_{ \vec{\alpha }^{*} }$ such that 
${X}_{ \vec{\alpha }^{*} } + \frac{\zeta }{\theta } = c$, $\mathbb{P}$-a.e. on $\{ {X}_{ \vec{\alpha }^{*} } + \frac{\zeta }{\theta } > {\lambda }_{ \vec{\alpha }^{*} } \}$.
Then, for the second term, 
\begin{equation*}
  \mathbb{E} \Big[ {\lambda }_{ \vec{\alpha }^{*} } \vee \Big( {X}_{ \vec{\alpha }^{*} } + \frac{\zeta }{\theta } \Big) \Big| {1}_{\{ {X}_{ \vec{\alpha }^{*} } + \frac{\zeta }{\theta } > {\lambda }_{ \vec{\alpha }^{*} } \}} \Big]
= {\lambda }_{ \vec{\alpha }^{*} } {1}_{\{ {X}_{ \vec{\alpha }^{*} } + \frac{\zeta }{\theta } \le {\lambda }_{ \vec{\alpha }^{*} } \}}
+ c {1}_{\{ {X}_{ \vec{\alpha }^{*} } + \frac{\zeta }{\theta } > {\lambda }_{ \vec{\alpha }^{*} } \}}
\end{equation*}
has a positive variance, which leads to a contradiction.
Therefore, ${X}_{ \vec{\alpha }^{*} } + \frac{\zeta }{\theta } \le {\lambda }_{ \vec{\alpha }^{*} }$, namely, ${X}_{ \vec{\alpha }^{*} } \in \mathcal{G}_{ \theta, \zeta }$, follows from $\Var [ \beta ] = 0$.
So we are done.

\subsection{Proof of \texorpdfstring{\Cref{thm: Lagrange multiplier as solution}}{Theorem 3.5}}
\label{pf-thm: Lagrange multiplier as solution}

By the variable replacement $Y = \kappa Z + \zeta $, we conclude that ${Z}_{*} = \frac{1}{\kappa } ( {Y}_{*} - \zeta )$ is the unique solution for
\begin{equation*}
\text{minimizing} \quad \mathbb{E} [ ( \kappa Z + \zeta )^{2} ] \quad
\text{subject to} \quad Z \in \mathbb{L}^{2}_{+} ( \Omega ), ~ \mathbb{E} [Z] = 1, ~ \mathbb{E} [ ( \vec{R} - r \vec{1} ) ( \kappa Z + \zeta ) ] = \vec{0}.
\end{equation*}
To characterize the minimizer ${Z}_{*}$ by Lagrange multiplier method with the Lagrangian
\begin{equation*}
\frac{1}{\kappa } \mathcal{L}_{2} ( \kappa Z + \zeta, \beta, \mu, \vec{\alpha } ) 
= \frac{1}{ 2 \kappa } \mathbb{E} [ ( \kappa Z + \zeta )^{2} ] - \mathbb{E} [ \beta Z ] + \theta ( r - \mu ) ( \mathbb{E} [Z] - 1 ) + \frac{\theta }{\kappa } \langle \vec{\alpha }, \mathbb{E} [ ( \vec{R} - r \vec{1} ) ( \kappa Z + \zeta ) ] \rangle,
\end{equation*}
one obtains the following KKT condition that is indeed equivalent to \cref{eq: modified KKT condition in SP :eq}:
\begin{equation*}
\left\{ \begin{aligned}
& 0 = \kappa {Z}_{*} + \zeta - \beta - \theta \mu + \theta {X}_{ \vec{\alpha } }, \quad \mathbb{P}-a.s.; \\
& \mathbb{E} [ {Z}_{*} ] = 1, \quad \mathbb{E} [ ( \vec{R} - r \vec{1} ) ( \kappa {Z}_{*} + \zeta ) ] = \vec{0}; \\
& \beta \ge 0, ~ {Z}_{*} \ge 0, ~ \beta {Z}_{*} = 0, \quad \mathbb{P}-a.s.
\end{aligned} \right.
\end{equation*}
Multiplying by $\kappa {Z}_{*}$ on both sides of the first equation in the aforementioned KKT condition, together with the application of the third line, 
yields $0 = \kappa {Z}_{*} ( \kappa {Z}_{*} - \theta \mu + \theta {X}_{ \vec{\alpha } } + \zeta )$.
For $\mathbb{P}$-a.e. $\omega \in \{ \theta \mu \le \theta {X}_{ \vec{\alpha } } + \zeta \}$, $\kappa {Z}_{*} ( \omega ) = 0$ is straightforward.
For $\mathbb{P}$-a.e. $\omega \in \{ \theta \mu > \theta {X}_{ \vec{\alpha } } + \zeta \}$, 
if $\kappa {Z}_{*} ( \omega ) = 0$, then the first equation in the aforementioned KKT condition gives $\beta ( \omega ) = \zeta ( \omega ) + \theta {X}_{ \vec{\alpha } } ( \omega ) - \theta \mu < 0$.
This contradicts $\beta \ge 0$, unless $\{ \theta \mu > \theta {X}_{ \vec{\alpha } } + \zeta \}$ is a $\mathbb{P}$-null set.
Summing up, we obtain $\kappa {Z}_{*} = ( \theta \mu - \theta {X}_{ \vec{\alpha } } - \zeta )_{+}$, and hence $\mu = {\lambda }_{ \vec{\alpha } }$ due to $\mathbb{E} [ {Z}_{*} ] = 1$ and the second assertion in \Cref{lem: identity and truncation}. 
Moreover, $\mathbb{E} [ ( \vec{R} - r \vec{1} ) ( \kappa {Z}_{*} + \zeta ) ] = \vec{0}$ implies that $\vec{\alpha }$ fulfills the gradient condition \cref{eq: gradient condition for SP :eq}.
Therefore, $\vec{\alpha }$ arising from \cref{eq: modified KKT condition in SP :eq} is the maximizer for \cref{eq: SP :eq}.

Finally, we suppose that $\Var [ \vec{R} | {Y}_{*} > \zeta ]$ is invertible. 
Notably, ${Y}_{*} ( \omega ) > \zeta ( \omega )$, namely ${Z}_{*} ( \omega ) > 0$, is equivalent to ${X}_{ \vec{\alpha } } ( \omega ) + \frac{1}{\theta } \zeta ( \omega ) < {\lambda }_{ \vec{\alpha } }$.
Therefore, by mirroring the rearrangement in \Cref{pf-thm: optimality condition for SP}, 
replacing each instance of ${X}_{ \vec{\alpha } } + \frac{\zeta }{\theta } \le {\lambda }_{ \vec{\alpha } }$ with ${X}_{ \vec{\alpha } } + \frac{\zeta }{\theta } < {\lambda }_{ \vec{\alpha } }$,
one obtains the desired expression for $\vec{\alpha }$.

\subsection{Proof of \texorpdfstring{\Cref{lem: MV optimal portfolio in monotonicity domain}}{Proposition 4.2}}
\label{pf-lem: MV optimal portfolio in monotonicity domain}

Applying It\^o's rule to ${X}^{\pi }_{t} {D}_{t}$, one obtains ${X}^{\pi }_{T} {D}_{T} = {x}_{0} + \int_{0}^{T} {D}_{t} ( {\pi }_{t} {\sigma }_{t} - {X}^{\pi }_{t} {\vartheta }_{t} ) d {W}_{t}$,
which coincides with the $( \mathbb{F}, \mathbb{P} )$-martingale representation of $\{ \mathbb{E}_{t} [ {X}^{\pi }_{T} {D}_{T} ] \}_{ t \in [ 0,T ] }$.
This implies that
\begin{equation*}
\mathbb{L}^{2}_{\mathbb{F}} ( 0,T; \mathbb{L}^{2} ( \Omega ) ) = \{ \pi: {X}^{\pi }_{T} \in \mathbb{L}^{2} ( \Omega ) \}, \quad and \quad
\mathbb{E} \Big[ \Big( {X}^{\pi }_{T} - \frac{ {x}_{0} }{ \mathbb{E} [ {D}_{T} ] } \Big) {D}_{T} \Big] = 0.
\end{equation*}
In other words, maximizing ${U}_{\theta } ( {X}^{\pi }_{T} )$ subject to $\pi \in \mathbb{L}^{2}_{\mathbb{F}} ( 0,T; \mathbb{L}^{2} ( \Omega ) )$ can be reduced to
maximizing ${U}_{\theta } ( {X}^{\pi }_{T} )$ subject to ${X}^{\pi }_{T} \in \mathbb{L}^{2} ( \Omega )$ with $\mathbb{E} [ ( {X}^{\pi }_{T} - R ) {D}_{T} ] = 0$, 
where $R = {x}_{0} / \mathbb{E} [ {D}_{T} ]$.
Applying \Cref{lem: general mean-variance optimization} to the aforementioned $R$ with $Y = {D}_{T}$, $\mathcal{G} = \mathcal{F}_{0}$ and $\mathbb{D} = \{ c Y \}_{c \ge 0}$, 
one can conclude that ${X}^{**}$ for maximizing ${U}_{\theta } (X)$ subject to $\mathbb{E} [ ( X - R ) Y ] \le 0$ satisfies
\begin{equation*}
{X}^{**} = \mathbb{E} [ {X}^{**} ] + \frac{1}{\theta } \Big( 1 - \frac{Y}{ \mathbb{E} [Y] } \Big), \quad \mathbb{E} [ ( {X}^{**} - R ) Y ] = 0.
\end{equation*}
Therefore, ${X}^{*}_{T} = \mathbb{E} [ {X}^{*}_{T} ] + \frac{1}{\theta } ( 1 - \frac{ {D}_{T} }{ \mathbb{E} [ {D}_{T} ] } )$.
The equivalence between ${X}^{*}_{T} \in \mathcal{G}_{ \theta, \zeta }$ and \Cref{ass: SMMV-MV comparison} follows, according to the expression of $\mathcal{G}_{ \theta, \zeta }$ in the first assertion in \Cref{lem: identity and truncation}.

\subsection{Proof of \texorpdfstring{\Cref{thm: identical solution of SMMV and MV}}{Theorem 4.3}}
\label{pf-thm: identical solution of SMMV and MV}

Consider an arbitrarily fixed $\pi \in \mathbb{L}^{2}_{\mathbb{F}} ( 0,T; \mathbb{L}^{2} ( \Omega ) )$. 
According to \Cref{lem: MV optimal portfolio in monotonicity domain}, if ${X}^{\pi }_{T} \in \mathcal{G}_{ \theta, \zeta }$, then
\begin{equation*}
    {V}_{ \theta, \zeta } ( {X}^{\pi }_{T} ) 
  = {U}_{\theta } ( {X}^{\pi }_{T} ) 
\le \max_{ \pi \in \mathbb{L}^{2}_{\mathbb{F}} ( 0,T; \mathbb{L}^{2} ( \Omega ) ) } {U}_{\theta } ( {X}^{\pi }_{T} ) 
  = {U}_{\theta } ( {X}^{*}_{T} ) 
  = {V}_{ \theta, \zeta } ( {X}^{*}_{T} ).
\end{equation*}
Otherwise, we write $\mathbb{X}_{T} := {X}^{\pi }_{T} \wedge ( {\lambda }_{ {X}^{\pi }_{T}, \theta, \zeta } - \frac{\zeta }{\theta } )$ for short. 
$\mathbb{X}_{T} \in \mathcal{G}_{ \theta, \zeta }$ follows from the second assertion in \Cref{lem: identity and truncation}.
Owing to \cref{eq: modificed MV form of SMMV :eq}, one obtains the decomposition ${V}_{ \theta, \zeta } ( {X}^{\pi }_{T} ) = {U}_{\theta } ( \mathbb{X}_{T} ) + \mathbb{E} [ ( {X}^{\pi }_{T} - \mathbb{X}_{T} ) \zeta ]$.
On the one hand, the concavity of ${U}_{\theta }$ results in ${U}_{\theta } ( \mathbb{X}_{T} ) \le {U}_{\theta } ( {X}^{*}_{T} ) + \mathbb{E} [ ( \mathbb{X}_{T} - {X}^{*}_{T} ) \partial {U}_{\theta } ( {X}^{*}_{T} ) ]$,
where $\partial {U}_{\theta } ( {X}^{*}_{T} ) = 1 - \theta ( {X}^{*}_{T} - \mathbb{E} [ {X}^{*}_{T} ] ) = {D}_{T} / \mathbb{E} [ {D}_{T} ]$ can be found at the end of the previous \Cref{pf-lem: MV optimal portfolio in monotonicity domain}.
On the other hand, $\mathbb{E} [ ( {X}^{\pi }_{T} - \mathbb{X}_{T} ) \zeta ] \le \mathbb{E} [ ( {X}^{\pi }_{T} - \mathbb{X}_{T} ) {D}_{T} ] / \mathbb{E} [ {D}_{T} ]$ if \Cref{ass: SMMV-MV comparison} holds. 
Summing up, one obtains
\begin{equation*}
{V}_{ \theta, \zeta } ( {X}^{\pi }_{T} ) 
\le {U}_{\theta } ( {X}^{*}_{T} ) 
  + \frac{ \mathbb{E} [ ( {X}^{\pi }_{T} - {X}^{*}_{T} ) {D}_{T} ] }{ \mathbb{E} [ {D}_{T} ] }
  = {U}_{\theta } ( {X}^{*}_{T} ).
\end{equation*}
In view of the arbitrariness of $\pi $, the optimality of ${X}^{*}_{T}$ is straightforward.

\subsection{Proof of \texorpdfstring{\Cref{thm: verification theorem for relaxed SDGs}}{Theorem 4.4}}
\label{pf-thm: verification theorem for relaxed SDGs}

Given an arbitrarily fixed $( \pi, \gamma ) \in \mathbb{L}^{2}_{\mathbb{F}} ( 0,T; \mathbb{L}^{2} ( \Omega ) ) \times {\Gamma }^{t,z}$ and \cref{eq: HJBI-BSPDE for SDGs :eq},
applying the It\^o-Kunita-Ventzel formula (as a generalized version of It\^o's rule, see \cite[Theorem 1.5.3.2]{Jeanblanc-Yor-Chesney-2009}) to $\mathcal{V} ( s, {X}^{ t,x, \pi }_{s}, {Z}^{ t,z, \gamma }_{s} )$ yields
\begin{align*}
& d \mathcal{V} ( s, {X}^{ t,x, \pi }_{s}, {Z}^{ t,z, \gamma }_{s} ) \\
& = \big( \mathcal{D}_{1}^{ \pi, \gamma } \mathcal{V} ( s, {X}^{ t,x, \pi }_{s}, {Z}^{ t,z, \gamma }_{s} ) + \mathcal{D}_{2}^{ \pi, \gamma } \Phi ( s, {X}^{ t,x, \pi }_{s}, {Z}^{ t,z, \gamma }_{s} )
     - \mathbb{H} ( s, {X}^{ t,x, \pi }_{s}, {Z}^{ t,z, \gamma }_{s} ) \big) dt \\
& \quad + \big( \Phi ( s, {X}^{ t,x, \pi }_{s}, {Z}^{ t,z, \gamma }_{s} ) + \mathcal{D}_{2}^{ {\pi }_{s}, {\gamma }_{s} } \mathcal{V} ( s, {X}^{ t,x, \pi }_{s}, {Z}^{ t,z, \gamma }_{s} ) \big) d {W}_{s}.
\end{align*}
Integrating both sides of the above SDE from $t$ to $T$, and taking expectation conditioned on $\mathcal{F}_{t}$ under $\mathbb{P}$,
together with the second line of \cref{eq: saddle point condition :eq} and the terminal condition \cref{eq: terminal condition for SDGs :eq}, one obtains
\begin{equation*}
\mathcal{V} ( t,x,z ) 
\le \mathbb{E}_{t} [ \mathcal{V} ( T, {X}^{ t,x, {\pi }^{*} }_{T}, {Z}^{ t,z, \gamma }_{T} ) ] 
  = \mathbb{E}_{t} [ {J}^{ {\pi }^{*}, \gamma } ( T, {X}^{ t,x, {\pi }^{*} }_{T}, {Z}^{ t,z, \gamma }_{T} ) ] 
  = {J}^{ {\pi }^{*}, \gamma } ( t,x,z ), \quad \forall \gamma \in {\Gamma }^{t,z}.
\end{equation*}
In the same manner, one can show that $\mathcal{V} ( t,x,z ) \ge {J}^{ \pi, {\gamma }^{*} } ( t,x,z )$ for any $\pi \in \mathbb{L}^{2}_{\mathbb{F}} ( 0,T; \mathbb{L}^{2} ( \Omega ) )$.
Therefore, 
\begin{equation*}
{J}^{ \pi, {\gamma }^{*} } ( t,x,z ) \le \mathcal{V} ( t,x,z ) = {J}^{ {\pi }^{*}, {\gamma }^{*} } ( t,x,z ) \le {J}^{ {\pi }^{*}, \gamma } ( t,x,z ), 
\quad \forall ( \pi, \gamma ) \in \mathbb{L}^{2}_{\mathbb{F}} ( 0,T; \mathbb{L}^{2} ( \Omega ) ) \times {\Gamma }^{t,z}.
\end{equation*}
This implies that $( {\pi }^{*}, {\gamma }^{*} )$ is the desired saddle point, which leads to the max-min equality
\begin{equation*}
  {J}^{ {\pi }^{*}, {\gamma }^{*} } ( t,x,z )
= \esssup_{ \pi \in \mathbb{L}^{2}_{\mathbb{F}} ( 0,T; \mathbb{L}^{2} ( \Omega ) ) } \essinf_{ \gamma \in {\Gamma }^{t,z} } {J}^{ \pi, \gamma } ( t,x,z )
= \essinf_{ \gamma \in {\Gamma }^{t,z} } \esssup_{ \pi \in \mathbb{L}^{2}_{\mathbb{F}} ( 0,T; \mathbb{L}^{2} ( \Omega ) ) } {J}^{ \pi, \gamma } ( t,x,z ).
\end{equation*}
So, the proof is completed.

\subsection{Proof of \texorpdfstring{\Cref{thm: optimal MV solution}}{Theorem 4.6}}
\label{pf-thm: optimal MV solution}

Owing to \Cref{thm: verification theorem for relaxed SDGs} and \Cref{rem: upper Isaacs BSPDE}, 
one can proceed with $\hat{\pi }, \hat{\gamma }: [ 0,T ] \times \Omega \times \mathbb{R} \times [ 0, + \infty ) \to \mathbb{R}$ fulfilling the optimality condition \cref{eq: optimality condition of HJBI-BSPDE for SDGs :eq}.
Notably, \cref{eq: optimality condition of HJBI-BSPDE for SDGs :eq} can be re-expressed as
\begin{equation*}
\left\{ \begin{aligned}
& 0 = \mathcal{D}_{2}^{ \hat{\pi } ( t,x,z ), \hat{\gamma } ( t,x,z ) } \mathcal{V}_{x} ( t,x,z ) + \mathcal{V}_{x} ( t,x,z ) {\vartheta }_{t} + {\Phi }_{x} ( t,x,z ), \\
& 0 = \mathcal{D}_{2}^{ \hat{\pi } ( t,x,z ), \hat{\gamma } ( t,x,z ) } \mathcal{V}_{z} ( t,x,z ) + {\Phi }_{z} ( t,x,z ),
\end{aligned} \right.
\end{equation*}
which will be used to rearrange and simplify the $d {W}_{t}$ terms in applying It\^o-Kunita-Ventzel formula.
In addition, readers can treat the quadratic form of \cref{eq: MV value random field :eq} as an Ansatz, so that the smoothness requirements in the following derivation are satisfied.
By envelope theorem or straightforward calculation, 
one can differentiate the both sides of \cref{eq: HJBI-BSPDE for SDGs :eq} with respect to $x$ to attain the semi-martingale decomposition:
\begin{align*}
  - d \mathcal{V}_{x} ( t,x,z )
& = \big( \mathcal{V}_{x} ( t,x,z ) {r}_{t} + \mathcal{D}_{1}^{ \hat{\pi } ( t,x,z ), \hat{\gamma } ( t,x,z ) } \mathcal{V}_{x} ( t,x,z ) + \mathcal{D}_{2}^{ \hat{\pi } ( t,x,z ), \hat{\gamma } ( t,x,z ) } {\Phi }_{x} ( t,x,z ) \big) dt \\
& \quad - {\Phi }_{x} ( t,x,z ) d {W}_{t}.
\end{align*}
Then, applying It\^o-Kunita-Ventzel formula to $\mathcal{V}_{x} ( s, {X}^{ t,x, {\pi }^{*} }_{s}, {Z}^{ t,z, {\gamma }^{*} }_{s} )$ yields
\begin{equation*}
  d \mathcal{V}_{x} ( s, {X}^{ t,x, {\pi }^{*} }_{s}, {Z}^{ t,z, {\gamma }^{*} }_{s} )
= - \mathcal{V}_{x} ( s, {X}^{ t,x, {\pi }^{*} }_{s}, {Z}^{ t,z, {\gamma }^{*} }_{s} ) ( {r}_{s} ds + {\vartheta }_{s} d {W}_{s} ).
\end{equation*}
Given \cref{eq: terminal condition for SDGs :eq}, one obtains 
$\mathcal{V}_{x} ( t,x,z ) \frac{ {D}_{T} }{ {D}_{t} } = \mathcal{V}_{x} ( T, {X}^{ t,x, {\pi }^{*} }_{T}, {Z}^{ t,z, {\gamma }^{*} }_{T} ) = \zeta + \kappa {Z}^{ t,z, {\gamma }^{*} }_{T}$, and hence
\begin{equation*}
\mathcal{V}_{x} ( t,x,z ) = ( \mathbb{E}_{t} [ \zeta ] + \kappa z ) \frac{ {D}_{t} }{ \mathbb{E}_{t} [ {D}_{T} ] },
\end{equation*}
which leads to $\mathcal{V}_{xx} ( t,x,z ) = 0$, $\mathcal{V}_{xz} ( t,x,z ) = \frac{ \kappa {D}_{t} }{ \mathbb{E}_{t} [ {D}_{T} ] }$,
\begin{equation*}
{\Phi }_{x} ( t,x,z ) = \frac{ {D}_{t} }{ \mathbb{E}_{t} [ {D}_{T} ] } {\eta }_{t} 
                      + ( \mathbb{E}_{t} [ \zeta ] + \kappa z ) \frac{ {D}_{t} }{ \mathbb{E}_{t} [ {D}_{T} ] } {\xi }^{(1)}_{t} 
\end{equation*}
and
\begin{equation*}
\zeta + \kappa {Z}^{ t,z, {\gamma }^{*} }_{T} = ( \mathbb{E}_{t} [ \zeta ] + \kappa z ) \frac{ {D}_{T} }{ \mathbb{E}_{t} [ {D}_{T} ] }.
\end{equation*}
In the same manner, one obtains
\begin{equation*}
\left\{ \begin{aligned}
  d \mathcal{V}_{z} ( s, {X}^{ t,x, {\pi }^{*} }_{s}, {Z}^{ t,z, {\gamma }^{*} }_{s} )
& = \big( {\Phi }_{z} ( s, {X}^{ t,x, {\pi }^{*} }_{s}, {Z}^{ t,z, {\gamma }^{*} }_{s} )
        + \mathcal{D}_{2}^{ {\pi }^{*}_{s}, {\gamma }^{*}_{s} } {V}_{z} ( s, {X}^{ t,x, {\pi }^{*} }_{s}, {Z}^{ t,z, {\gamma }^{*} }_{s} ) \big) d {W}_{s} = 0, \\ 
    \mathcal{V}_{z} ( T, {X}^{ t,x, {\pi }^{*} }_{T}, {Z}^{ t,z, {\gamma }^{*} }_{T} )
& = \frac{\kappa }{\theta } \mathcal{V}_{x} ( T, {X}^{ t,x, {\pi }^{*} }_{T}, {Z}^{ t,z, {\gamma }^{*} }_{T} ) + \kappa {X}^{ t,x, {\pi }^{*} }_{T}
  = \mathcal{V}_{x} ( t,x,z ) \frac{ \kappa {D}_{T} }{ \theta {D}_{t} } + \kappa {X}^{ t,x, {\pi }^{*} }_{T},
\end{aligned} \right.
\end{equation*}
which results in
\begin{align*}
\mathcal{V}_{z} ( t,x,z ) & = \mathcal{V}_{x} ( t,x,z ) \frac{ \kappa \mathbb{E}_{t} [ | {D}_{T} |^{2} ] }{ \theta {D}_{t} \mathbb{E}_{t} [ {D}_{T} ] } + \kappa x \frac{ {D}_{t} }{ \mathbb{E}_{t} [ {D}_{T} ] }
                            = ( \mathbb{E}_{t} [ \zeta ] + \kappa z ) \frac{ \kappa \mathbb{E}_{t} [ | {D}_{T} |^{2} ] }{ \theta | \mathbb{E}_{t} [ {D}_{T} ] |^{2} } + \kappa x \frac{ {D}_{t} }{ \mathbb{E}_{t} [ {D}_{T} ] }, \\
\mathcal{V}_{zz} ( t,x,z ) & = \frac{ {\kappa }^{2} \mathbb{E}_{t} [ | {D}_{T} |^{2} ] }{ \theta | \mathbb{E}_{t} [ {D}_{T} ] |^{2} }, \\
{\Phi }_{z} ( t,x,z ) & = \frac{ \kappa \mathbb{E}_{t} [ | {D}_{T} |^{2} ] }{ \theta | \mathbb{E}_{t} [ {D}_{T} ] |^{2} } {\eta }_{t} 
                        + ( \mathbb{E}_{t} [ \zeta ] + \kappa z ) \frac{ \kappa \mathbb{E}_{t} [ | {D}_{T} |^{2} ] }{ \theta | \mathbb{E}_{t} [ {D}_{T} ] |^{2} } ( {\xi }^{(2)}_{t} + 2 {\xi }^{(1)}_{t} )
                        + \kappa x \frac{ {D}_{t} }{ \mathbb{E}_{t} [ {D}_{T} ] } {\xi }^{(1)}_{t},
\end{align*}
and
\begin{equation*}
  {X}^{ t,x, {\pi }^{*} }_{T} 
= \frac{1}{\kappa } \mathcal{V}_{z} ( t,x,z ) - \mathcal{V}_{x} ( t,x,z ) \frac{ {D}_{T} }{ \theta {D}_{t} }
= x \frac{ {D}_{t} }{ \mathbb{E}_{t} [ {D}_{T} ] }
+ ( \mathbb{E}_{t} [ \zeta ] + \kappa z ) \frac{ \mathbb{E}_{t} [ | {D}_{T} |^{2} ] - {D}_{T} \mathbb{E}_{t} [ {D}_{T} ] }{ \theta | \mathbb{E}_{t} [ {D}_{T} ] |^{2} }.
\end{equation*}
On the one hand, plugging the above partial derivatives back into \cref{eq: optimality condition of HJBI-BSPDE for SDGs :eq} immediately yields \cref{eq: optimal MV portfolio :eq}.
On the other hand, as the pair $( {X}^{ t,x, {\pi }^{*} }_{T}, {Z}^{ t,z, {\gamma }^{*} }_{T} )$ has been derived, \cref{eq: optimal MV state pair :eq} arises from
\begin{equation*}
{X}^{ t,x, {\pi }^{*} }_{s} = \frac{1}{ {D}_{s} } \mathbb{E}_{s} [ {X}^{ t,x, {\pi }^{*} }_{T} {D}_{T} ], \quad {Z}^{ t,z, {\gamma }^{*} }_{s} = \mathbb{E}_{s} [ {Z}^{ t,z, {\gamma }^{*} }_{T} ].
\end{equation*}
Moreover, the value random field \cref{eq: MV value random field :eq} can be derived by straightforward calculation with 
\begin{equation*}
\mathcal{V} ( t,x,z ) = {J}^{ {\pi }^{*}, {\gamma }^{*} } ( t,x,z )
= \mathbb{E}_{t} \Big[ {X}^{ t,x, {\pi }^{*} }_{T} ( \zeta + \kappa {Z}^{ t,z, {\gamma }^{*} }_{T} )
                     + \frac{1}{ 2 \theta } | \zeta + \kappa {Z}^{ t,z, {\gamma }^{*} }_{T} |^{2} - \frac{1}{ 2 \theta } | \zeta |^{2} \Big].
\end{equation*}

Furthermore, ${Z}^{ t,z, {\gamma }^{*} }_{s} \ge \frac{1}{\kappa } \mathbb{E}_{s} [ \frac{ {D}_{T} }{ \mathbb{E} [ {D}_{T} ] } - \zeta ]$ 
follows from \cref{eq: optimal MV state pair :eq} with $z \ge \frac{1}{\kappa } \mathbb{E}_{t} [ \frac{ {D}_{T} }{ \mathbb{E} [ {D}_{T} ] } - \zeta ]$ 
(including the case with $( t,z ) = ( 0,1 )$, as $1 = \frac{1}{\kappa } \mathbb{E} [ \frac{ {D}_{T} }{ \mathbb{E} [ {D}_{T} ] } - \zeta ]$).
If \Cref{ass: SMMV-MV comparison} holds, then ${\gamma }^{*} \in {\Gamma }^{t,z}$, and hence $( {\pi }^{*}, {\gamma }^{*} )$ is also a saddle point 
for \cref{eq: control problem primal :eq}, \cref{eq: saddle point definition :eq} and \cref{eq: control problem for DP :eq} with $z \ge \frac{1}{\kappa } \mathbb{E}_{t} [ \frac{ {D}_{T} }{ \mathbb{E} [ {D}_{T} ] } - \zeta ]$.

\subsection{Proof of \texorpdfstring{\Cref{thm: duality characterization}}{Theorem 4.7}}
\label{pf-thm: duality characterization}

The existence and uniqueness of $( \bar{y}, \bar{h} ) \in \mathbb{R}^{2}$ fulfilling \cref{eq: y-h system :eq} can be seen from the following fact: 
\begin{itemize}
\item the implicit function $y(h)$ determined by $\mathbb{E} [ {D}_{T} \mathcal{X} ( y {D}_{T}, h ) ] = {x}_{0}$ is strictly increasing with $y ( + \infty ) = + \infty $ and $| y ( - \infty ) | < + \infty $,
\item and the implicit function $y(h)$ determined by $\mathbb{E} [ \mathcal{Z} ( y {D}_{T}, h ) ] = 1$ is strictly decreasing with $y ( + \infty ) = - \infty $ and $y ( - \infty ) = + \infty $.
\end{itemize}
Given that $( \bar{y}, \bar{h} ) \in \mathbb{R}^{2}$ satisfies \cref{eq: y-h system :eq}, 
one can find the hedging strategies $\bar{\pi }, \bar{\gamma } \in \mathbb{L}^{2}_{\mathbb{F}} ( 0,T; \mathbb{L}^{2} ( \Omega ) )$ such that 
${X}^{ 0, {x}_{0}, \bar{\pi } }_{T} = \mathcal{X} ( \bar{y} {D}_{T}, \bar{h} )$ and ${Z}^{ 0,1, \bar{\gamma } }_{T} = \mathcal{Z} ( \bar{y} {D}_{T}, \bar{h} )$, i.e. \cref{eq: minimax equality conditions :eq} holds,  
according to the martingale representation theorem or the classical martingale method.
Moreover, $\bar{\gamma } \in {\Gamma }^{0,1}$ follows from $\mathcal{Z} \ge 0$. 
Notably, since \cref{eq: y-h system :eq} can also be derived from \cref{eq: minimax equality conditions :eq},
one can regard \cref{eq: y-h system :eq} as an equivalent characterization of \cref{eq: minimax equality conditions :eq}.
Furthermore, since $\tilde{U}$ is convex in $y$, concave in $h$ and continuously differentiable with $\tilde{U}_{y} = - \mathcal{X}$ and $\tilde{U}_{h} = - \mathcal{Z}$,
\cref{eq: y-h system :eq} is also the first-order derivative condition for \cref{eq: minimax equality in duality :eq}.

On the other hand, \cref{eq: minimax inequalities :eq} is straightforward due to \cref{eq: saddle point in duality :eq} with $\mathbb{E} [ {D}_{T} {X}^{ 0, {x}_{0}, \pi }_{T} ] = {x}_{0}$ and $\mathbb{E} [ {Z}^{ 0, 1, \gamma }_{T} ] = 1$.
When \cref{eq: minimax equality conditions :eq} holds, in conjunction with \cref{eq: y-h system :eq},
one obtains
\begin{align*} 
&   \mathbb{E} [ U ( {X}^{ 0, {x}_{0}, \bar{\pi } }_{T}, {Z}^{ 0,1, \bar{\gamma } }_{T} ) ] 
  = \mathbb{E} [ \tilde{U} ( \bar{y} {D}_{T}, \bar{h} ) ] + \bar{y} {x}_{0} + \bar{h}, \\
and \quad    
&   \mathbb{E} [ U ( {X}^{ 0, {x}_{0}, \pi }_{T}, {Z}^{ 0,1, \bar{\gamma } }_{T} ) ]
\le \mathbb{E} [ \tilde{U} ( \bar{y} {D}_{T}, \bar{h} \bar{I}_{T} ) ] + \bar{y} {x}_{0} + \bar{h} 
\le \mathbb{E} [ U ( {X}^{ 0, {x}_{0}, \bar{\pi } }_{T}, {Z}^{ 0,1, \gamma }_{T} ) ].
\end{align*}
In view of the arbitrariness of $( \pi, \gamma )$, the pair $( \bar{\pi }, \bar{\gamma } )$ is the saddle point for \cref{eq: approximate control problem :eq}.
In addition, by plugging $( y,h ) = ( \bar{y} {D}_{T}, \bar{h} )$ and $( x,z ) = ( \mathcal{X} ( y {D}_{T}, h ), \mathcal{Z} ( y {D}_{T}, h ) )$ 
(where the last $( y,h )$ in $( \mathcal{X}, \mathcal{Z} )$ is arbitrarily fixed) back into \cref{eq: saddle point in duality :eq} in sequence
and $( \pi, \gamma ) = ( \bar{\pi }, \bar{\gamma } )$ back into \cref{eq: minimax inequalities :eq}, together with \cref{eq: minimax equality conditions :eq} and \cref{eq: y-h system :eq}, it yields
\begin{align*}
      \mathbb{E} [ \tilde{U} ( \bar{y} {D}_{T}, \bar{h} ) ] + \bar{y} {x}_{0} + h
& \le \mathbb{E} \big[ U \big( \mathcal{X} ( \bar{y} {D}_{T}, \bar{h} ), \mathcal{Z} ( y {D}_{T}, h ) \big) \big] - \mathbb{E} [ h \mathcal{Z} ( y {D}_{T}, h ) ] + h \\
& \le \mathbb{E} [ \tilde{U} ( y {D}_{T}, h ) ] + y {x}_{0} + h \\
& \le \mathbb{E} \big[ U \big( \mathcal{X} ( y {D}_{T}, h ), \mathcal{Z} ( \bar{y} {D}_{T}, \bar{h} ) \big) \big] - \mathbb{E} [ y {D}_{T} \mathcal{X} ( y {D}_{T}, h ) ] + y {x}_{0} \\
& \le \mathbb{E} [ \tilde{U} ( \bar{y} {D}_{T}, \bar{h} ) ] + y {x}_{0} + \bar{h}.
\end{align*}
As $( y,h ) \in \mathbb{R}^{2}$ is arbitrarily fixed, one obtains
\begin{equation*}
    \mathbb{E} [ \tilde{U} ( \bar{y} {D}_{T}, h ) ] + \bar{y} {x}_{0} + h
\le \mathbb{E} [ \tilde{U} ( \bar{y} {D}_{T}, \bar{h} ) ] + \bar{y} {x}_{0} + \bar{h} 
\le \mathbb{E} [ \tilde{U} ( y {D}_{T}, \bar{h} ) ] + y {x}_{0} + \bar{h},
\end{equation*}
implying that $( \bar{y}, \bar{h} )$ is the saddle point for \cref{eq: minimax equality in duality :eq}.

\subsection{Proof of \texorpdfstring{\Cref{thm: embedding method}}{Theorem 4.8}}
\label{pf-thm: embedding method}

Assume by contradiction that ${\gamma }^{**} \notin \bar{\Gamma }^{0,1} ( {w}^{**} )$. 
Thus, there exists some $\gamma \in {\Gamma }^{0,1}$ such that
\begin{align*}
&   \frac{ \theta + \rho }{ 2 \theta } \big( \mathbb{E} [ ( \kappa {Z}^{ 0,1, \gamma }_{T} + \zeta )^{2} ] - \mathbb{E} [ ( \kappa {Z}^{ 0,1, {\gamma }^{**} }_{T} + \zeta )^{2} ] \big)
  + \rho \big( \mathbb{E} [ c ( \kappa {Z}^{ 0,1, \gamma }_{T} + \zeta ) ] - \mathbb{E} [ c ( \kappa {Z}^{ 0,1, {\gamma }^{**} }_{T} + \zeta ) ] \big) \\
& < {w}^{**} \big( \mathbb{E} [ {D}_{T} ( \kappa {Z}^{ 0,1, \gamma }_{T} + \zeta ) ] - \mathbb{E} [ {D}_{T} ( \kappa {Z}^{ 0,1, {\gamma }^{**} }_{T} + \zeta ) ] \big).
\end{align*}
As a consequence, for the jointly concave auxiliary function 
\begin{equation*}
F ( x,y,z ) = \frac{ \theta + \rho }{ 2 \theta } x + \rho y
            - \frac{1}{ 2 \mathbb{E} [ | {D}_{T} |^{2} ] } {z}^{2} + \frac{\rho }{ \mathbb{E} [ | {D}_{T} |^{2} ] } ( {x}_{0} - \mathbb{E} [ {D}_{T} c ] ) z,
\end{equation*}
one obtains
\begin{align*}
& \frac{ \theta + \rho }{ 2 \theta } \mathbb{E} [ ( \kappa {Z}^{ 0,1, \gamma }_{T} + \zeta )^{2} ] 
+ \rho \mathbb{E} [ c ( \kappa {Z}^{ 0,1, \gamma }_{T} + \zeta ) ] 
- {F}_{ \rho, 1 } ( \gamma ) \mathbb{E} [ {D}_{T} ( \kappa {Z}^{ 0,1, \gamma }_{T} + \zeta ) ] \\
&   = F \big( \mathbb{E} [ ( \kappa {Z}^{ 0,1, \gamma }_{T} + \zeta )^{2} ], \mathbb{E} [ c ( \kappa {Z}^{ 0,1, \gamma }_{T} + \zeta ) ], \mathbb{E} [ {D}_{T} ( \kappa {Z}^{ 0,1, \gamma }_{T} + \zeta ) ] \big) \\
& \le F \big( \mathbb{E} [ ( \kappa {Z}^{ 0,1, {\gamma }^{**} }_{T} + \zeta )^{2} ], \mathbb{E} [ c ( \kappa {Z}^{ 0,1, {\gamma }^{**} }_{T} + \zeta ) ], \mathbb{E} [ {D}_{T} ( \kappa {Z}^{ 0,1, {\gamma }^{**} }_{T} + \zeta ) ] \big) \\
& \quad + \frac{ \theta + \rho }{ 2 \theta } \big( \mathbb{E} [ ( \kappa {Z}^{ 0,1, \gamma }_{T} + \zeta )^{2} ] - \mathbb{E} [ ( \kappa {Z}^{ 0,1, {\gamma }^{**} }_{T} + \zeta )^{2} ] \big) \\
& \quad + \rho \big( \mathbb{E} [ c ( \kappa {Z}^{ 0,1, \gamma }_{T} + \zeta ) ] - \mathbb{E} [ c ( \kappa {Z}^{ 0,1, {\gamma }^{**} }_{T} + \zeta ) ] \big) \\
& \quad - {w}^{**} \big( \mathbb{E} [ {D}_{T} ( \kappa {Z}^{ 0,1, \gamma }_{T} + \zeta ) ] - \mathbb{E} [ {D}_{T} ( \kappa {Z}^{ 0,1, {\gamma }^{**} }_{T} + \zeta ) ] \big) \\
&   < F \big( \mathbb{E} [ ( \kappa {Z}^{ 0,1, {\gamma }^{**} }_{T} + \zeta )^{2} ], \mathbb{E} [ c ( \kappa {Z}^{ 0,1, {\gamma }^{**} }_{T} + \zeta ) ], \mathbb{E} [ {D}_{T} ( \kappa {Z}^{ 0,1, {\gamma }^{**} }_{T} + \zeta ) ] \big),
\end{align*}
which contradicts the minimality of ${\gamma }^{**}$. Hence, ${\gamma }^{**} \in \bar{\Gamma }^{0,1} ( {w}^{**} )$.

\subsection{Proof of \texorpdfstring{\Cref{lem: uniqueness of w}}{Proposition 4.9}}
\label{pf-lem: uniqueness of w}

By rearrangement, the equation $w = {F}_{ \rho, 2 } ( {\gamma }^{\dag } (w) )$ can be re-expressed as
\begin{equation*}
  \mathbb{E} [ {D}_{T} \zeta ] - \rho ( {x}_{0} - \mathbb{E} [ {D}_{T} c ] )
= \frac{\theta }{ \theta + \rho } \mathbb{E} \bigg[ {D}_{T} \bigg( \frac{\rho }{\theta } w {D}_{T} + w {D}_{T} \wedge \Big( \frac{ \theta + \rho }{\theta } \zeta + \rho c - \frac{h}{\kappa } \Big) \bigg) \bigg],
\end{equation*}
with $h = {h}^{\dag } (w)$.
Indeed, the above statement determines an implicit function $w = {w}^{\dag } (h)$.
Consequently, it suffices to show the existence and uniqueness of the solution to the fully coupled equations $h = {h}^{\dag } (w)$ and $w = {w}^{\dag } (h)$.
Since ${h}^{\dag }$ is continuous and strictly decreasing, while ${w}^{\dag }$ is continuous and non-decreasing, the solution must be unique. 
Since 
\begin{align*}
0 & \le \limsup_{ h \to - \infty } \mathbb{E} \bigg[ {D}_{T} \bigg( w {D}_{T} - w {D}_{T} \wedge \Big( \frac{ \theta + \rho }{\theta } \zeta + \rho c - \frac{h}{\kappa } \Big) \bigg) \bigg] \\
  & \le \mathbb{E} \bigg[ \limsup_{ h \to - \infty } {D}_{T} \bigg( w {D}_{T} - w {D}_{T} \wedge \Big( \frac{ \theta + \rho }{\theta } \zeta + \rho c - \frac{h}{\kappa } \Big) \bigg) \bigg]
      = 0,
\end{align*}
one obtains
\begin{equation*}
{w}^{\dag } ( - \infty ) := \lim_{ h \to - \infty } {w}^{\dag } (h) 
                          = \frac{ \mathbb{E} [ {D}_{T} \zeta ] - \rho ( {x}_{0} - \mathbb{E} [ {D}_{T} c ] ) }{ \mathbb{E} [ | {D}_{T} |^{2} ] }.
\end{equation*}
In the same manner, one obtains $\lim_{ h \to + \infty } {w}^{\dag } (h) = + \infty $.
In conjunction with ${h}^{\dag } ( {w}^{\dag } ( - \infty ) ) > - \infty $ and $\lim_{ w \to + \infty } {h}^{\dag } (w) = - \infty $,
we conclude that the fully coupled equations $h = {h}^{\dag } (w)$ and $w = {w}^{\dag } (h)$ admits a unique solution $( {h}^{\S }, {w}^{\S } )$ with ${w}^{\S } \ge {w}^{\dag } ( - \infty )$.

\bibliographystyle{apacite}
\bibliography{SMMV-references}

\end{document}